\documentclass[sigconf]{acmart}
\usepackage{booktabs} 
\makeatletter
\renewcommand\@formatdoi[1]{\ignorespaces}
\makeatother
\usepackage{enumerate}
\usepackage{soul}
\usepackage{framed}
\usepackage{subfigure}
\usepackage[linesnumbered,lined,boxed]{algorithm2e}
\usepackage[toc, page]{appendix}
\newcommand{\Adj}{\text{Adj}}

\newcommand{\D}{\mathsf D}
\newcommand{\R}{\mathsf R}

\newcommand{\Prob}{\mathbb P}
\newcommand{\RN}[1]{%
  \textup{\uppercase\expandafter{\romannumeral#1}}%
}
\usepackage{natbib}
\usepackage{booktabs}
\usepackage{bigints}
\usepackage{hyperref}

\usepackage{adjustbox,multirow}
\newcommand{\subhead}[1]{\vspace {0.04in}\noindent{\textbf{#1.}}}
\newtheorem{exmp}{Example}[section]
\newcommand{\DP}{differential privacy }
\newcommand\norm[1]{\left\lVert#1\right\rVert}
\usepackage{mathtools}
\usepackage[thinc]{esdiff}
\usepackage{courier}
\usepackage{amsthm}

\newtheorem{thm}{Theorem}[section]
\newtheorem{lem}[thm]{Lemma}

\newtheorem{coro}[thm]{Corollary}

\newtheorem{defn}{Definition}[section]
\usepackage{bbm, dsfont}
\setcopyright{none}
\author{Meisam Mohammady{$^\dagger$}, Shangyu Xie{$^\ddagger$}, Yuan Hong$^\ddagger$, Mengyuan Zhang$^\S$}
\author{Lingyu Wang$^\dagger$, Makan Pourzandi$^\S$, and Mourad Debbabi$^\dagger$}

\affiliation{%
  \institution{$^\dagger$Concordia University, $^\ddagger$Illinois Institute of Technology, $^\S$Ericsson Research Security}
}
\email{m_ohamma@encs.concordia.ca, sxie14@hawk.iit.edu, yuan.hong@iit.edu,  mengyuan.zhang@ericsson.com}
\email{wang@ciise.concordia.ca, makan.pourzandi@ericsson.com, debbabi@encs.concordia.ca}

\begin{document}

\fancyhead{}

 \title{
     R\texorpdfstring{$^2$}{2}DP: A Universal and Automated Approach to Optimizing the Randomization Mechanisms of Differential Privacy for Utility Metrics with No Known Optimal Distributions}
\thanks{\small Appears in Proceedings of the 27th ACM Conference on Computer and Communications Security, 2020}

\begin{abstract}
Differential privacy (DP) has emerged as a de facto standard privacy
notion for a wide range of applications. Since the meaning of data
utility in different applications may vastly differ, a key challenge
is to find the optimal randomization mechanism, i.e., the distribution
and its parameters, for a given utility metric. Existing works have
identified the optimal distributions in some special cases, while
leaving all other utility metrics (e.g., usefulness and graph
distance) as open problems.  Since existing works mostly rely on
manual analysis to examine the search space of all distributions, it
would be an expensive process to repeat such efforts for each utility metric. To address such deficiency, we propose a novel
approach that can automatically optimize different utility metrics
found in diverse applications under a common framework. Our key idea that, by regarding the
variance of the injected noise itself as a random variable, a two-fold
distribution may approximately cover the search space of all
distributions. Therefore, we can automatically find distributions in this search
space to optimize different utility metrics in a similar manner,
simply by optimizing the parameters of the two-fold distribution.
Specifically, we define a universal framework, namely,
\textit{r}andomizing the \textit{r}andomization mechanism of
\textit{d}ifferential \textit{p}rivacy (R$^2$DP), and we formally
analyze its privacy and utility. Our experiments show that R$^2$DP can
provide better results than the baseline distribution (Laplace) for several utility metrics with no known optimal distributions, whereas our results asymptotically approach to the optimality for utility metrics
having known optimal distributions. As a side benefit, the added degree
of freedom introduced by the two-fold distribution allows R$^2$DP to
accommodate the preferences of both data owners and recipients.

\end{abstract}




\maketitle
\thispagestyle{empty}
\section{Introduction}
Significant amounts of individual information are being collected and analyzed today through a wide variety of applications across different industries~\cite{aarons2012dynamic}. Differential privacy has been widely recognized as the de facto standard notion~\cite{Dwork08differentialprivacy:,10.1007/11681878_14} in protecting individuals' privacy during such data collection and analysis. On the other hand, since the privacy constraints (e.g., the
degree of randomization) imposed by differential privacy may render
the released data less useful for analysis, the fundamental trade-off
between privacy and utility (i.e., analysis accuracy) has
attracted significant attention in various settings~\cite{10.1007/11681878_14,erlingsson2014rappor,lee2013pibox,NissimRS07,RastogiN10,DworkNRRV09}.

In this context, a key issue is to identify the optimal randomization
mechanisms (i.e., distributions and their
parameters)~\cite{Ghosh1,Gupte1,6875258,DBLP:journals/jstsp/GengKOV15,DBLP:conf/icml/BalleW18,DBLP:journals/corr/abs-1809-10224, Hardt:2010:GDP:1806689.1806786,5670945}).  While optimizing the
parameters of a given distribution can be easily automated,
identifying the optimal distribution for different utility metrics is
more challenging, and typically requires manual analysis to examine
the search space of all distributions. In fact, recent
studies~\cite{Ghosh1,Gupte1,6875258,DBLP:journals/jstsp/GengKOV15,DBLP:conf/icml/BalleW18,DBLP:journals/corr/abs-1809-10224,
  Hardt:2010:GDP:1806689.1806786,5670945} have only identified the
optimal randomization mechanisms for a limited number of cases with
specific utility metrics and queries. For instance, Ghosh et
al.~\cite{Ghosh1,Gupte1}
showed that an optimal randomization mechanism (adding a specific
class of \textit{geometric noise}) can be used to preserve
differential privacy under the class of negative expected loss utility
metrics for a single counting query. Subsequently, Geng et
al.~\cite{6875258} showed that, under the $\ell_1$ and $\ell_2$ norms,
the widely used standard Laplace mechanism is asymptotically optimal as $\epsilon \rightarrow 0$, whereas the Staircase
mechanism (which can be viewed as a geometric mixture of uniform
probability distributions) performs exponentially better than the
Laplace mechanism in case of weaker privacy guarantees (a
comprehensive literature review will be given in
Section~\ref{rel}).

However, this has left the optimal distributions of many other utility
metrics as open problems, e.g., usefulness (for machine learning
applications~\cite{Blum:2008:LTA:1374376.1374464}), entropy-based
measures (for signal processing
applications~\cite{cohen1993majorization,cdcc}, and semi-supervised learning~\cite{grandvalet2005semi}), and graph distance
metrics (for social network
applications~\cite{10.1007/978-3-642-36594-2_26}).  As shown in the
works of Ghosh et
al.~\cite{Ghosh1,Gupte1}
and Geng et al.~\cite{6875258}, different utility metrics will likely
lead to different optimal distributions. Moreover, since those
existing works mostly rely on manual analysis to examine the search
space of all distributions, it would be an expensive process to repeat
such efforts for each utility metric.  Consequently, many existing
works simply employ a well-known distribution (e.g., Laplace noise
with constant scale parameter or Gaussian noise with constant
variance) without worrying about its optimality. Unfortunately, as our
experimental results will show (Section~\ref{exp:sec}), choosing a non-optimal distribution (even with its
parameters optimized) may lead to rather poor utility.

\vspace{-0.1in}
\begin{figure}[ht]
\centering
\includegraphics[width=1\linewidth,viewport=82 2 920 563,clip]{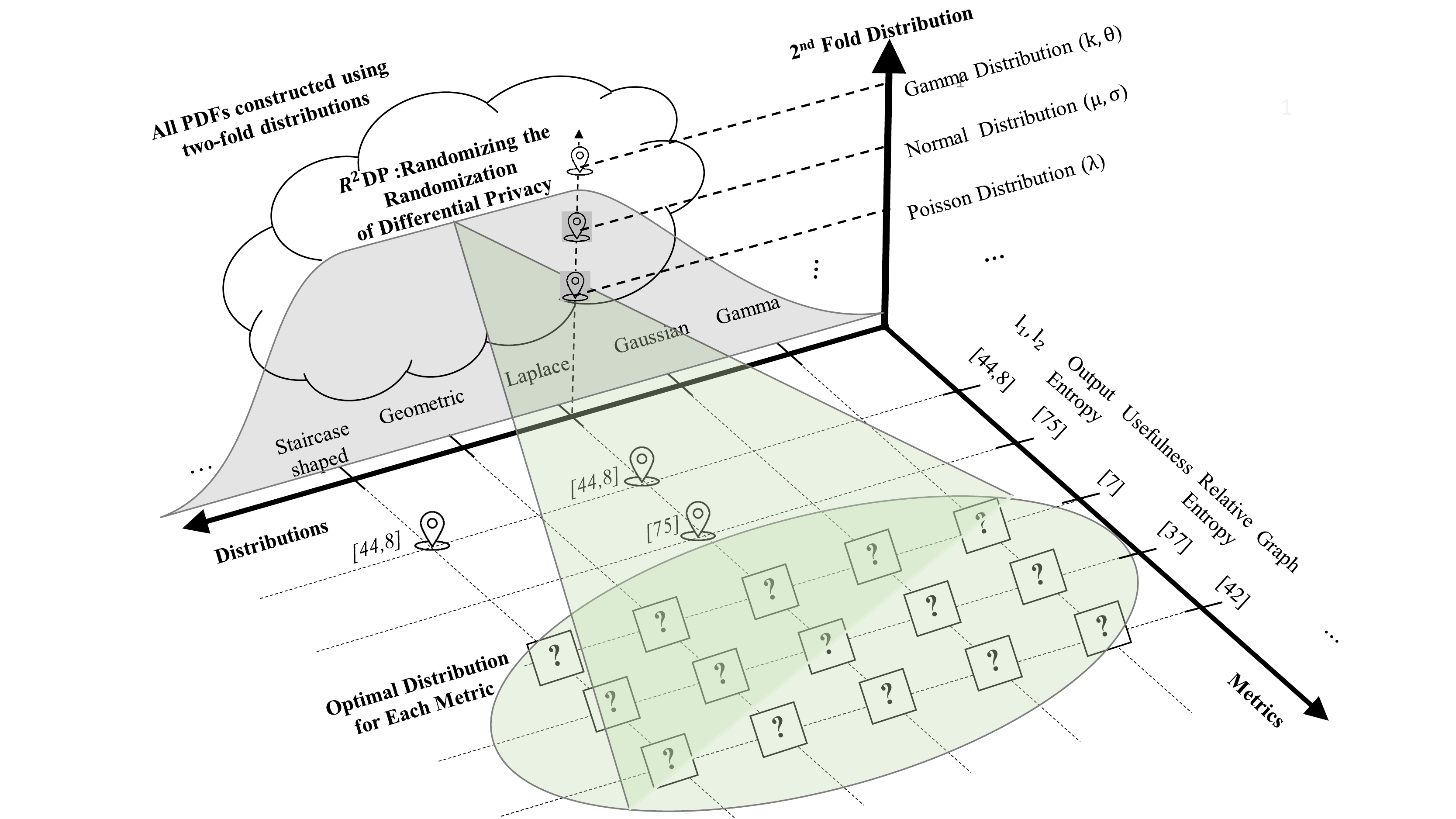}
\vspace{-0.05in}
\caption{R$^2$DP can automatically optimize different utility metrics which have no known optimal distributions.}\vspace{-0.1in}
\label{Fig:2}
\end{figure}

\subsection{R$^2$DP: A Universal Framework}
Our key observation is the following. To build a universal framework
that can automatically find the optimal distribution in the search
space of all distributions, we would need a formulation to link the differential privacy guarantee to the parameters of different distributions (e.g., in Laplace mechanism, $\epsilon$ is proportionally related to the inverse of variance). However, it is a known fact that such a formulation varies for each distribution, which explains why existing works have to rely on manual efforts to cover the search space of all distributions, and it also becomes the main obstacle to finding a universal solution that works for all utility metrics employed in different applications.

As depicted in Figure~\ref{Fig:2}, our key idea is that, although it is not possible to directly cover the search space of all distributions in an
automated fashion, we can indirectly do so based on the following known fact in probability
theory, i.e., a two-fold
randomization over the exponential class of distributions may yield
many other distributions to approximately cover the search
space~\cite{Charalambides:2005:CMD:1196346}. Since this class of distributions are all originated from one
of the exponential family distributions, their \DP guarantee will
become a unique function of the parameters of the second fold distribution. Therefore, these parameters can be used to automatically
optimize utility w.r.t. different utility metrics through a universal
framework, namely, \emph{r}andomizing the \emph{r}andomization
mechanism in \emph{d}ifferential \emph{p}rivacy
(R$^2$DP). Furthermore, the two-fold distribution introduces an added
degree of freedom, which allows R$^2$DP to incorporate the
requirements of both data owners and data recipients.

\subsection{Contributions}

Specifically, we make the following contributions:
\begin{enumerate}
    \item We define the R$^2$DP framework with several unique
      benefits. First, it provides the first universal solution that
      is applicable to different utility metrics, which makes it an
      appealing solution for applications whose utility metrics have
      no known optimal distributions (e.g.,
      ~\cite{Blum:2008:LTA:1374376.1374464,cohen1993majorization,cdcc,10.1007/978-3-642-36594-2_26}). Second,
      unlike most existing works which rely on manual
      analysis~\cite{Ghosh1,Gupte1},
      R$^2$DP can automatically identify a distribution that yields
      near-optimal utility, and hence is more practical for emerging
      applications. Third, R$^2$DP can incorporate the requirements of both data owners and data recipients, which addresses a
      practical limitation of most existing approaches, i.e., only the privacy budget $\epsilon$ is considered in designing the differentially private mechanisms.
      
    \item We formally benchmark R$^2$DP under the well-studied Laplace
      mechanism. We tackle several key challenges related to
      the two-fold distribution in R$^2$DP.
      We then show that this mechanism yields a class of log-convex
      distributions for which the differential privacy guarantee can
      globally be given in terms of the PDFs' parameters. We also show
      that it can generate near-optimal results w.r.t. a
      variety of utility metrics whose optimality is known, e.g., Staircase-shape distribution for large $\epsilon$ and Laplace itself for small $\epsilon$ \cite{6875258}.
\item We evaluate R$^2$DP using six different utility metrics, both
  numerically and experimentally on real data, using both statistical
  queries (e.g., count and average), and data analytics
  applications (e.g., machine learning and social network). The
  experimental results demonstrate that R$^2$DP can significantly
  increase the utility for those utility metrics with no known optimal
  distributions (compared to the baseline Laplace distribution). We
  also evaluate the optimality of R$^2$DP using utility metrics whose
  optimal distributions are known (e.g., Staircase-shape for $\ell_1$ and $\ell_2$ norms \cite{6875258}) and our results confirm that
  R$^2$DP can generate near-optimal results.

\item We discuss the potential of adapting R$^2$DP to improve a variety of other applications related to differential privacy.  
  
\end{enumerate}

The rest of the paper is organized as follows. Section~\ref{sec:model} provides some related background. Section~\ref{sec4} defines the R$^2$DP framework. Section~\ref{sec5} formally studies the \DP guarantee and the utility of R$^2$DP. Section~\ref{exp:sec} presents the experiments. Section~\ref{relat} reviews the related work, and
Section~\ref{sec:conclusion} concludes the paper.

\section{Preliminaries} \label{sec:model}
We review some background on \DP for the theoretical foundations of the R$^2$DP framework. 

\subsection{Differential Privacy}\label{section: DP def}
\label{def: differential privacy original}

We follow the standard definitions of
$\epsilon$-\DP~\cite{DworkNRRV09,NissimRS07}. Let $\D$ be a dataset of
interest and $d$, $d'$ be two adjacent subsets of $\D$ meaning that we
can obtain $d'$ from $d$ simply by adding or subtracting the data of
one individual. A randomization mechanism $\mathcal M: \D \times \Omega
\to \R$ which is $\epsilon$-differentially private, necessarily
randomizes its output in such a way that for all $S \subset \R$, 
\begin{align}\label{eq: standard def approximate DP original}
\Prob(\mathcal M(d) \in S) \leq e^{\epsilon} \Prob(\mathcal M(d') \in S) \;\; 
\end{align}

If the inequality fails, then a leakage ($\epsilon$ breach) takes place, which means the difference between the prior
distribution and posterior one is tangible. We recall below a basic
mechanism that can be used to answer queries in an $\epsilon$-differentially private way. We will only be concerned with queries that return numerical answers, i.e., a
query is a mapping $q: \D \to \mathbb R$, where $\mathbb R$ is a set of real numbers. The following sensitivity concept plays an important role in the design of differentially private
mechanisms~\cite{10.1007/11681878_14}.
\begin{defn}	\label{defn: sensitivity}
The sensitivity of a query $q: \D \to \R$ is defined as
$\Delta q= \max_{d,d':\Adj(d,d')} |q(d) - q(d')|$~\cite{DworkNRRV09,NissimRS07}.
\end{defn}
\label{section: basic mech}

\subsection{Laplace Mechanism} The Laplace mechanism~\cite{10.1007/11681878_14}
    modifies a numerical query result by adding zero-mean noise (denoted as $Lap(b)$) distributed according to a Laplace distribution with mean zero and scale parameter $b$. It has density
    $p(x;b)=\frac{1}{2b}exp(-\frac{|x|}{b})$ and variance $2b^2$.

\begin{thm}	\label{thm: Lap mech}
Let $q: \D \to \mathbb R$ be a query , $\epsilon>0$.
Then the  mechanism $\mathcal M_q: \D \times \Omega \to \mathbb R$ 
defined by $\mathcal M_q(d) = q(d) + w$, with $w \sim Lap(b)$, 
                                                                                                                                                                                                                                                where $b \geq \frac{\Delta q}{\epsilon}$,
is $\epsilon$-differentially private~\cite{10.1007/11681878_14}.
\end{thm}

\subsection{Utility Metrics}
\label{sec:metric}

\subhead{$\ell_p$ Metrics}
In penalized regression, ``$\ell_p$ penalty'' refer to penalizing the $\ell_p$ norm of a solution's vector of parameter values (i.e., the sum of its absolute values, or its Euclidean length)~\cite{10.5555/21713}. In our privacy-utility setting, the $\ell_p$ utility metric is defined as follows.
\begin{defn}($\ell_p$).
\label{defn:lp}
For a database mechanism $\mathcal M_q(D)$ the $\ell_p$ utility metric is defined as $\mathbb E(|\mathcal M_q(D)-q(D)|^p)^{1/p}$.
\end{defn}

\subhead{Usefulness}
Following Blum et al.~\cite{Blum:2008:LTA:1374376.1374464}, the following utility metric is commonly used for machine learning. 
\begin{defn}(Usefulness).
\label{defn:useful}
A mechanism $\mathcal M_q$ is ($\gamma, \zeta$)-useful if, with probability
$1 - \zeta$, for any dataset $d \subseteq \D$, $|\mathcal M_q(d)- q(d)| \leq \gamma$.
\end{defn}
\begin{thm}
\label{thmuselap}
The Laplace Mechanism is $(\frac{\Delta q}{\epsilon} \ln \frac{1}{\zeta},\zeta)$-useful, or equivalently, the Laplace Mechanism is $(\gamma,e^{\frac{-\gamma}{b(\epsilon)}})$-useful~\cite{Chan:2011:PCR:2043621.2043626}.
\end{thm}

\subhead{Mallows Metric} The Mallows metric has been applied for evaluating the private estimation of the degree distribution of
a social network~\cite{Hay:2009:AED:1674659.1677046}. It is defined to test if two samples are drawn from the same distribution. Given two random variables $X$ and $Y$, we have $Mallows(X,Y)=\frac{1}{n}\sum_{i=1}^{n} (|X_i-Y_i|^p)^{1/p}$ (similar to $p$-norm).

\subhead{Relative Entropy (R\'enyi Entropy)} The relative entropy, also known as the \textit{Kullback-Leibler (KL)} divergence, measures the distance between two probability distributions ~\cite{cohen1993majorization}. Formally, given two probability distributions $p(x)$ and $q(x)$ over a discrete random variable $x$, the relative entropy given by $D(p||q)$ is defined as follows: $D(p||q) = \sum_{x\in\mathcal{X}} p(x) \log \frac{p(x)}{q(x)}
$. Further generalization came from R\'enyi~\cite{renyi1961measures,gil2011renyi}, who introduced an indexed family of generalized information and divergence measures akin to the Shannon entropy and KL divergence. R\'enyi introduced the entropy of order $\alpha$ as $I_{\alpha}(p||q) = \frac{1}{\alpha-1} \log (\sum_{x\in\mathcal{X}}   p(x)^\alpha q(x)^{1-\alpha})$ , $\alpha>0$ and $\alpha \neq 1$.
\section{The R$^2$DP framework}
\label{sec4}

\begin{figure*}[ht]
\includegraphics[width=0.95\linewidth]{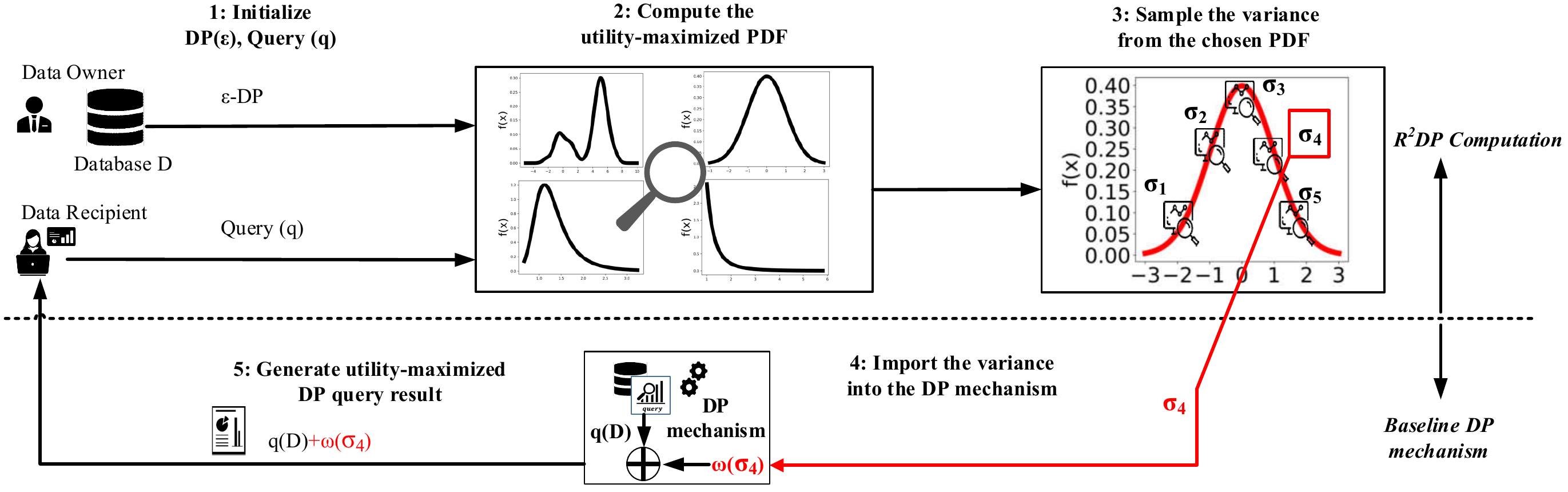}
\centering
\vspace{-0.15in}
\caption{The high level overview of
the R$^2$DP framework.}
\label{figoverview}
\end{figure*}

In this section, we define the R$^2$DP framework and its main building block which is the \textit{Utility-maximized PDF} finder. 
\subsection{Notions and Notations}
\label{subsec4}
In probability and statistics, a random variable (RV) that is distributed according to some parameterized PDFs, with (some of) the
parameters of that PDFs themselves being random variables, is known as a \emph{mixture} distribution~\cite{Charalambides:2005:CMD:1196346}
when the underlying RV is discrete (or a \emph{compound} distribution when the RV is continuous). Compound (or mixture) distributions have been applied in many contexts in the literature~\cite{panjer_1981} and arise naturally where a statistical population contains two or more sub-populations. 
\begin{defn}
 Let ($\Omega,\mathcal F, \Prob$) be a probability space and let $X$ be a RV that is distributed according to some parameterized distribution $f(\theta) \in\mathcal F$ with an unknown parameter $\theta$ that is again distributed according to some other distribution $g$. The resulting distribution $h$ is said to be the distribution that results from compounding $f$ with $g$, 
  \begin{equation}
 h(X)= \int_{\mathbb R} f(X|\theta) g(\theta)\operatorname d\theta 
 \end{equation}
Then for any Borel subset $B$ of $\mathbb R$,
 \begin{equation}
 \label{probcomp}
   \Prob(X \in B)= \int_{B} \int_{\mathbb{R}} f(X|\theta) g(\theta)\operatorname d\theta dX
 \end{equation}
\end{defn}

In general, we call any differentially private query answering mechanisms that leverage two-fold probability distribution functions in their randomization, an \emph{R$^2$DP mechanism}.

\begin{defn}(\textit{R$^2$DP Mechanism}).
\label{Def:RPm}
Let $\mathcal M_q(d,u) = q(d) \bigoplus \omega(u)$ be a mechanism
randomizing the answer of a query $q$ using a random oracle
$\omega(u)$, where $u$ is the set of parameters (mean, variance, etc.) of the PDF of
$\omega$ and $\bigoplus$ stands for the corresponding operator. Denote by $\mathcal F$ the space of PDFs, we call $\mathcal M_q(d,u)$ an \textit{R$^2$DP mechanism} if at least one of the parameters
$u_i\in u$, $( i\leq |u|)$ is/are chosen randomly w.r.t. a
specified probability distribution $f_{u_i} \in \mathcal F$.

\end{defn}
    
\vspace{0.05in}
In particular, the \emph{R{$^2$}DP Laplace mechanism} will modify the answer to a numerical query by adding zero-mean noise distributed according to a
compound Laplace distribution with the scale parameter $b$ itself
distributed according to some distribution $f_b$.
\begin{exmp}
\label{exmprpdp}
 Suppose that the scale parameter $b$ in a Laplace mechanism is randomized as follows:
 \[ b=\begin{cases}b_1&{\text{w.p.}} \ \ \ p     ,\\b_2&{\text{w.p.}} \ \ \ 1-p   .\end{cases}\]

 Then, the perturbed result $q(D)+Lap(b)$ is an example R$^2$DP Laplace mechanism using a Bernoulli distribution.
 \end{exmp}

\begin{defn}
Let $q: \D \to \mathbb R$ be a query and suppose $f_{b} \in \mathcal F$ is a probability density function of the scale parameter $b$. Then, the  mechanism $\mathcal M_q: \D \times \Omega \to \mathbb R$, defined by $\mathcal M_q(d,b) = q(d) + Lap(b)$ is an R$^2$DP Laplace mechanism that utilizes PDF $f_{b}$.
\end{defn}

\subsection{The Framework}
 \label{subsec:Overview}
 As shown in Figure~\ref{figoverview}, R$^2$DP framework include the following steps. 

\vspace{0.05in}

\noindent\textbf{R$^2$DP Computation:}
	\begin{itemize}
	    \item \textbf{Step 1}: The data owner specifies the \DP budget $\epsilon$ and the data recipient specifies his/her query of
          interest together with its required utility metric.
        
        \item \textbf{Step 2}: Given the input triplets ($\epsilon,\text{query},\text{metric}$), the
          \emph{utility-maximized PDF computing module} computes the provably optimal probability density function and its parameters for the variance of the
          additive noise. For example, in Figure~\ref{figoverview},
          the PDF computing module returns a lower tail truncated Gaussian
          distribution for the specified inputs.  
          
          \item \textbf{Step 3}: The \emph{variance sampler} module randomly samples
     (w.r.t. the PDF found in Step 2) one standard
     deviation $\sigma_i$ of the noise to be eventually added.
	\end{itemize}

\noindent\textbf{Baseline DP Randomization:}     
    
    \begin{itemize}
        \item \textbf{Step 4}: Next, the computed standard deviation $\sigma_i$ is used to generate a noise $\omega(\sigma_4)$ for the baseline DP mechanism, which is a DP mechanism of exponential order, e.g ., Laplace, Gaussian and exponential mechanisms. 
        \item \textbf{Step 5}: The computed noise $\omega(\sigma_i)$ is added to the query result $q(D)$ to provide a utility-maximized DP result to the data recipient. 
    \end{itemize}
    The most important module of the R$^2$DP framework is the \emph{utility-maximized PDF computing module} (Step 2) which will be described in more details in the following. Furthermore, to make our discussions more concrete, we instantiate the R$^2$DP framework based on the well studied Laplace mechanism, namely, \textit{the R$^2$DP Laplace mechanism}, where other baseline DP mechanisms will be discussed in Appendix~\ref{sec:Gauss} due to space limitation (from now on, we will simply refer to the R$^2$DP Laplace mechanism as R$^2$DP).
    Particularly, we show that, with a two-fold Laplace distribution, an infinite-size class of log-convex distributions can be identified. This class of distributions pertains a differential privacy guarantee which can globally be given in terms of the PDFs' parameters, and hence is automatically optimizable under the \DP constraint.
    
\vspace{0.05in}

\subsection{Computing Utility-Maximized PDF}
\label{secsearc}
In Figure~\ref{figoverview}, to compute the utility-maximized PDF (Step 2), a key challenge is to establish the search space of automatically optimizable PDFs, from which the utility-maximized PDF is computed. Ideally, the search space of an R$^2$DP mechanism can be defined as the collection of all two-fold distributions, e.g., with Laplace and exponential as the first and second fold distributions, respectively. However, the key challenge here is that a mixture of distributions is itself a distribution which does not necessarily provide a global \DP guarantee in terms of the resulting PDFs' parameters (automatically optimizable under the differential privacy constraint). To address this issue, the \emph{Moment Generating Function (MGF)}~\cite{fisz2018probability} of the second fold distribution could be utilized, e.g., given the first fold as Laplace distribution.  
Specifically, MGF of a random variable is an alternative specification of its
probability distribution, and hence provides the basis of an alternative route to analytical results compared with directly using probability density functions or cumulative distribution
functions~\cite{fisz2018probability}. In particular, the MGF of a random variable is a log-convex function of its probability distribution which can provide a global \DP guarantee~\cite{fisz2018probability} (see Theorem~\ref{simple DP}). 
\begin{defn}(\textit{Moment Generating Function}~\cite{fisz2018probability}). The moment-generating function of a random variable $x$ is $M_{X}(t):=\mathbb E \left[e^{tX}\right], t\in \mathbb {R}$ wherever this expectation exists. The moment-generating function is the expectation of the random variable $e^{tX}$.
\end{defn}
\begin{thm} 	\label{thm: RPLap mech}
We can write the CDF of the output of an R$^2$DP mechanism in terms of the \emph{Moment Generating Function (MGF)}~\cite{fisz2018probability} of the probability distribution $f_{\frac{1}{b}}$, where $b$ is the randomized scale parameter (see Appendix~\ref{demonst} and~\ref{thm3.1} for the details and the proof).

\end{thm}

Thus, for a PDF with non-negative support (since scale parameter is always non-negative), the R$^2$DP mechanism outputs another
PDF using the MGF (where CDF is the moment and PDF is its derivative, as shown in Equation~\ref{eqn11} in Appendix \ref{proofsec}) . Moreover, since MGF is a bijective function~\cite{feller2008introduction}, the R$^2$DP mechanism can in fact generate a search space as large as the space of all PDFs with non-negative support and an existing MGF.  However, the next challenge is that not all random variables have moment generating functions (MGFs), e.g., Cauchy distribution~\cite{bulmer1979principles}. Fortunately, MGFs possess an appealing composability property between
independent probability distributions~\cite{Charalambides:2005:CMD:1196346}, which can be used to provide a search space of all linear combinations of a set of popular distributions with known MGFs (infinite number of RVs). 
\begin{thm}[MGF of Linear Combination of RVs]
\label{thm:lin}
 If $x_1, \cdots, x_n$ are $n$ independent RVs with MGFs $M_{x_i}(t)=\mathbb E (e^{t x_i})$ for $i = 1,\cdots , n$, then the MGF of the linear combination $Y=\sum\limits_{i=1}^{n}a_ix_i$ is $\prod \limits_{i=1}^{n} M_{x_i} (a_it)$.
\end{thm}
Consequently, we define the search space of the R$^2$DP mechanism as all possible linear combinations of a set of independent RVs with existing MGF (Section~\ref{pdffind} will provide more details on how to choose the set of independent RVs). Although this search space is only a subset of all two-fold distributions, we will show through both numerical results (in Section~\ref{numericsec}) and experiments with real data (Section~\ref{exp:sec}) that this search space is indeed sufficient to generate near-optimal utility w.r.t. all  utility metrics (universality).

\section{Privacy and Utility}
\label{sec5}
In this section, we analyze the privacy and utility of the R$^2$DP, and then discuss extensions for improving and implementing R$^2$DP. 

\subsection{Privacy Analysis}
\label{subsec:privacy}

 We now show the R$^2$DP mechanism provides differential privacy guarantee. By Theorem~\ref{thm: RPLap mech}, the DP bound of the R$^2$DP is 
 
 \vspace{-0.1in}

\begin{eqnarray*}
\label{DPlapexmp1}
 &\hspace{-.4cm} e^{\epsilon}= \max\limits_{\forall S \in \R} \left\{\frac{-M_{\frac{1}{b}}(-|x-q(d)|)|_{S_{\geq q(d)}}+M_{\frac{1}{b}}(-|x-q(d)|)|_{S_{< q(d)}}}{-M_{\frac{1}{b}}(-|x-q(d')|)|_{S_{\geq q(d')}}+M_{\frac{1}{b}}(-|x-q(d')|)|_{S_{< q(d')}}}\right\} 
\end{eqnarray*}
 
 Hence, the value of $e^{\epsilon}$ only depends on the distribution of reciprocal of the scale parameter $b$, i.e., $f_{\frac{1}{b}}$. Moreover, an MGF is positive and log-convex~\cite{fisz2018probability} where the latter property is desirable in defining various natural logarithm upper bounds, e.g., DP bound. In the following theorem, our MGF-based formula for the probability $\Prob(\{q(d)+Lap(b)\}\in S)$ can be easily applied to calculate the \DP guarantee (see Appendix~\ref{proofsec} for the proof).

\begin{thm}
\label{simple DP}
The R$^2$DP mechanism $\mathcal M_q(d,b)$ is 

\vspace{-0.15in}

\begin{equation}
\label{simple DPeq}
    \ln \left[ \cfrac{ \mathbb E(\frac{1}{b})} {\diff{M_{\frac{1}{b}}(t)}{t}|_{t=- \Delta q}} \right]\text{-differentially private.}
\end{equation}
\end{thm}
 
Moreover, Theorem~\ref{thm:lin} can be directly applied to calculate the \DP guarantee of any RV from the search space defined in Section~\ref{secsearc} (i.e., all linear combinations of a set of independent RVs with known MGFs).

\begin{coro}
[Differential Privacy of Combined PDFs]
\label{thm:mgffin}
If $x_1, \cdots, x_n$ are $n$ independent random variables with respective MGFs $M_{x_i}(t)=\mathbb E (e^{t x_i})$ for $i = 1, \cdots, n$, then the R$^2$DP mechanism $\mathcal M_q(d,b)$ where $\frac{1}{b}$ is defined as the linear combination $\frac{1}{b}=\sum\limits_{i=1}^{n}a_ix_i$ is $\epsilon$-differentially private, where
 \vspace{-0.2cm}
 \begin{eqnarray}
 \scriptsize
    \label{mgfdp}
  \epsilon=\ln\left[\cfrac{\sum \limits_{j=1}^{n} a_j\cdot E_{x_j}(\frac{1}{b})}{\sum \limits_{j=1}^{n} a_j\cdot M'_{x_j}(-a_j\cdot \Delta q) \cdot \prod \limits_{\substack{i=1 \\ i\neq j}}^{n} M_{x_i} (-a_i\cdot \Delta q)}\right]
    \end{eqnarray}

\end{coro}

Therefore, we have established a search space of probability distributions with a universal formulation for their \DP guarantees, which is the key enabler for the universality of R$^2$DP. Next, we characterize the utility of R$^2$DP mechanisms.

\subsection{Utility Analysis}

We now characterize the utility of the R$^2$DP mechanism. To make
concrete discussions, we focus on the usefulness metric (see
Section~\ref{sec:metric}), and a similar logic can also be applied to
other metrics.

\subsubsection{Characterizing the Utility}

Denote by $U(\epsilon, \Delta q, \gamma)$ the usefulness of an R$^2$DP mechanism for all $ \epsilon>0$, sensitivity $\Delta q$ and error bound $\gamma$. The optimal usefulness is then given as the answer of the following optimization problem over the search space of PDFs. 
\vspace{-0.1cm}
\begin{eqnarray}
\label{multi-obj1}
&\hspace{-0.5cm} \max\limits_{f_{\frac{1}{b}}\in F} \big\{U(\epsilon, \Delta q, \gamma) \big\}=\max\limits_{f_{\frac{1}{b}}\in F} \bigg \{\frac{1}{2} \cdot \Big[-M_{\frac{1}{b}}(-|x-q(d)|)|_{q(d)}^{q(d)+\gamma}\nonumber\\
&\hspace{3.5cm}+M_{\frac{1}{b}}(-|x-q(d)|)|_{q(d)-\gamma}^{q(d)}\Big] \bigg \}, \nonumber\\
& \text{subject to     } \ \ \ \epsilon=\ln \left[ \cfrac{ \mathbb E(\frac{1}{b})} {\diff{M_{\frac{1}{b}}(t)}{t}|_{t=- \Delta q}} \right]  \nonumber
\end{eqnarray}

where the utility function is the probability of generating $\epsilon$-DP query results within a distance of $\gamma$-error (using Theorem~\ref{thm: RPLap mech}). Note that $\epsilon$ and $\Delta q$ do not directly impact the usefulness but they do so indirectly through the \DP constraint. Furthermore, as shown in Theorem~\ref{simple DP}, the \DP guarantee $\epsilon$ over the established search space is a unique function of the parameters of the second-fold distribution. 

\begin{coro}
\label{co1}
Denote by $u$, the set of parameters for a probability distribution $f_{\frac{1}{b}}$, and by $M_{f(u)}$ its MGF. Then, the optimal usefulness of an R$^2$DP mechanism utilizing $f_{\frac{1}{b}}$, at each triplet $(\epsilon, \Delta q, \gamma)$ is
\begin{eqnarray}
\small
    \label{gen:conddelt}
       &\hspace{-0.5cm}  U_f(\epsilon, \Delta q, \gamma)=\max\limits_{u\in \mathbb{R}^{|u|}} \bigg \{\frac{1}{2} \cdot \Big[-M_{f(u)}(-|x-q(d)|)|_{q(d)}^{q(d)+\gamma}\nonumber\\
&\hspace{3.5cm}+M_{f(u)}(-|x-q(d)|)|_{q(d)-\gamma}^{q(d)}\Big] \bigg \}, \nonumber\\
& \text{subject to     } \ \ \ \epsilon=\ln \left[ \cfrac{ \mathbb E(\frac{1}{b})} {\diff{M_{\frac{1}{b}}(t)}{t}|_{t=- \Delta q}} \right]  \nonumber
    \end{eqnarray}
\end{coro}

Since MGFs are positive and log-convex, with $M(0)=1$, we have $U_f(\epsilon, \Delta q, \gamma)=1-\min\limits_{u\in \mathbb{R}^{|u|}} M_{f(u)}(-\gamma)$. Thus, for usefulness metric, the optimal distribution for $\epsilon$ is the one with the minimum MGF evaluated at $\gamma$. In particular, for a set of privacy/utility parameters, we can find the optimal PDF using the \textit{Lagrange multiplier}~\cite{bertsekas2014constrained}. i.e.,

\begin{eqnarray}
    \label{lagrange1}
 \mathcal{L} (u,\lambda)= M_{f(u)}(-\gamma)+ \lambda \cdot (\ln \left[ \cfrac{ \mathbb E(\frac{1}{b})} {\diff{M_{\frac{1}{b}}(t)}{t}|_{t=- \Delta q}} \right]-\epsilon) 
    \end{eqnarray}
Moreover, Theorem~\ref{thm:lin} can be directly applied to design a utility-maximizing R$^2$DP mechanism with a sufficiently large search space
(with an infinite number of different random variables).
\begin{coro}
[Optimal Utility for Combined RVs]
\label{thm:mgff}
 If $ \ x_1, x_2,$ $ \cdots, x_n$ are $n$ independent random variables with respective MGFs $M_{x_i}(t)=\mathbb E (e^{t x_i})$ for $i = 1, 2,\cdots, n$, then for the linear combination $Y=\sum\limits_{i=1}^{n}a_ix_i$, 
the optimal usefulness (similar relation holds for other metrics) under $\epsilon$-\DP constraint is given as 
 \begin{eqnarray}
 \scriptsize
    \label{mgffin}
       &\hspace{-1.5cm}  U_{Y}(\epsilon, \Delta q, \gamma)=1-\min\limits_{\mathcal{A,U}} \bigg\{ \prod \limits_{i=1}^{n} M_{x_i} (-a_i\gamma) \bigg\}\\
&\hspace{-7cm} \text{subject to           }  \nonumber\\ 
& \hspace{-.3cm} \epsilon= \ln\left[\cfrac{\sum \limits_{j=1}^{n} a_j\cdot E_{x_j}(\frac{1}{b})}{\sum \limits_{j=1}^{n} a_j\cdot M'_{x_j}(a_j\cdot -\Delta q) \cdot \prod \limits_{\substack{i=1 \\ i\neq j}}^{n} M_{x_i} (-a_i\cdot \Delta q)}\right] \nonumber  \nonumber
    \end{eqnarray}
    where $\mathcal{A}=\{a_1,a_2,\cdots, a_n\}$ is the set of the coefficients and $\mathcal{U}=\{u_1,u_2,\cdots, u_n\}$ is the set of parameters of the probability distributions of RVs $x_i, \ \forall i \leq n$. 
\end{coro}

Similar to the case of a single RV, we can compute the optimal solution for this
optimization problem using the Lagrange multiplier function in
Equation~\ref{lagrange1}.
   
\subsubsection{Finding Utility-Maximizing Distributions} 
\label{pdffind}

Since not all second-fold probability distributions can boost the
utility of the baseline Laplace mechanism, leveraging all RVs into our
search space would only result in redundant computation by the
utility-maximized PDF computing module. Accordingly, in this section,
we first derive a necessary condition on the \DP guarantee of R$^2$DP
to boost the utility of the baseline Laplace mechanism (refer to
Appendix~\ref{proofsec} for the proof). Using this necessary
condition, we can easily filter out those probability distributions
that cannot deliver any utility improvement.

\begin{thm}
\label{thm: RPLap mechut}
The utility of R$^2$DP with $\epsilon\geq \ln \Big[ \mathbb E_{\frac{1}{b}} \big(e^{\epsilon(b)} \big ) \Big]$ is always upper bounded by the utility of the $\epsilon$-differentially private baseline Laplace mechanism. Equivalently, for an R$^2$DP mechanism to boost the utility, the following relation is necessarily true.
\begin{equation}
\label{necc}
e^\epsilon=\frac{\mathbb E(\frac{1}{b})}{M'_{\frac{1}{b}}(- \Delta q)} < M_{\frac{1}{b}} (\Delta q)
\end{equation}
\end{thm}

We note that $\epsilon= \ln \Big[ \mathbb E_{\frac{1}{b}} \big(e^{\epsilon(b)} \big ) \Big]$ provides a tight upper bound since it gives the overall $e^\epsilon$ of an R$^2$DP mechanism as the average of \DP leakage. Next, we examine a set of well-known
PDFs as second-fold distribution to identify the distribution that offers a significantly improved utility compared with the bound given in Theorem~\ref{thm: RPLap mechut}. Promisingly, our analytic evaluations for \emph{three} of these distributions, i.e., Gamma, uniform and
truncated Gaussian distributions demonstrate such a payoff (Appendix~\ref{cases} theoretically analyzes several case study PDFs). We note that those chosen distributions are general enough to cover many of other probability distributions (e.g., Exponential, Erlang, and Chi-squared distributions are special cases of Gamma distribution).

\subsubsection{Deriving Error Bounds}
The error bounds of the R$^2$DP mechanism under some well-known
utility metrics are shown in Table~\ref{tablemetrics}. The key idea in deriving these results is to calculate the mean of each utility
metric over the PDF of RV $1/b$ (which is the linear combination of
RVs in multiple PDFs). Specifically, given the error bound $e_L(b)$
for deterministic variance (i.e., Laplace mechanism), the total error
bound of an R$^2$DP mechanism will be the mean $\int_0^\infty e_L(b)
f_b(b) db$.  The results shown in Table~\ref{tablemetrics} can be
easily applied to optimize those metrics in corresponding applications
(e.g., \textbf{$\ell_1$} for private record
matching~\cite{inan2010private}, \textbf{ $\ell_2$} for location
privacy~\cite{bordenabe2014optimal}, usefulness for machine
learning~\cite{Blum:2008:LTA:1374376.1374464}, Mallows for social
network analysis~\cite{Hay:2009:AED:1674659.1677046}, and relative
entropy (with a degree $\alpha)$ for semi-supervised
learning~\cite{grandvalet2005semi}). 

\begin{table}[ht]
\caption{Error bound of R$^2$DP under different metrics}
\centering
\begin{adjustbox}{width=0.5\textwidth,center}

\begin{tabular}{|c|c|c|}
\hline
\bf{Metric} & \bf{Dependency to Prior}& \textbf{R$^2$DP Error Bound}
\\
\hline
\bf{$\ell_1$} & independent& $\int\limits_{0}^\infty M_{\frac{1}{b}}(-x) dx$ 
\\
\hline
\bf{ $\ell_2$}& independent & $\sqrt{2\iint\limits_0^\infty  M_{\frac{1}{b}}(-u) du dx}$
\\
 \hline 
 
Usefulness & independent & $1-M_{\frac{1}{b}}(-\gamma)$\\
\hline
Mallows (p) & dependent & $\left([\sum_{i=1}^n |N_i\sim[-M'_{\frac{1}{b}}(-x)/2]|^p]/n\right)^{1/p}  $\\

\hline
Relative Entropy ($\alpha)$ & dependent & $\frac{\log \sum_{x\in\mathcal{X}}^n p(x)^\alpha q(x)^{1-\alpha}}{\alpha-1} \text{s.t} (q(x)-p(x))\sim -M'_{\frac{1}{b}}(-x)/2$\\
\hline
\end{tabular}
\end{adjustbox}
	\label{tablemetrics}
\end{table}

In this context, the $\ell_1$, $\ell_2$ and usefulness metrics (as
defined in Section \ref{sec:metric}) are independent to the prior
(i.e., not depending on the distribution of the true results). The
metrics will be evaluated based on the deviation between the true and
noisy results (which does not change regardless of the prior). On the
contrary, some other metrics (e.g., Mallows and relative entropy)
depend on the prior distribution of the true results
\cite{Hay:2009:AED:1674659.1677046,grandvalet2005semi}. In such cases,
the metrics will be evaluated based on the deviation between the true
and noisy results w.r.t. the prior in specific experimental settings
(we will discuss those specific priors used in the experiments in
Section~\ref{exp:sec}).

\begin{table*}[ht]
\caption{R$^2$DP compared to Laplace w.r.t. error bounds for learning algorithms}
\centering
\begin{adjustbox}{width=0.9\textwidth,center}
\begin{tabular}{|c|c|c|c|c|}
\hline
&Linear SVM~\cite{ji2014differential} & Bayesian Inference (statistician)~\cite{zhang2016differential} &  Robust Linear Regression~\cite{dwork2009differential} & Naive Bayes~\cite{10.1109/WI-IAT.2013.80} \\
\hline
Laplace & $O(
\frac{\log(1/\beta)}{\alpha^2}+ \frac{1}
{\epsilon\alpha} +
\frac{\log(1/\beta)}
{\alpha \epsilon})$&  $O(mn \log(n)) [1-exp(-\frac{n\epsilon}{2|\mathcal{I}|}]$ & $O(n^{-\epsilon log n})$&$O(\frac{1}{n\epsilon})$
\\
\hline

R$^2$DP & $O(
\frac{\log(1/\beta)}{\alpha^2}+ \mathbb{E}_{\frac{1}{b}}(\frac{b}
{\alpha} +
\frac{b\log(1/\beta)}
{\alpha}))$&  $O(mn \log(n)) [1-M_{\frac{1}{b}}(-\frac{n}{2|\mathcal{I}|})]$ & $O(\mathbb{E}_{\frac{1}{b}}(n^{-\frac{log n}{b}}))$&$O(\mathbb{E}_{\frac{1}{b}}({\frac{b}{n}}))$
\\
\hline
\end{tabular}
\end{adjustbox}
	\label{tableerror}
\end{table*}

In addition to the error bounds given in Table~\ref{tablemetrics}, an
analyst can derive error bounds for more advanced queries,
e.g., those pertaining to learning algorithms
\cite{ji2014differential,zhang2016differential,dwork2009differential,10.1109/WI-IAT.2013.80}. Given
the error bound of Laplace mechanism in an application (e.g., Linear
SVM~\cite{ji2014differential}), the error bound of the R$^2$DP
framework for this application can be derived by taking average of the
Laplace's result over the PDF of $\frac{1}{b_r}$. In particular,
Table~\ref{tableerror} demonstrates the error bounds of R$^2$DP for
some learning algorithms (as shown in Section~\ref{exp:sec}, those
learning algorithms can benefit from integrating R$^2$DP instead of
Laplace).

To derive the error bounds shown in Table~\ref{tablemetrics} and
Table~\ref{tableerror}, the noise parameter(s) and the PDFs used in
R$^2$DP can be released to a downstream analyst. This will not cause
any privacy leakage because, similar to other differential privacy mechanisms, the
privacy protection of R2DP comes from the (first-fold) randomization
(whose generated random noises are never disclosed), which will not be
affected even if all the noise parameter(s) and the PDFs are disclosed
(see Section~\ref{subsec:privacy} and Appendix~\ref{proofsec} for the
formal privacy analysis and proof). We note that, although R$^2$DP
replaces the fixed variance of a standard differential privacy mechanism with a random
variance, this second-fold randomization is not meant to keep the
generated parameters (e.g., the variance) secret, but designed to
cover a larger search space (as detailed in Section~\ref{secsearc}).

\subsection{R$^2$DP Algorithm}

Algorithm~\ref{alg:analyst-actions1} details an instance of the R$^2$DP framework using linear combination of three different PDFs. In particular, the algorithm with $\epsilon$-DP finds the best second-fold distribution using the Lagrange multiplier function (see
Appendix~\ref{sec:lag}) that optimizes the utility metric. Then, it randomly generates the noise using the two-fold distribution (e.g., first-fold Laplace) and injects it into the query. 

\begin{algorithm}[!h]
 \small
 \SetKwInOut{Input}{Input}
 \SetKwInOut{Output}{Output}
\Input{Dataset $D$, Privacy budget $\epsilon$, Query $q(\cdot)$, Metric and its parameters (from data recipient)}

\Output{Query result $q(D)+Lap(b_r)$, DP guarantee $\epsilon$, Second-fold PDF's parameters}

$\Delta q\leftarrow$ Sensitivity ($q(\cdot)$) 

Find optimal parameters from Lagrange Multiplier $\mathcal{L}(\epsilon, \Delta q,\text{metric})=$ $a_{1}^{opt},a_{2}^{opt}, a_{3}^{opt}, k^{opt}, \theta^{opt},a^{opt}_u,b^{opt}_u,\mu^{opt},\sigma^{opt},a_{\mathcal N^T}^{opt}, b_{\mathcal N^T}^{opt}$ 

$X_1 \sim \Gamma(k^{opt}, \theta^{opt})$\\
 $X_2 \sim U(a^{opt}_u,b^{opt}_u)$	\\
 $X_3\sim \mathcal{N}^T(\mu^{opt},\sigma^{opt},a_{\mathcal N^T}^{opt} , b_{\mathcal N^T}^{opt})$

 $\frac{1}{b_r} = a_1^{opt} \cdot X_1 + a_2^{opt} \cdot X_2+a_3^{opt} \cdot X_3$

\textbf{return} $q(D)+Lap(b_r)$, $\epsilon$, $\mathcal{L}(\epsilon, \Delta q,\text{metric})$ 
 	\caption{The Ensemble R$^2$DP Algorithm}
 	\label{alg:analyst-actions1}
 \end{algorithm} 

Some advanced applications (e.g., workload queries) that integrate R$^2$DP to improve their utility are discussed in Appendix \ref{sec:RPDP as a Patch to Existing Work1}.

\section{Experimental Evaluations}
\label{exp:sec}
\begin{figure*}[!tb]
\centering
	\subfigure[$\Delta q=0.5, \gamma=0.1$]{
		\includegraphics[angle=0, width=0.25\linewidth]{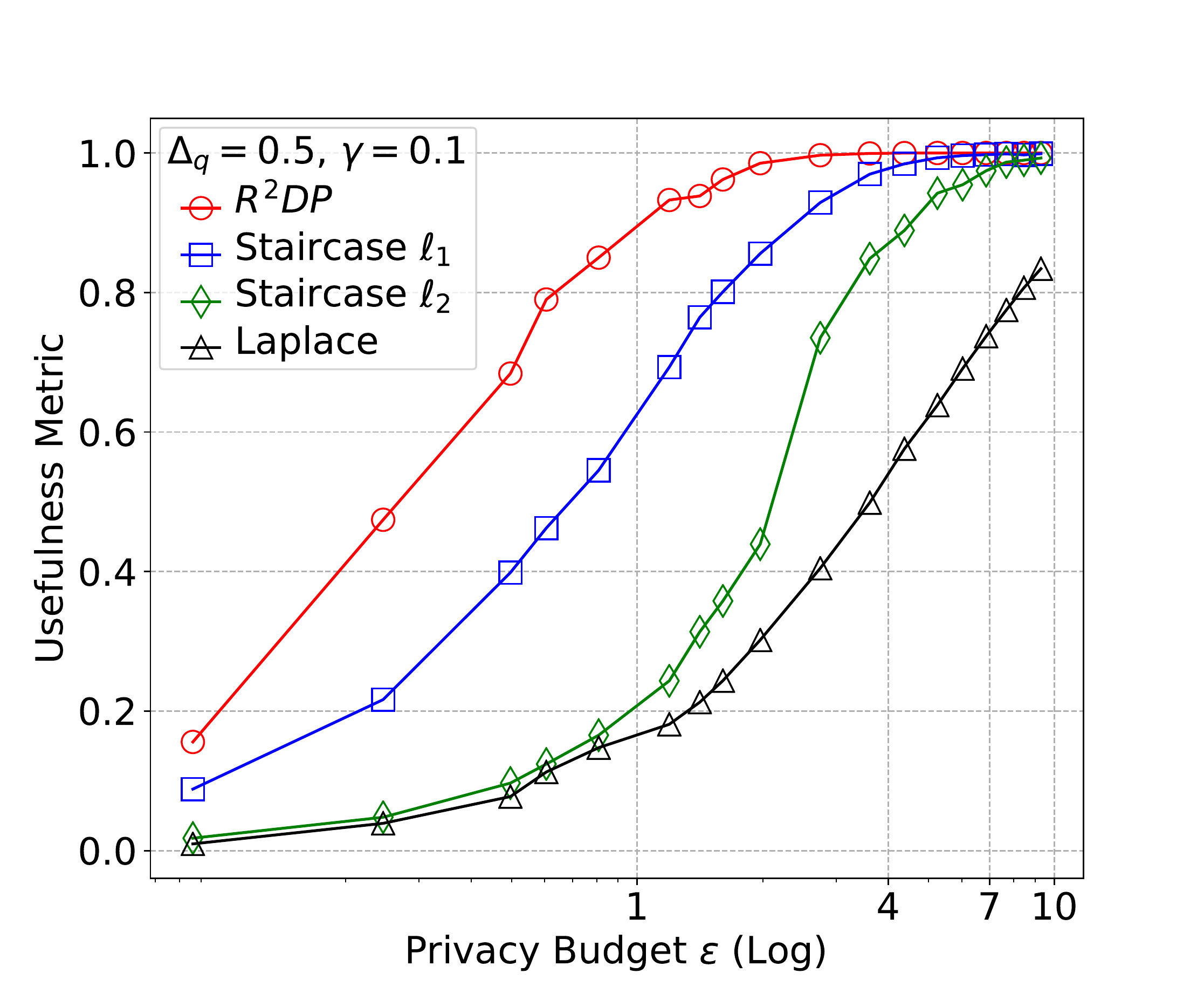}
		 }\hspace{-0.25in}
	\subfigure[$\Delta q=0.5, \gamma=0.4$]{
		\includegraphics[angle=0, width=0.25\linewidth]{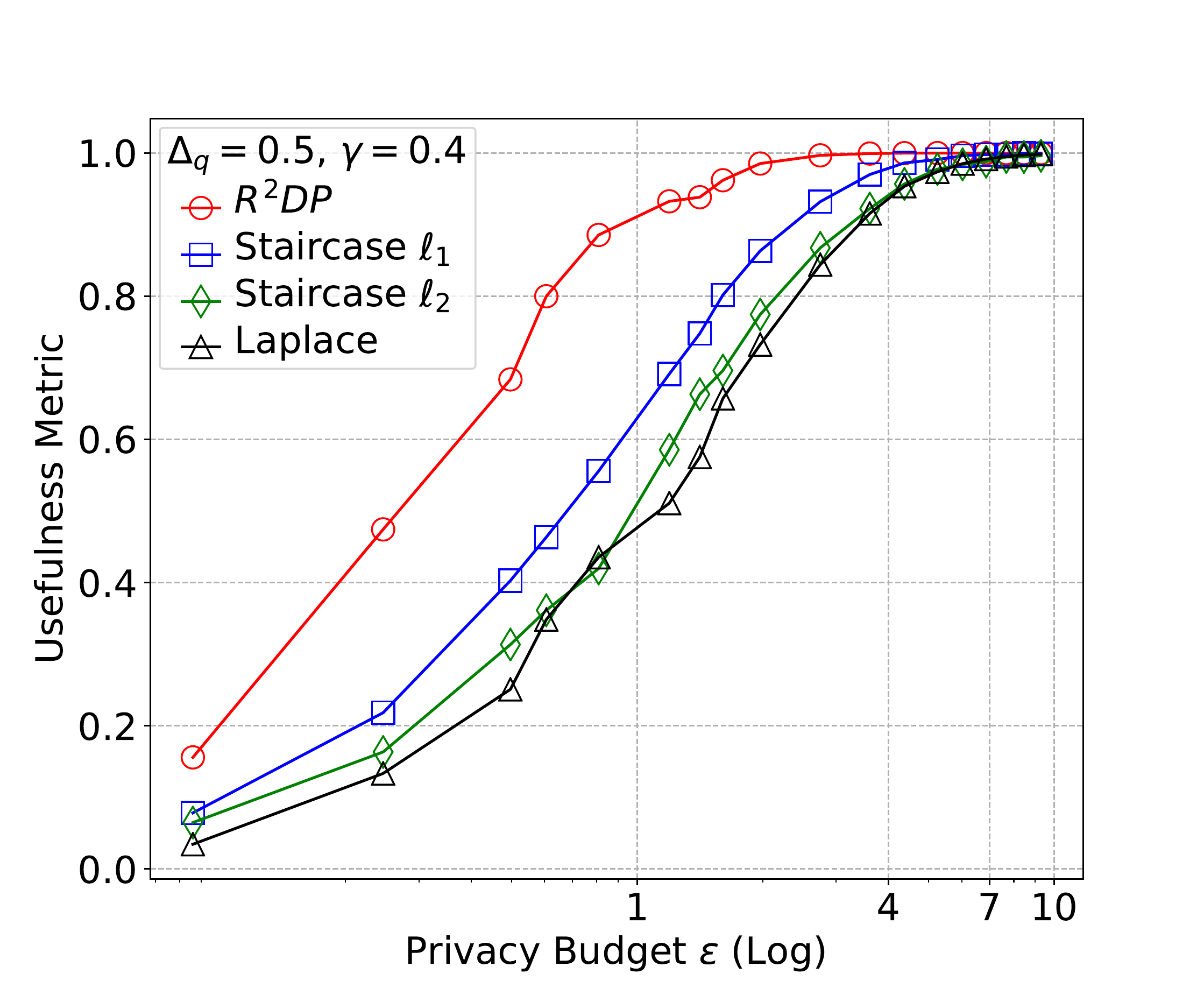}
		 }\hspace{-0.25in}
	\subfigure[$\Delta q=0.5, \gamma=0.6$]{
		\includegraphics[angle=0, width=0.25\linewidth]{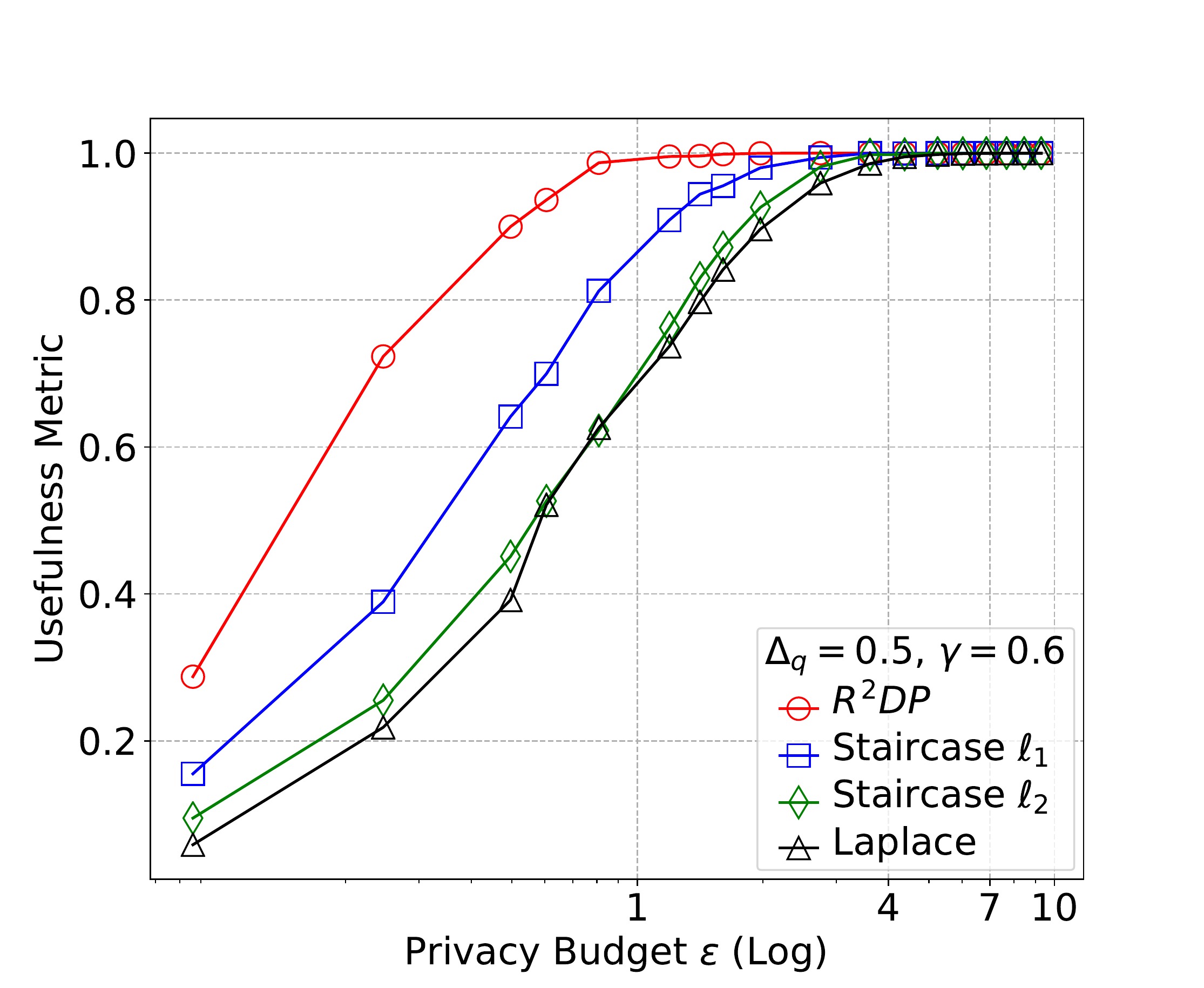}
		 }\hspace{-0.25in}
	\subfigure[$\Delta q=0.5, \gamma=0.9$]{
		\includegraphics[angle=0, width=0.25\linewidth]{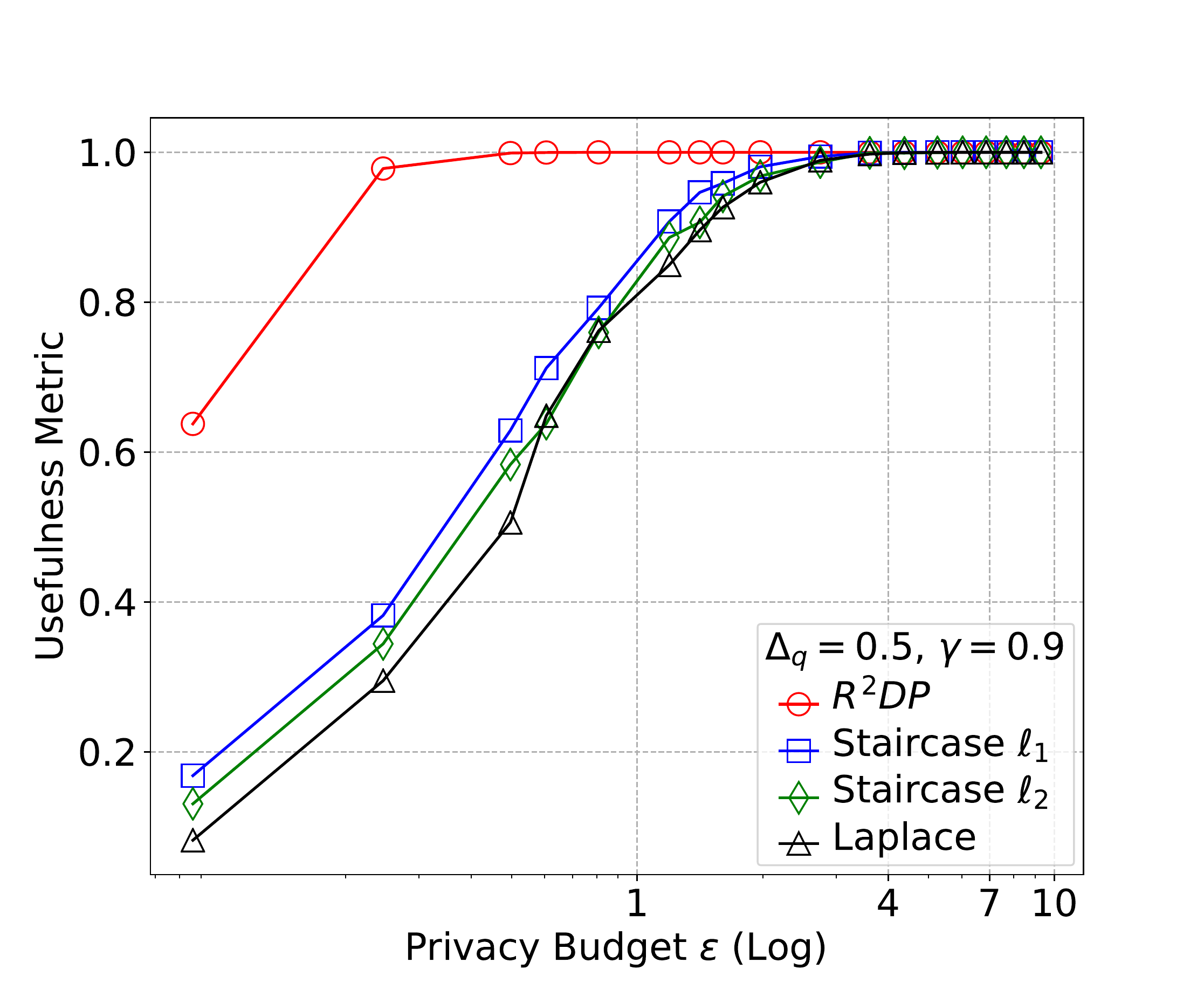}
 		}\hspace{-0.3in}
	\subfigure[$\Delta q=1, \gamma=0.1$]{
		\includegraphics[angle=0, width=0.25\linewidth]{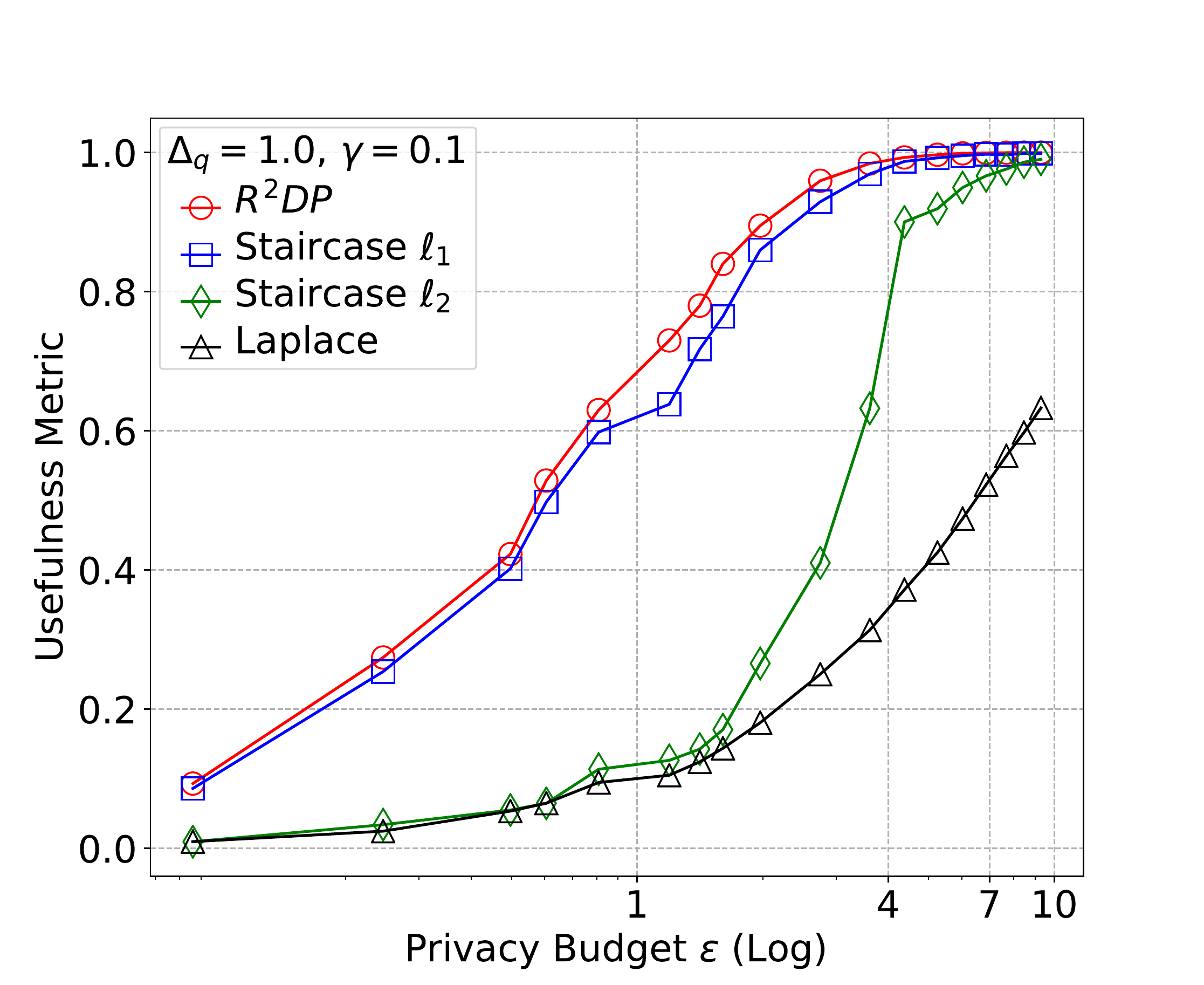}
		 }\hspace{-0.25in}
	\subfigure[$\Delta q=1, \gamma=0.4$]{
		\includegraphics[angle=0, width=0.25\linewidth]{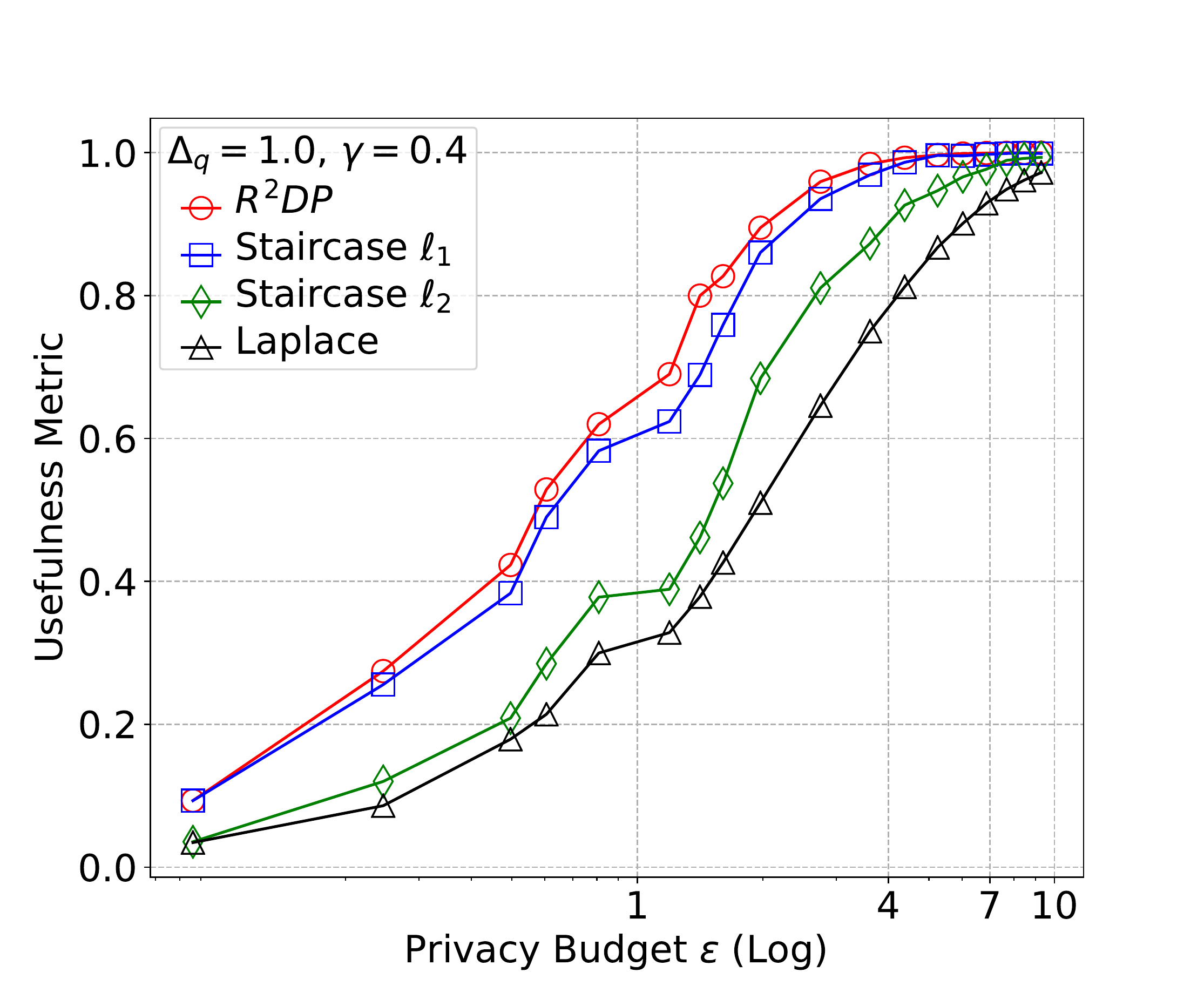}
		 }\hspace{-0.25in}
	\subfigure[$\Delta q=1, \gamma=0.6$]{
		\includegraphics[angle=0, width=0.25\linewidth]{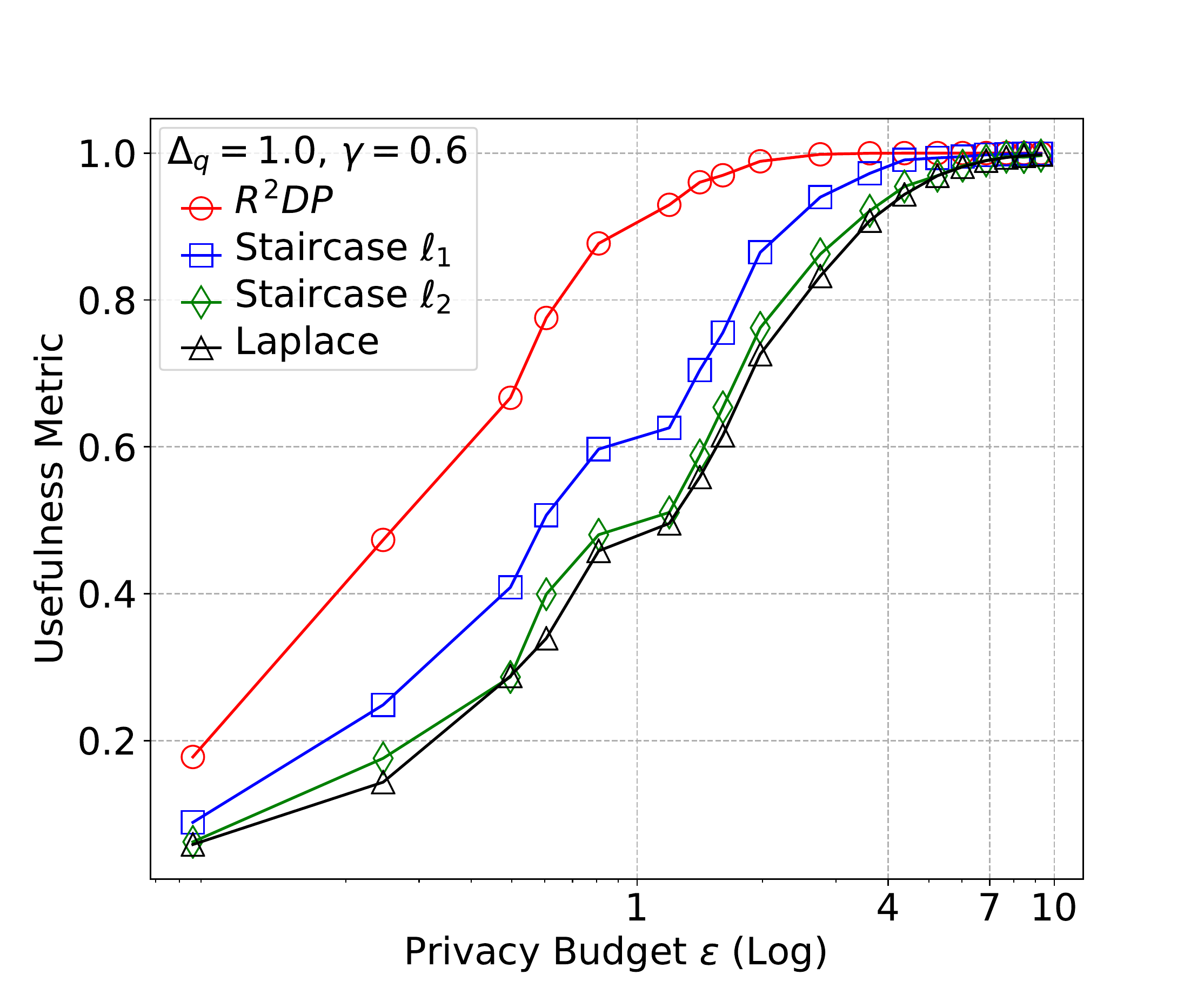}
		 }\hspace{-0.25in}
	\subfigure[$\Delta q=1, \gamma=0.9$]{
		\includegraphics[angle=0, width=0.25\linewidth]{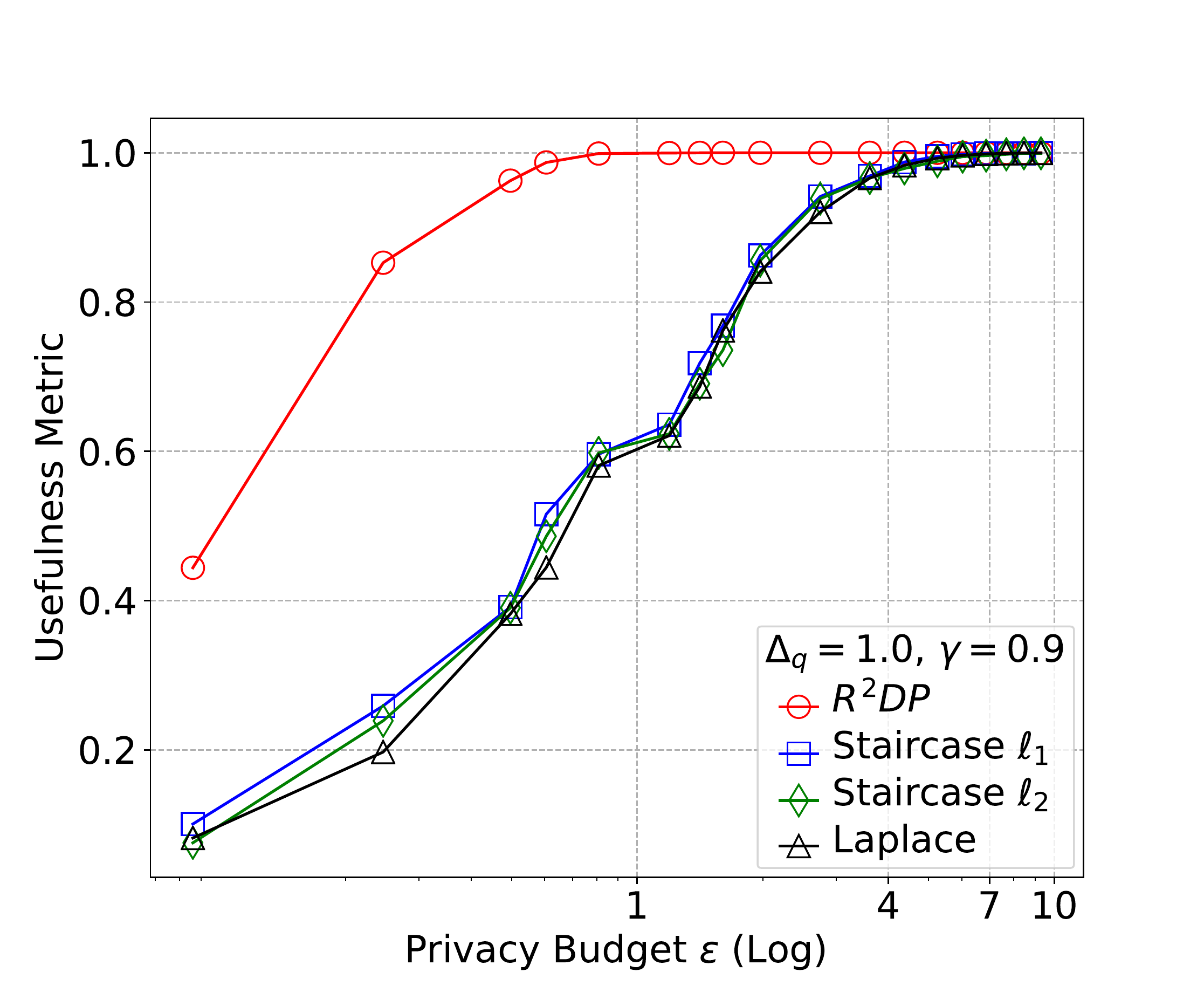}
		 }
	\caption[Optional caption for list of figures]
	{Usefulness metric: R$^2$DP (with five PDFs, i.e., Gamma, Uniform, Truncated Gaussian, Noncentral Chi-squared and Rayleigh distributions) strictly outperforms Laplace and Staircase mechanisms for statistical queries, where the ratio of improvement depends on the values of $\Delta q$, $\gamma$ and $\epsilon$.}\vspace{-0.1in}
	\label{fig:newcnt}
\end{figure*}

In this section, we experimentally evaluate the performance of R$^2$DP using six different utility metrics, i.e., $\ell_1$,
$\ell_2$, entropy, usefulness, Mallows and R\'enyi
divergence. Furthermore, we investigate the tightness of R$^2$DP under R\'enyi \DP (RDP in short)~\cite{mironov2017renyi}
which provides a universal formulation of the privacy losses of various
DP mechanisms, as shown in Appendix~\ref{renyiDPsec} (facilitating the
comparison between different mechanisms).  Our objective is to verify
the following two properties about the performance of the R$^2$DP
framework w.r.t. all seven utility and privacy metrics: (1) R$^2$DP
produces near-optimal results and (2) R$^2$DP performs strictly better
than well-known baseline mechanisms, e.g, Laplace and Staircase
mechanisms, in settings where an optimal PDF is not known, e.g.,
usefulness utility metric or R\'enyi \DP.

\subsection{Experimental Setting}

We perform all the experiments and comparisons on the Privacy Integrated Queries (PINQ) platform \cite{McSherry09}. Besides basic statistical queries, two applications in the current suite (\emph{machine learning} and \emph{social network analysis}) are employed to evaluate the accuracy of R$^2$DP and compare it to Laplace and Staircase mechanisms. 
\subsubsection{Statistical Queries}

In the first set of our experiments, we examine the benefits of
R$^2$DP using basic statistical functions, i.e., count and
average. The dataset comes from a sensor network experiment carried
out in the Mitsubishi Electric Research Laboratories (MERL) and
described in~\cite{Wren:2007:MMD:1352922.1352926}. MERL has collected
motion sensor data from a network of over 200 sensors for a year and
the dataset contains over 30 million raw motion records. To illustrate
the query performance with different sensitivities, we create the
queries based on a subset of the data including aggregated events that
are recorded by closely located sensors over 5-minute intervals. We
formed in this way $10$ input signals corresponding to $10$ spatial
zones (each zone is covered by a group of sensors). Since each
individual can activate several sensors and travel through different
zones, we define moving average functions with arbitrary sensitivity
values, e.g., $\Delta q \in [0.1,5]$. For instance, we could be
interested in the summation of the moving averages over the past 30
min for zones 1 to 4. We apply R$^2$DP w.r.t. usefulness, $\ell_1$,
$\ell_2$, entropy, and R\'enyi metrics, respectively.

    \subsubsection{Social Network} 
    Social network degree distribution is performed on a Facebook dataset~\cite{snapnets}. They consist of ``circles'' and ``friends lists'' from Facebook by representing different individuals as nodes (47,538 nodes) and friend connections as edges (222,887 edges). Recall that the Mallows metric is frequently used for social network (graph-based) applications  \cite{10.1007/978-3-642-36594-2_26}. We thus apply R$^2$DP w.r.t. the Mallows metric in this group of experiments.

    \subsubsection{Machine Learning} 
    
    Naive Bayes classification is performed on two datasets: Adult dataset (in the UCI ML Repository) \cite{Kohavi96scalingup} and KDDCup99 dataset \cite{821515}. First, the Adult dataset includes the demographic information of 48,842 different adults in the US (14 features). It can be utilized to train a Naive Bayes classifier to predict if any adult's annual salary is greater than 50k or not. Second, the KDD competition dataset was utilized to build a network intrusion detector (given 24 training attack types) by classifying ``bad'' connections and ``good'' connections. Recall that the usefulness metric is commonly used for machine learning \cite{Blum:2008:LTA:1374376.1374464}. We thus apply R$^2$DP w.r.t. the usefulness in this group of experiments.

\subsection{Basic Statistical Queries}
We validate the effectiveness of R$^2$DP using two basic statistical queries: count (sensitivity=1) and moving average with different window sizes, e.g., sensitivity $\in[0.1,2]$ to comprehensively study the performance of R$^2$DP by benchmarking with Laplace and Staircase mechanisms. We have the following observations.

\begin{figure*}[!tb]
\label{fig:newcnt1}
\centering
	\subfigure[$\ell_1, \Delta q=0.1$]{
		\includegraphics[angle=0, width=0.25\linewidth]{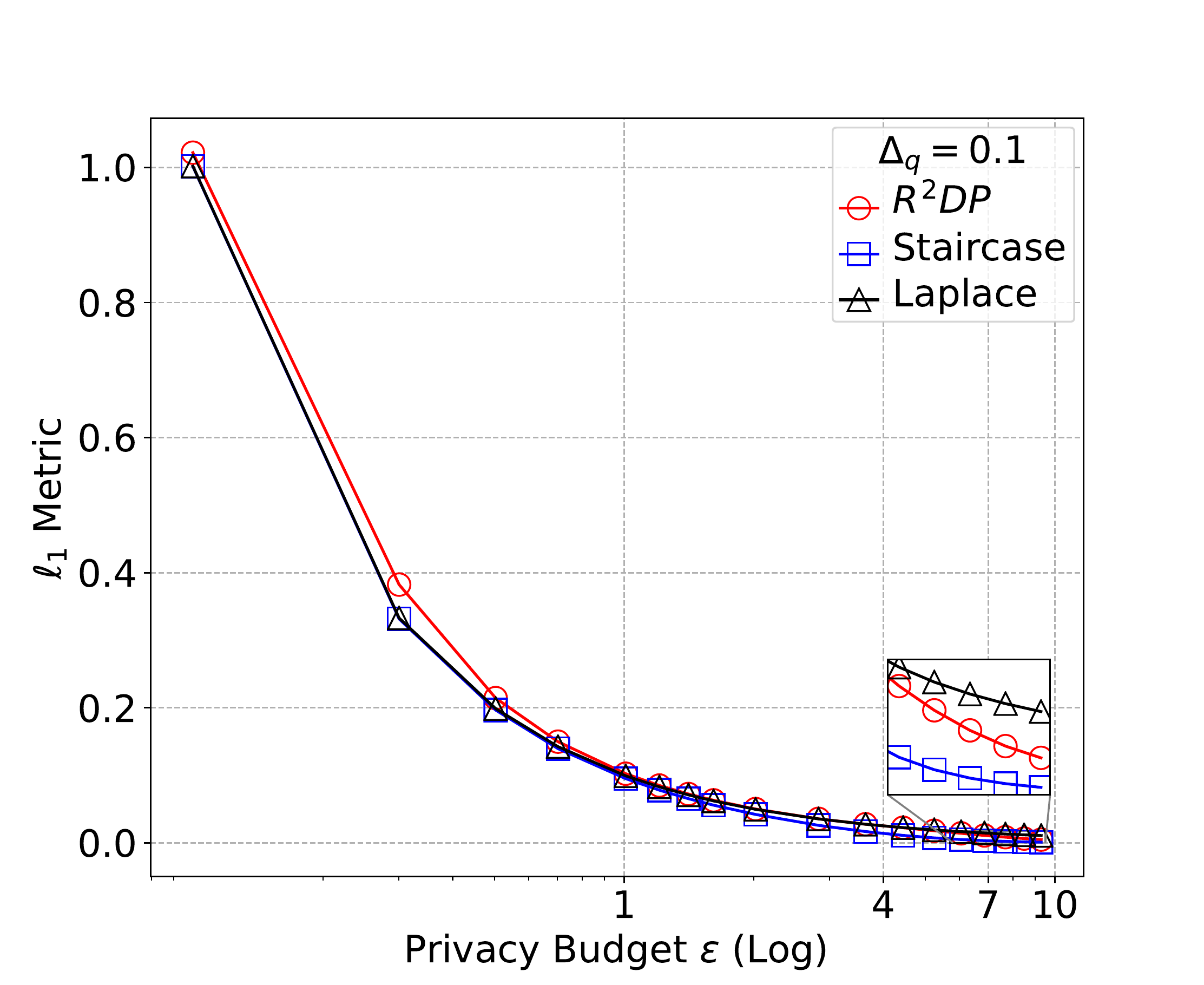}
		 }\hspace{-0.25in}
	\subfigure[$\ell_1, \Delta q=0.5$]{
		\includegraphics[angle=0, width=0.25\linewidth]{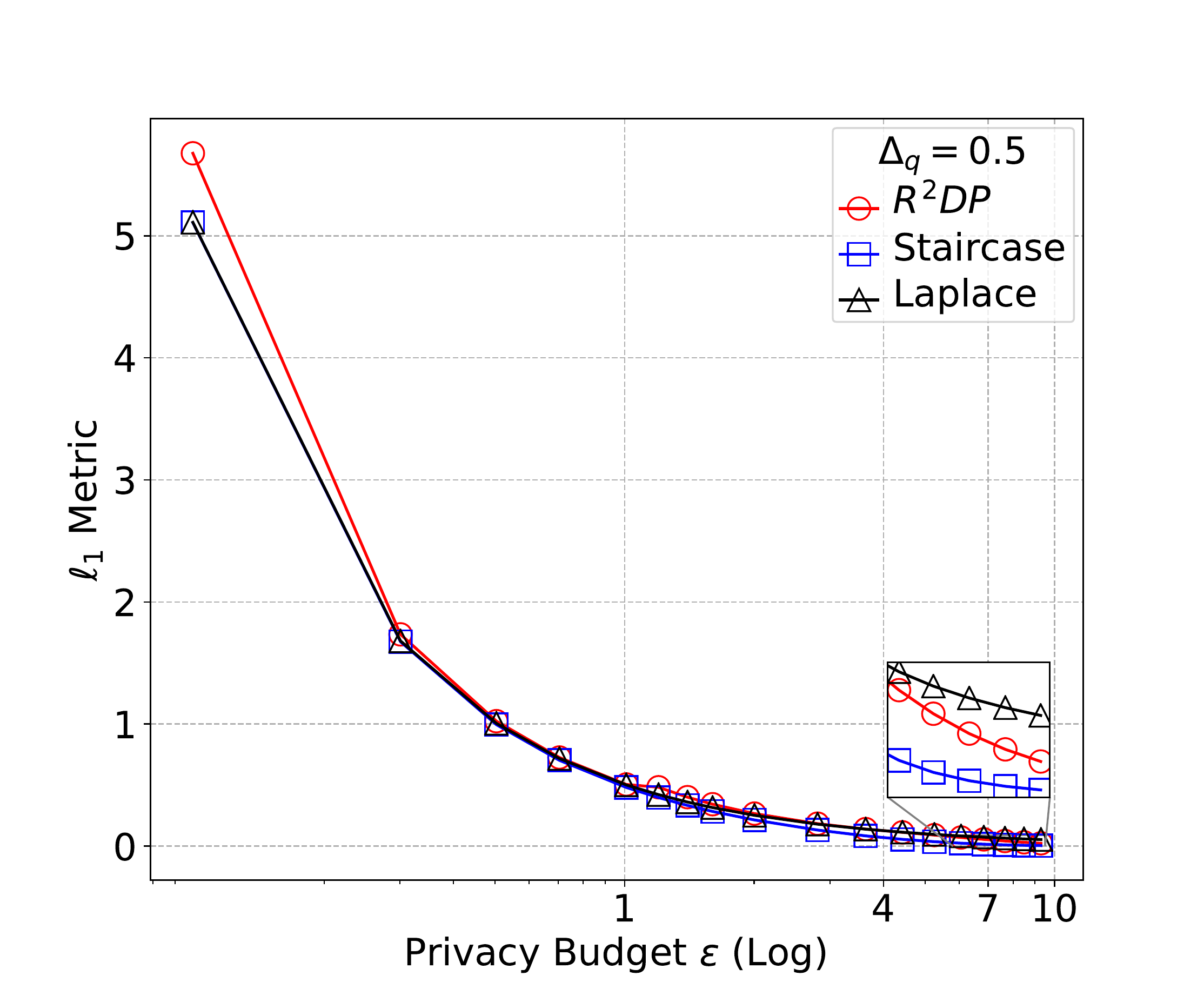}
		 }\hspace{-0.25in}
	\subfigure[$\ell_1, \Delta q=1$]{
		\includegraphics[angle=0, width=0.25\linewidth]{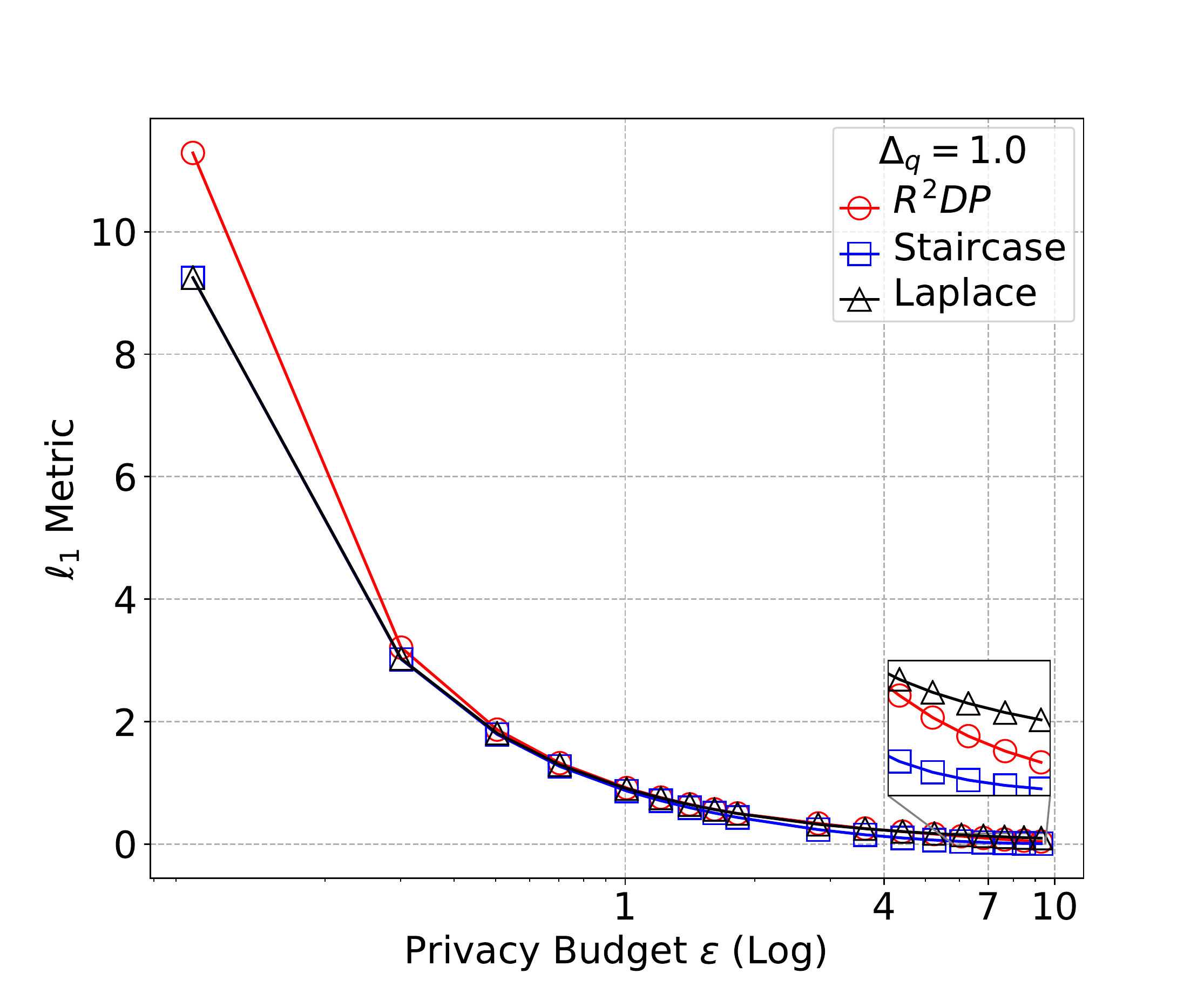}
		 }\hspace{-0.25in}
	\subfigure[$\ell_1, \Delta q=1.5$]{
 		\includegraphics[angle=0, width=0.25\linewidth]{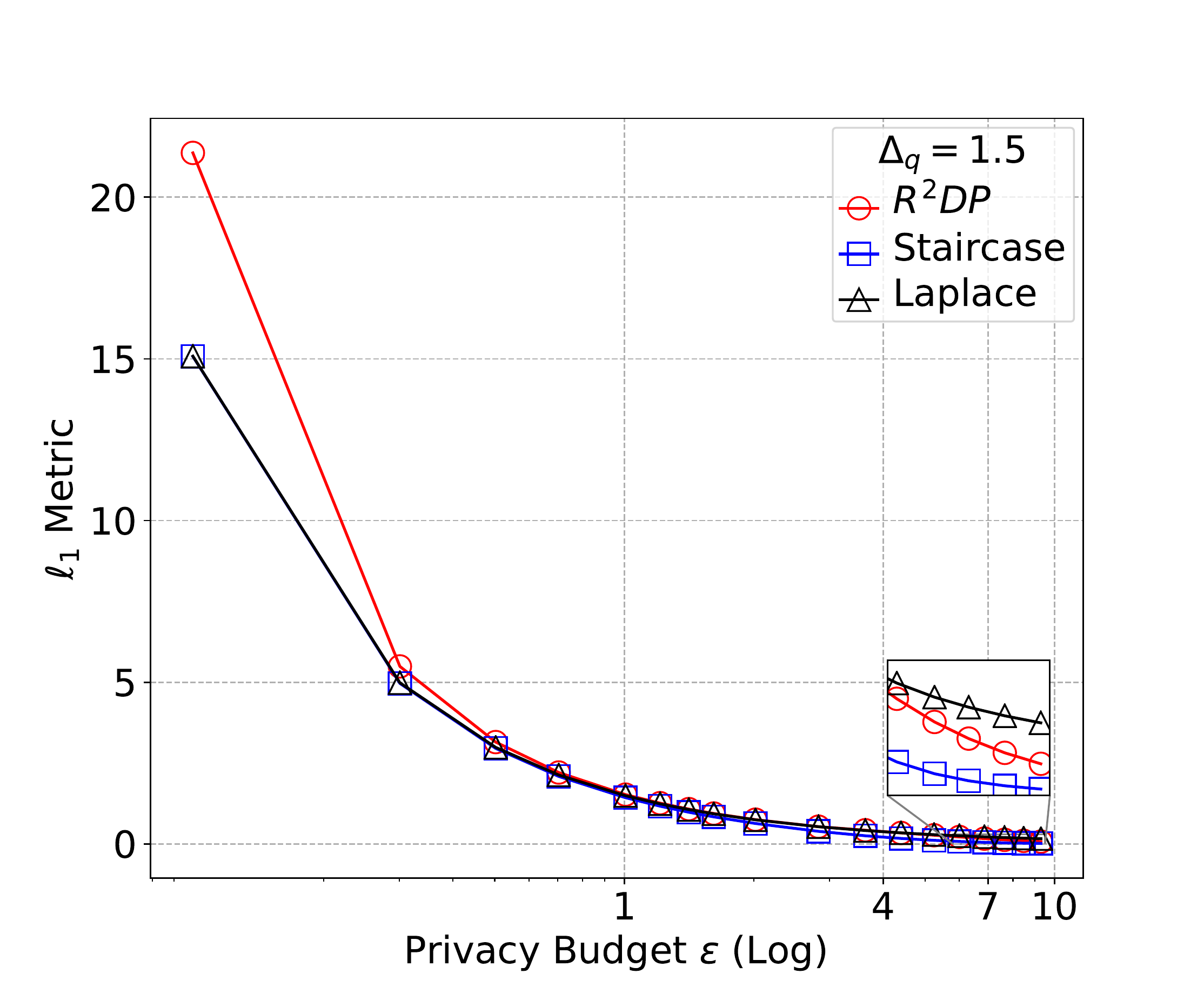}
  		}\hspace{-0.3in}
 	\subfigure[$\ell_2, \Delta q=0.1$]{
 		\includegraphics[angle=0, width=0.25\linewidth]{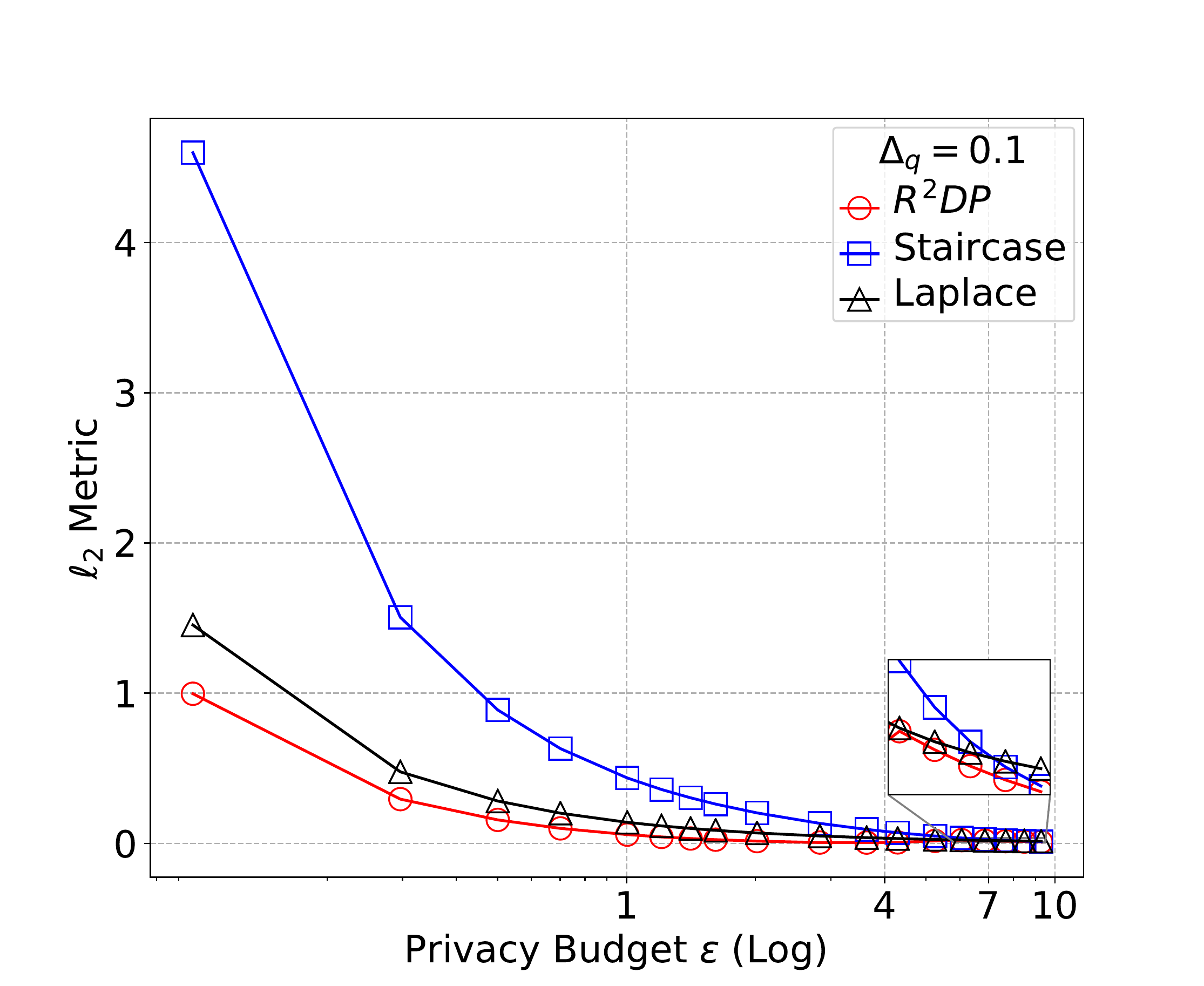}
 		 }\hspace{-0.25in}
 	\subfigure[$\ell_2, \Delta q=0.5$]{
 		\includegraphics[angle=0, width=0.25\linewidth]{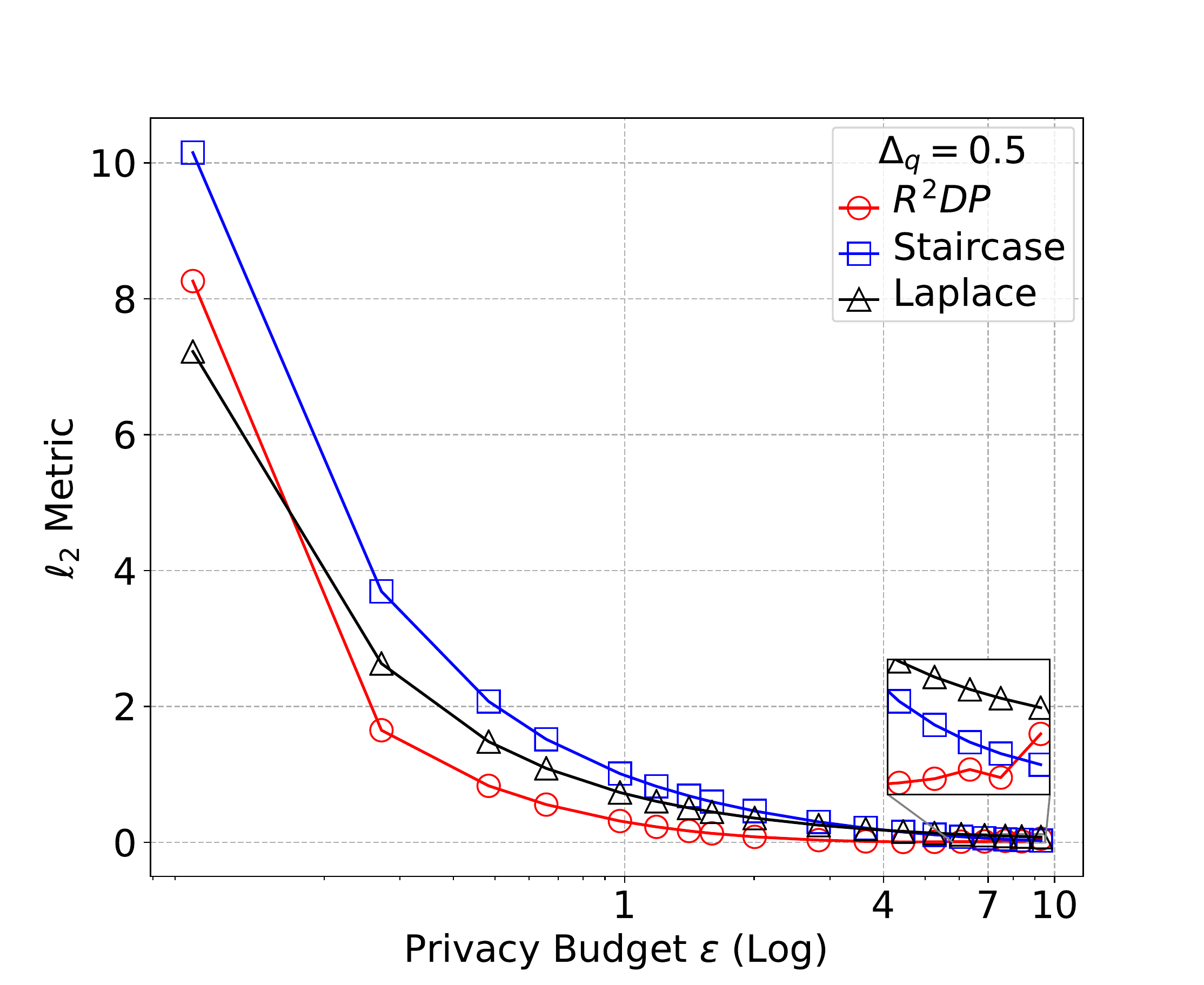}
 }\hspace{-0.25in}
 	\subfigure[$\ell_2, \Delta q=1$]{
 		\includegraphics[angle=0, width=0.25\linewidth]{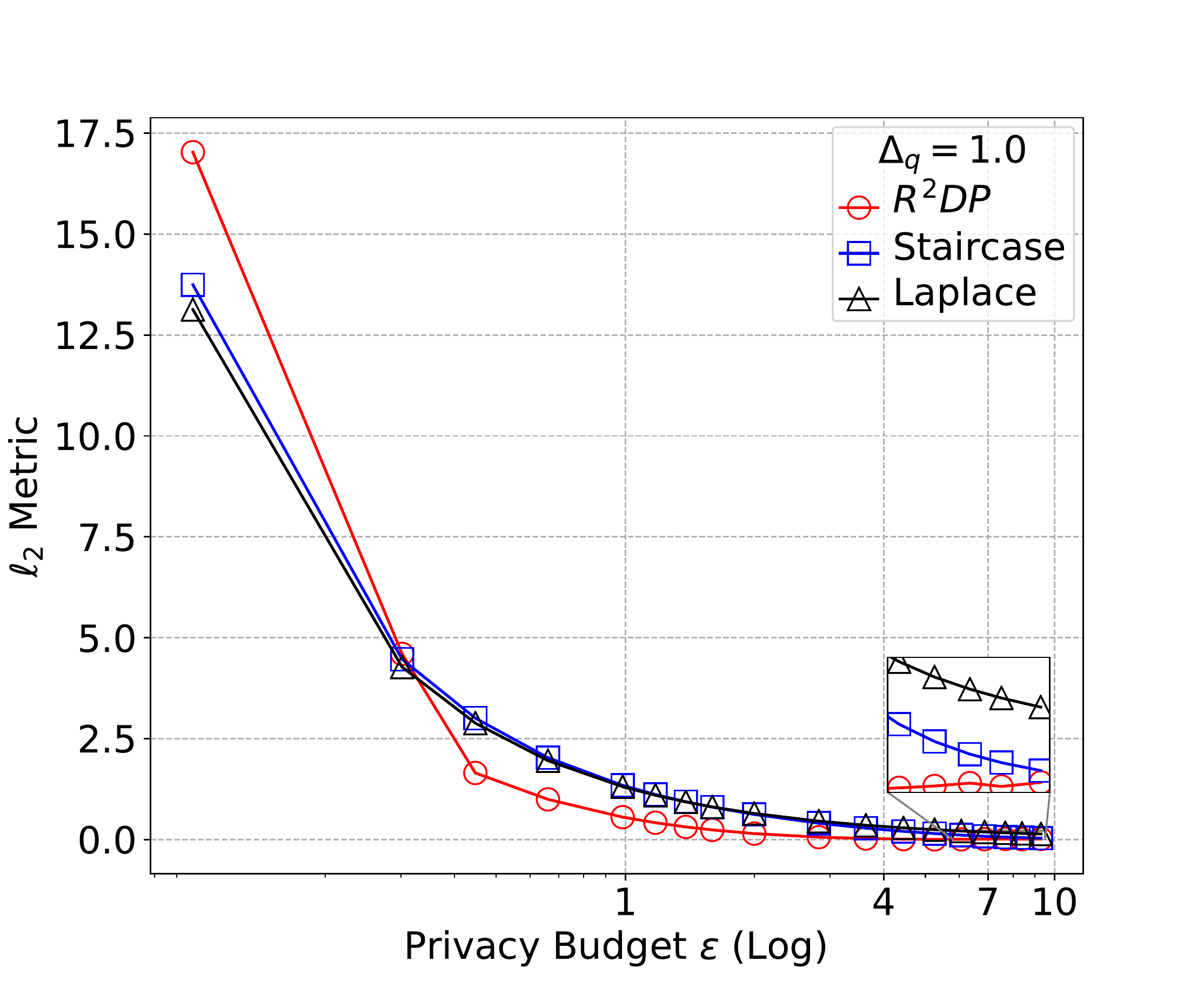}
 		 }\hspace{-0.25in}
 	\subfigure[$\ell_2, \Delta q=1.5$]{
 		\includegraphics[angle=0, width=0.25\linewidth]{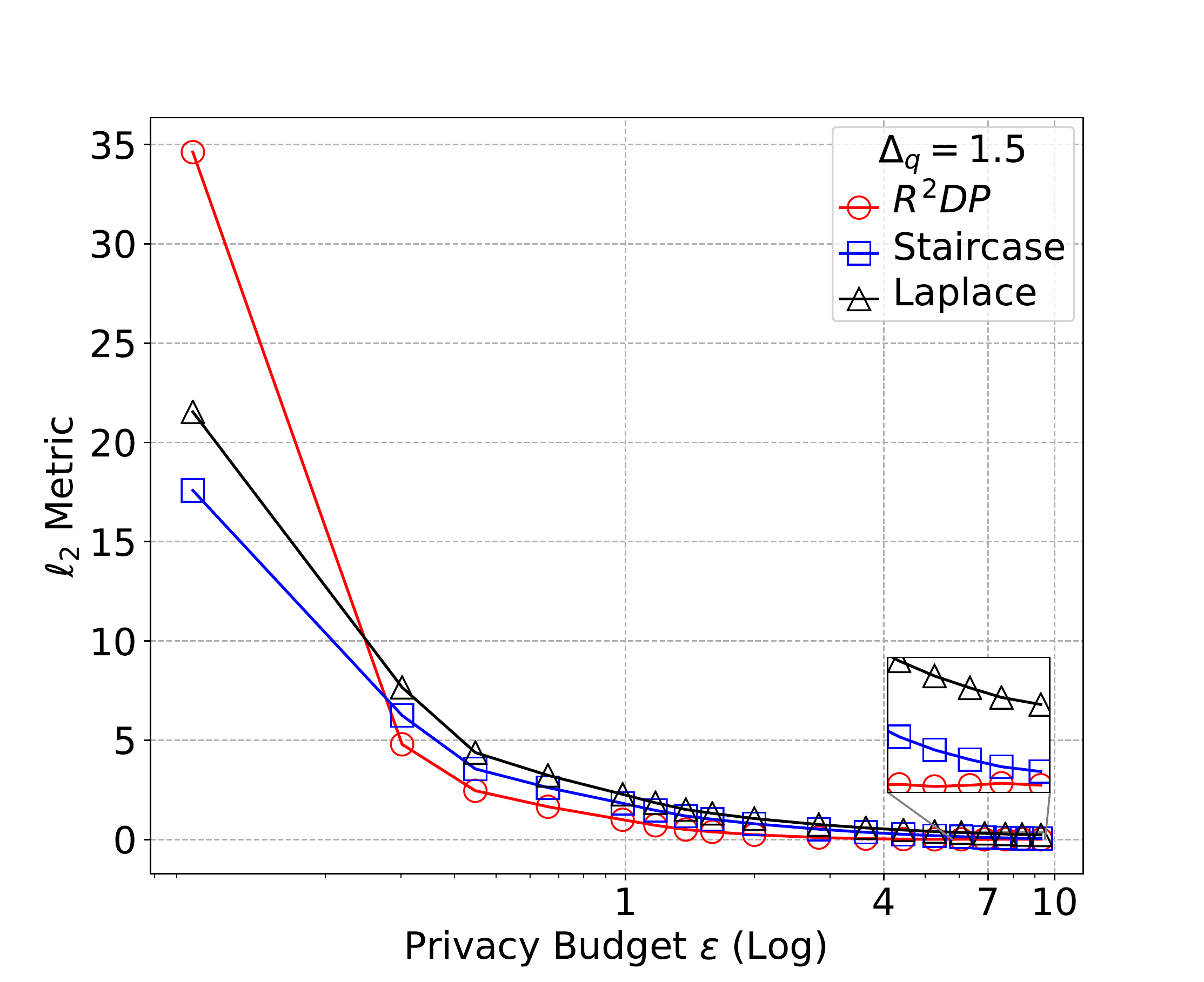}
		 }
	\caption[Optional caption for list of figures]
	{$\ell_1$ and $\ell_2$ metrics:  R$^2$DP compared to Laplace and Staircase mechanisms for statistical queries (with five PDFs, i.e., Gamma, Uniform, Truncated Gaussian, Noncentral Chi-squared and Rayleigh distributions).}
	\label{fig:newld}
\end{figure*}

\subsubsection{Usefulness Metric} 
We compare R$^2$DP with the baseline Laplace and
two classes of Staircase mechanisms proposed in~\cite{Gupte1}
w.r.t. $\ell_1$ and $\ell_2$ metrics, by varying the privacy budget
$\epsilon$, four error bounds $\gamma\in\{0.1, 0.4, 0.6, 0.9\}$ and
two different sensitivities (Section \ref{numericsec} additionally
shows numerical results to provide a more comprehensive evaluation for
the usefulness metric). As shown in Figure \ref{fig:newcnt}, R$^2$DP
generates strictly better results w.r.t. the usefulness metric, and
the ratio of improvement depends on values of $\epsilon$, $\Delta q$
and $\gamma$. In particular, we observe that the improvement is
relatively larger for a larger error bound and smaller sensitivity
(Figure~\ref{fig:newcnt} (a,b,e,f) vs. (c,d,g,h)). One important
factor determining the improvement is the ratio between $\gamma$ and
$\Delta q$, since it exponentially affects the search space of the
R$^2$DP mechanism. Furthermore, we observe that the Laplace and
the staircase mechanisms are not optimal (w.r.t. usefulness) for very
small and large values of $\epsilon$, respectively, even though they
are known to be optimal under other utility metrics (e.g.,
\cite{7353177}).

\begin{figure*}[!tb]
\centering
	\subfigure[$\Delta q=0.1$]{
		\includegraphics[angle=0, width=0.25\linewidth]{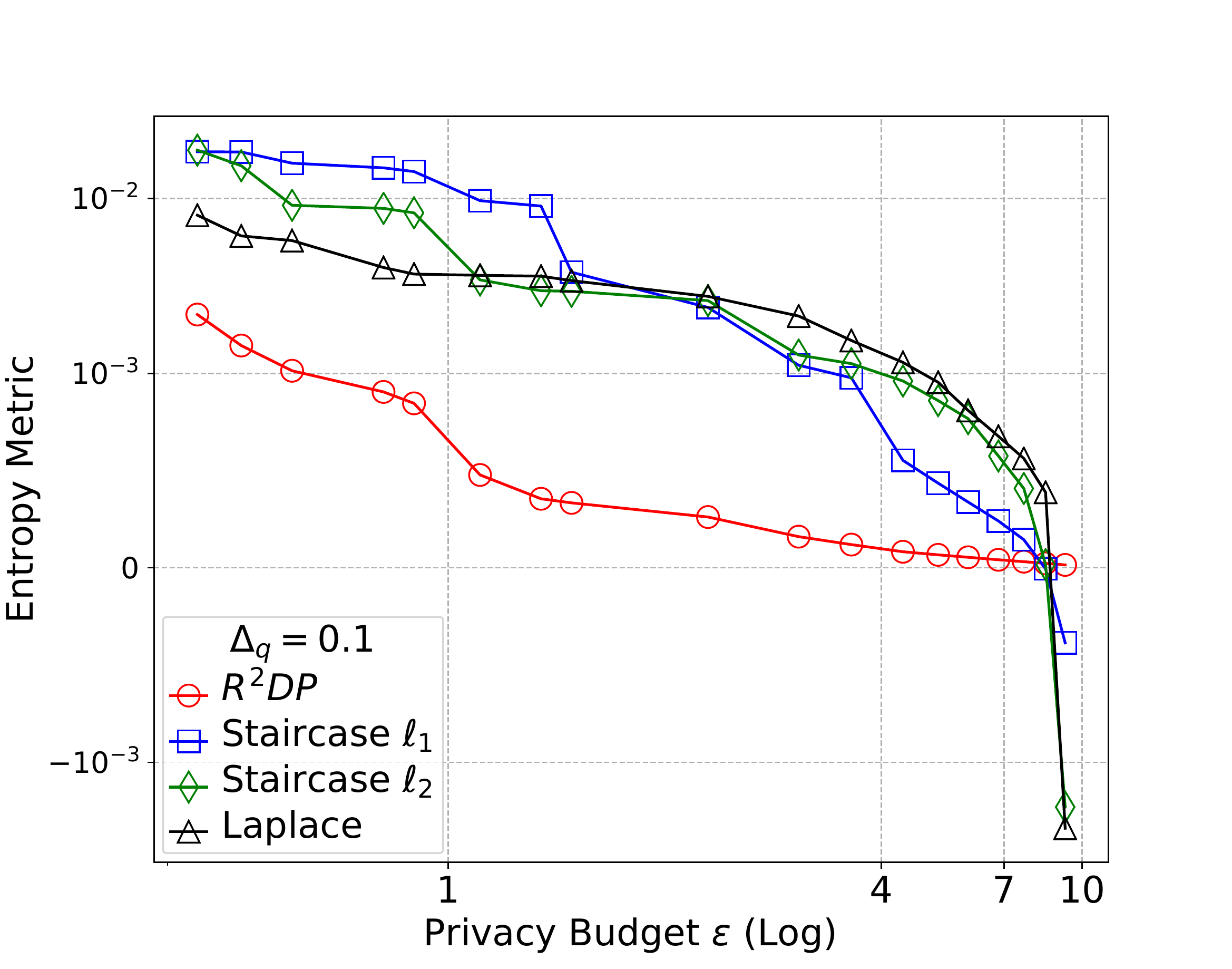}
		\label{entropy1} }\hspace{-0.25in}
	\subfigure[$\Delta q=0.5$]{
		\includegraphics[angle=0, width=0.25\linewidth]{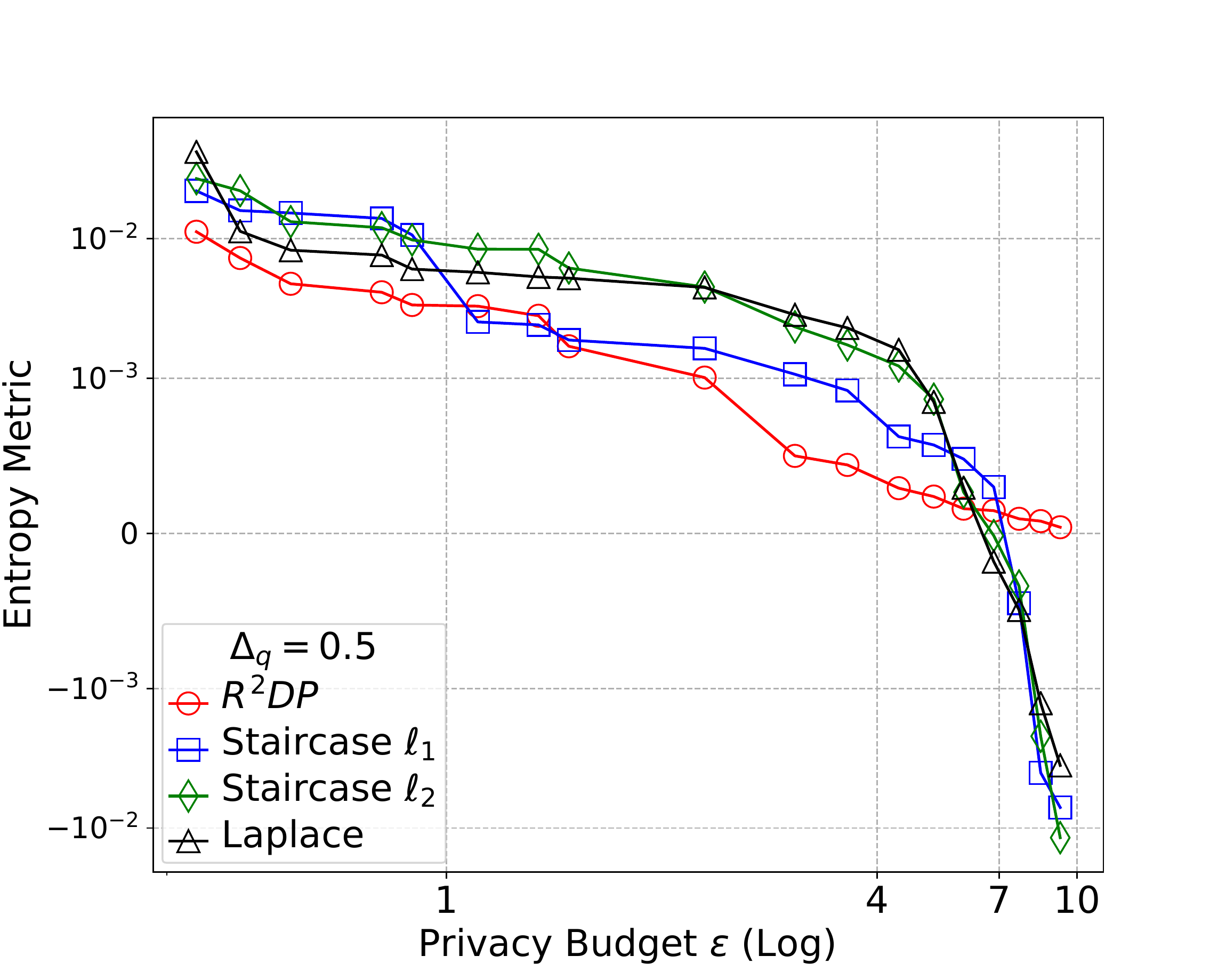} 		\label{entropy2}}\hspace{-0.25in}
\subfigure[$\Delta q=1$]{
		\includegraphics[angle=0, width=0.25\linewidth]{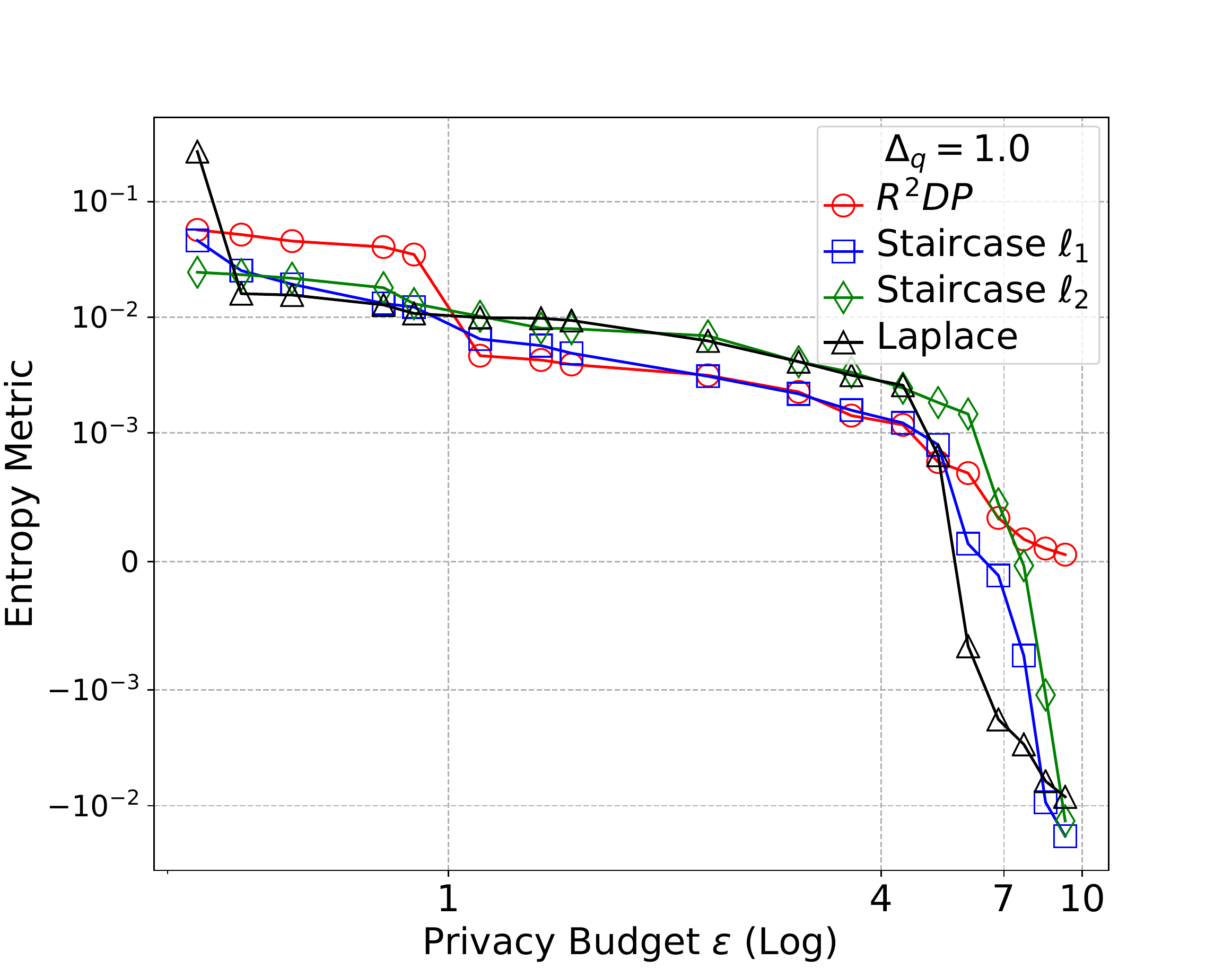}
		 }\hspace{-0.25in}
	\subfigure[$\Delta q=1.5$]{
		\includegraphics[angle=0, width=0.25\linewidth]{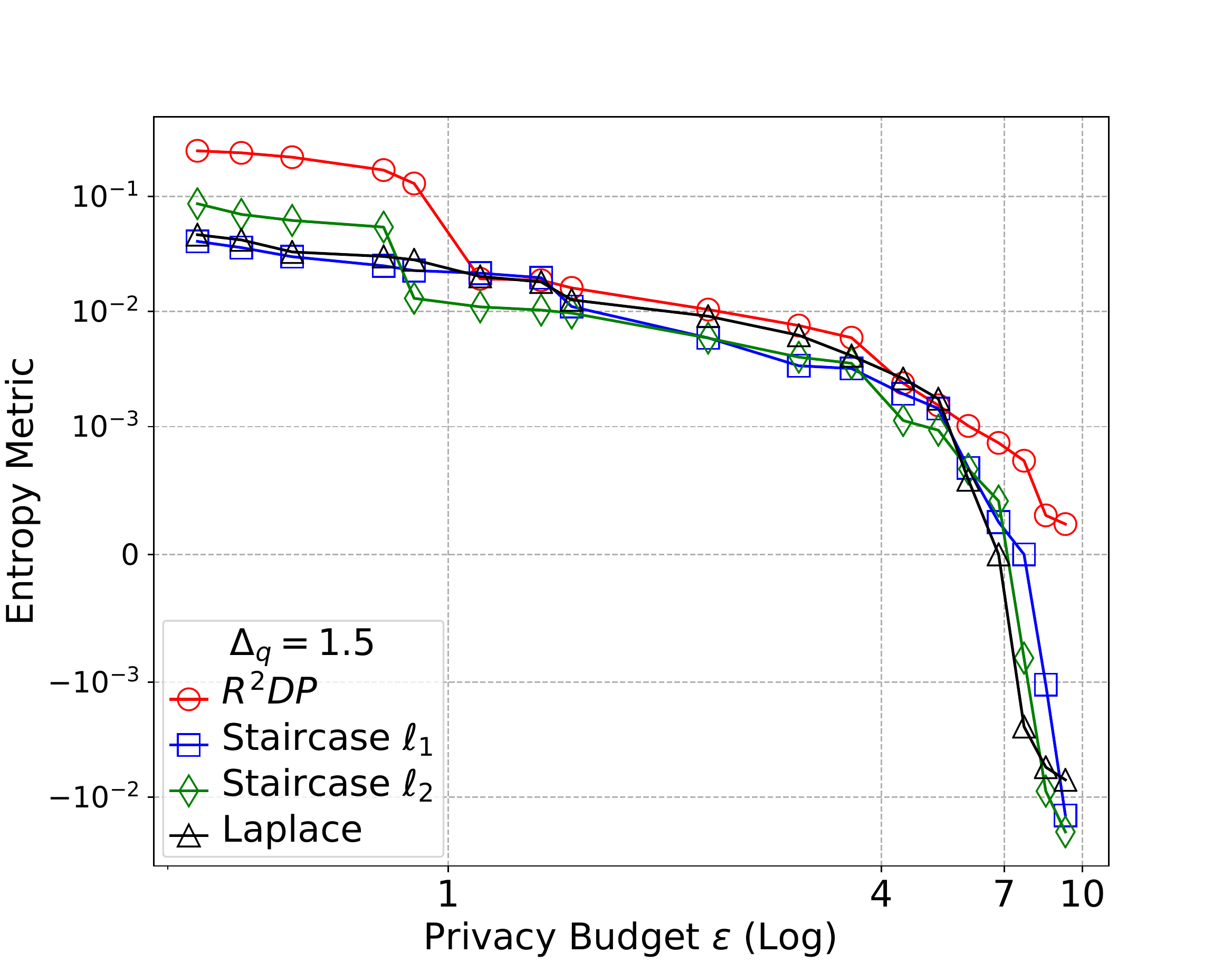} }
	\caption[Optional caption for list of figures]
	{KL Divergence (Relative entropy metric): R$^2$DP (with five PDFs, i.e., Gamma, Uniform, Truncated Gaussian, Noncentral Chi-squared and Rayleigh distributions) compared
to Laplace and Staircase mechanisms.}\vspace{0.05in}
	\label{fig:KLentropy_2}
\end{figure*}

\begin{figure*}[!tb]
\centering
	\subfigure[$\alpha=2, \Delta q=0.5$]{
		\includegraphics[angle=0, width=0.25\linewidth]{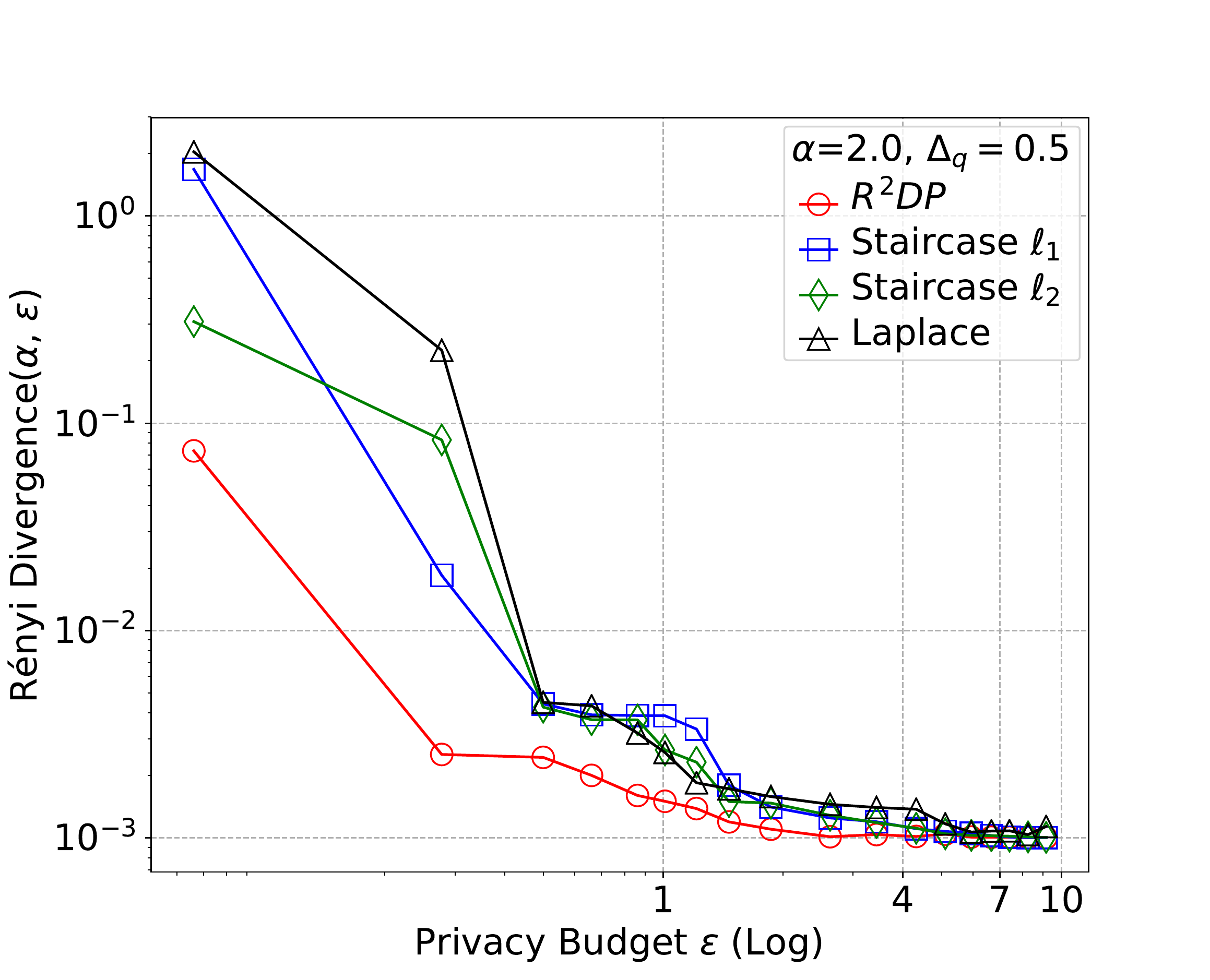}
		 }\hspace{-0.25in}
	\subfigure[$\alpha=2, \Delta q=1$]{
		\includegraphics[angle=0, width=0.25\linewidth]{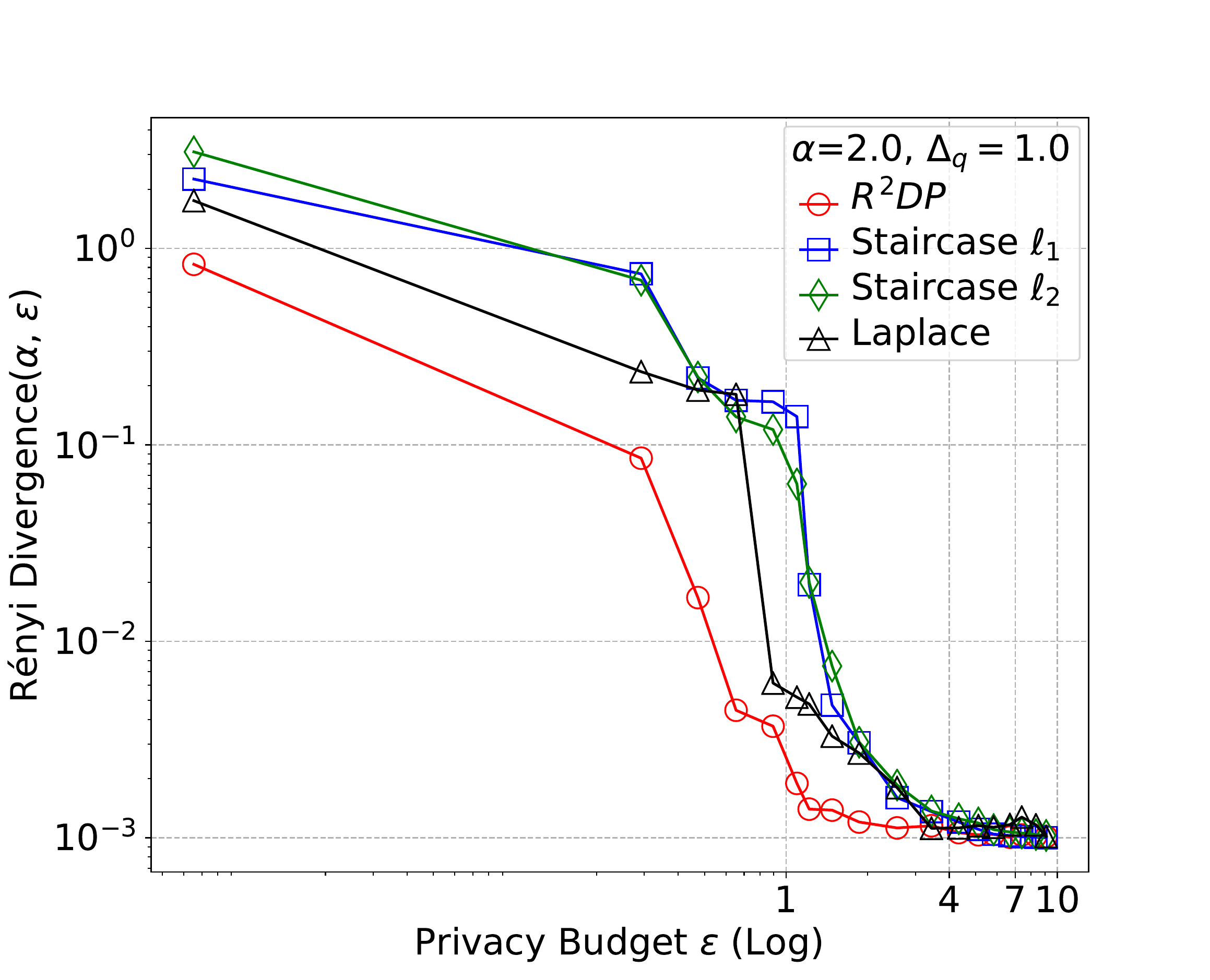} 		}\hspace{-0.25in}
\subfigure[$\alpha=3, \Delta q=0.5$]{
		\includegraphics[angle=0, width=0.25\linewidth]{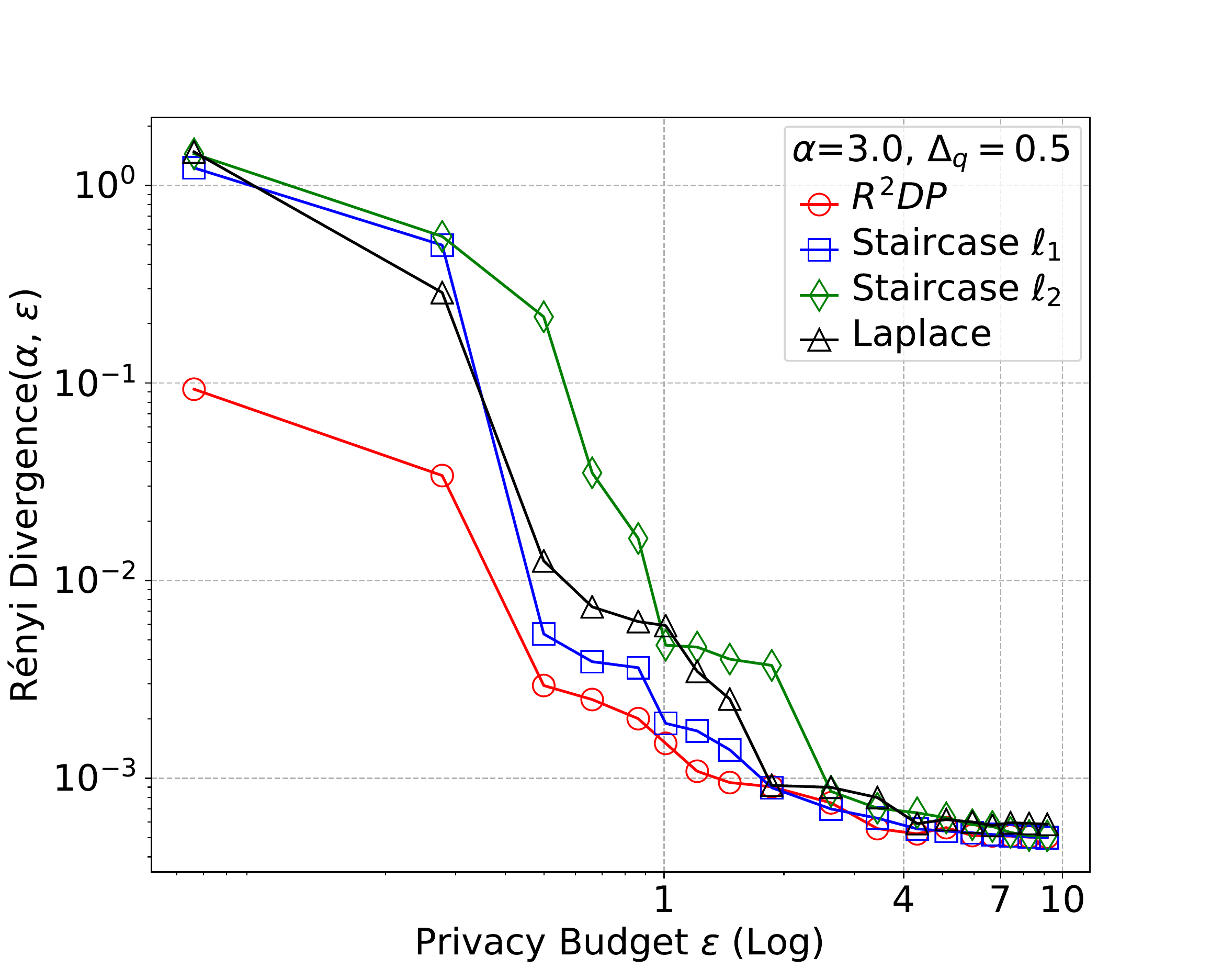}
	 }\hspace{-0.25in}
	\subfigure[$\alpha=2, \Delta q=1$]{
		\includegraphics[angle=0, width=0.25\linewidth]{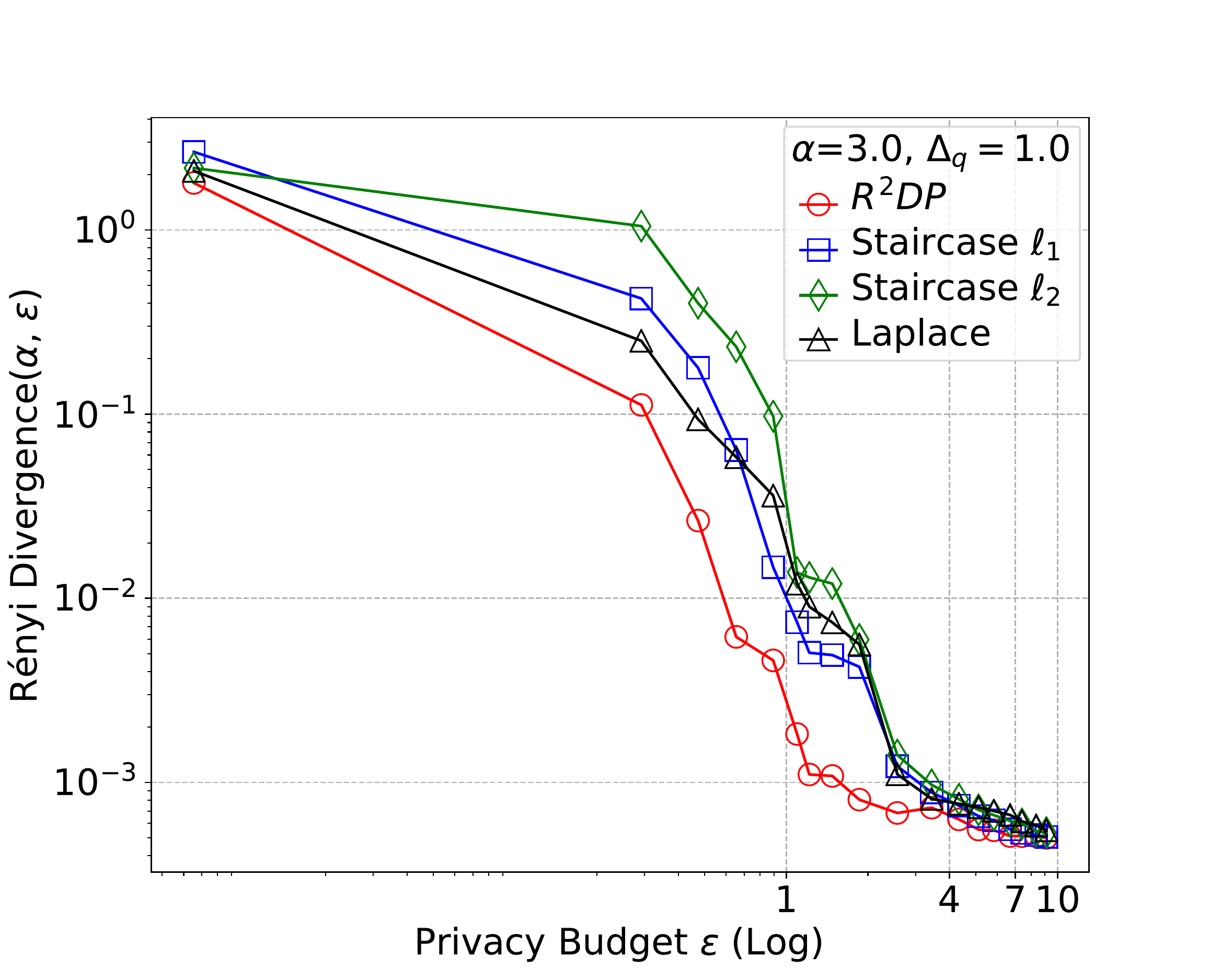} }
	\caption[Optional caption for list of figures]
	{R\'enyi Divergence (Relative entropy metric): R$^2$DP (with five PDFs, i.e., Gamma, Uniform, Truncated Gaussian, Noncentral Chi-squared and Rayleigh distributions) compared
to Laplace and Staircase mechanisms.}
	\label{fig:Renentropy_11}
\end{figure*}

\begin{figure*}[!tb]
\centering
	\subfigure[$\alpha=1$]{
		\includegraphics[angle=0, width=0.25\linewidth]{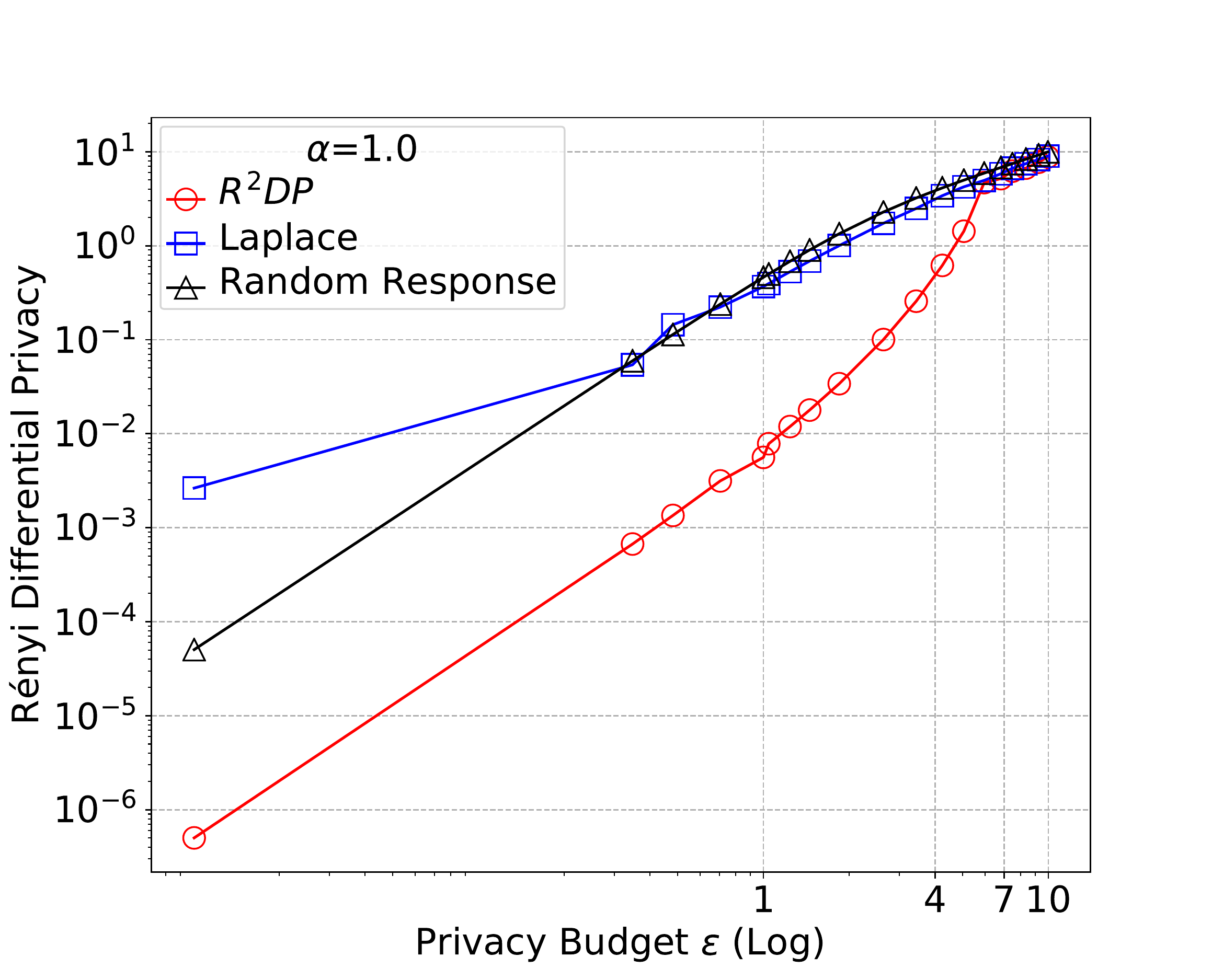}
		 }\hspace{-0.25in}
	\subfigure[$\alpha=3$]{
		\includegraphics[angle=0, width=0.25\linewidth]{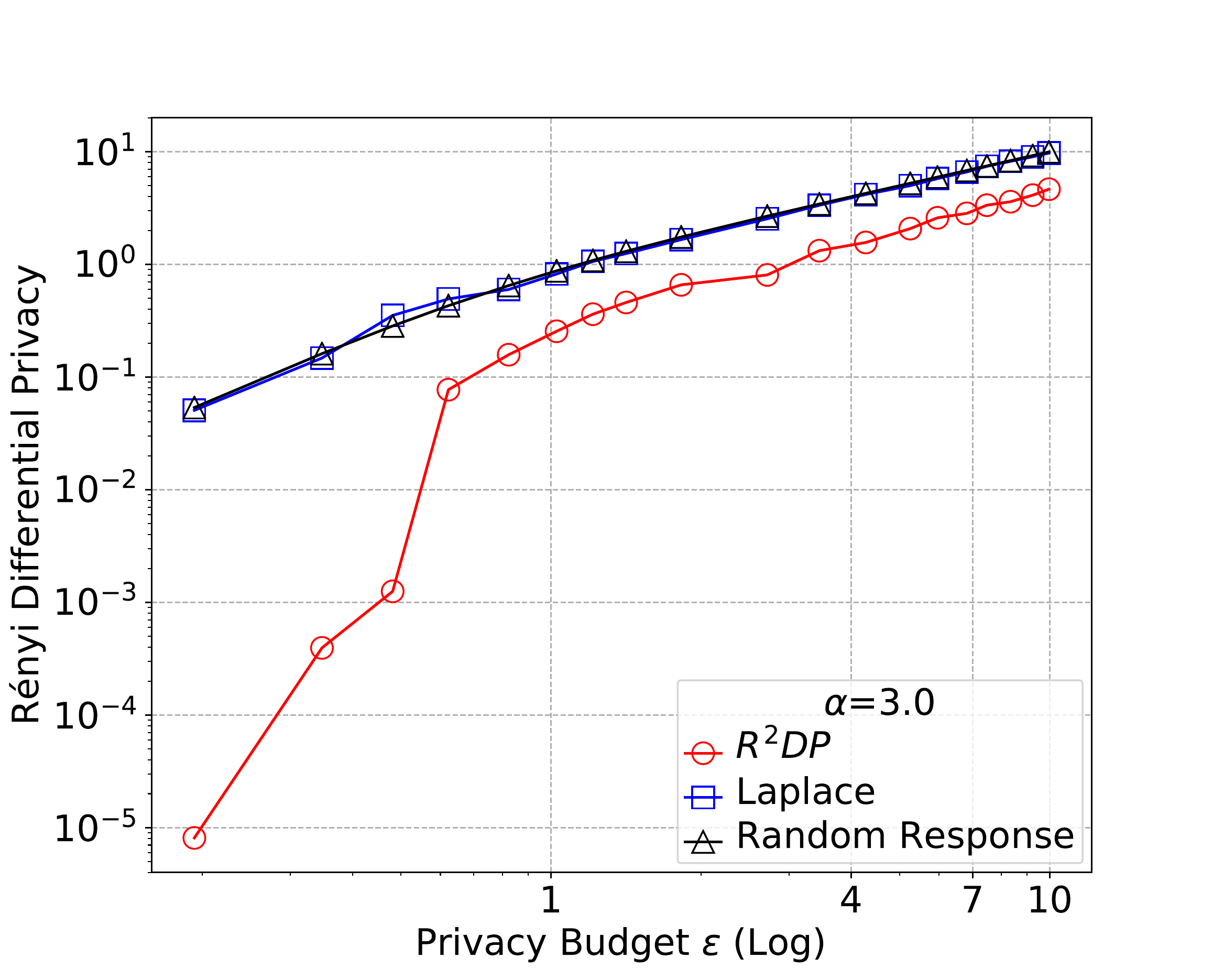} }\hspace{-0.25in}
\subfigure[$\alpha=4$]{
		\includegraphics[angle=0, width=0.25\linewidth]{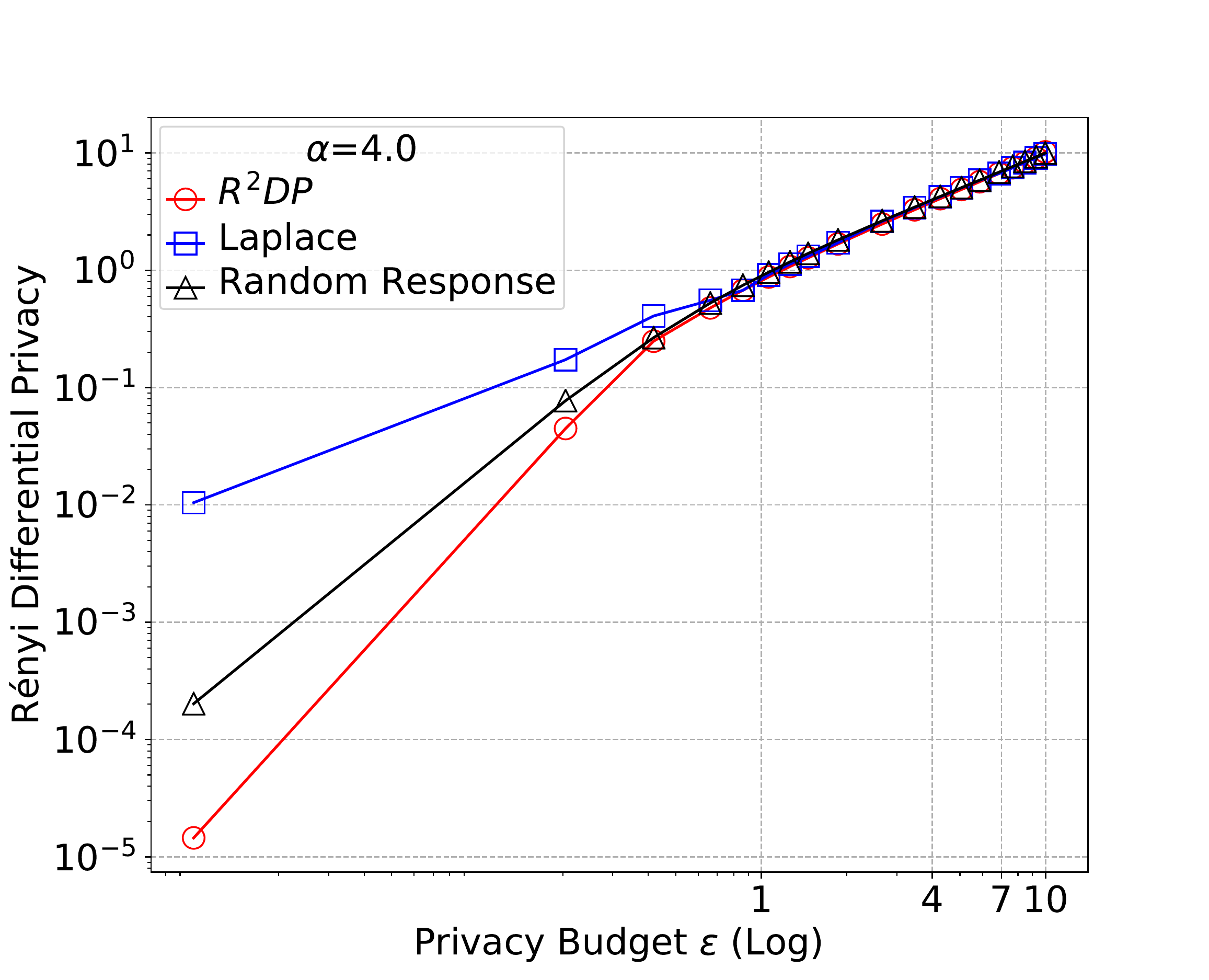}
		 }\hspace{-0.25in}
	\subfigure[$\alpha=5$]{
		\includegraphics[angle=0, width=0.25\linewidth]{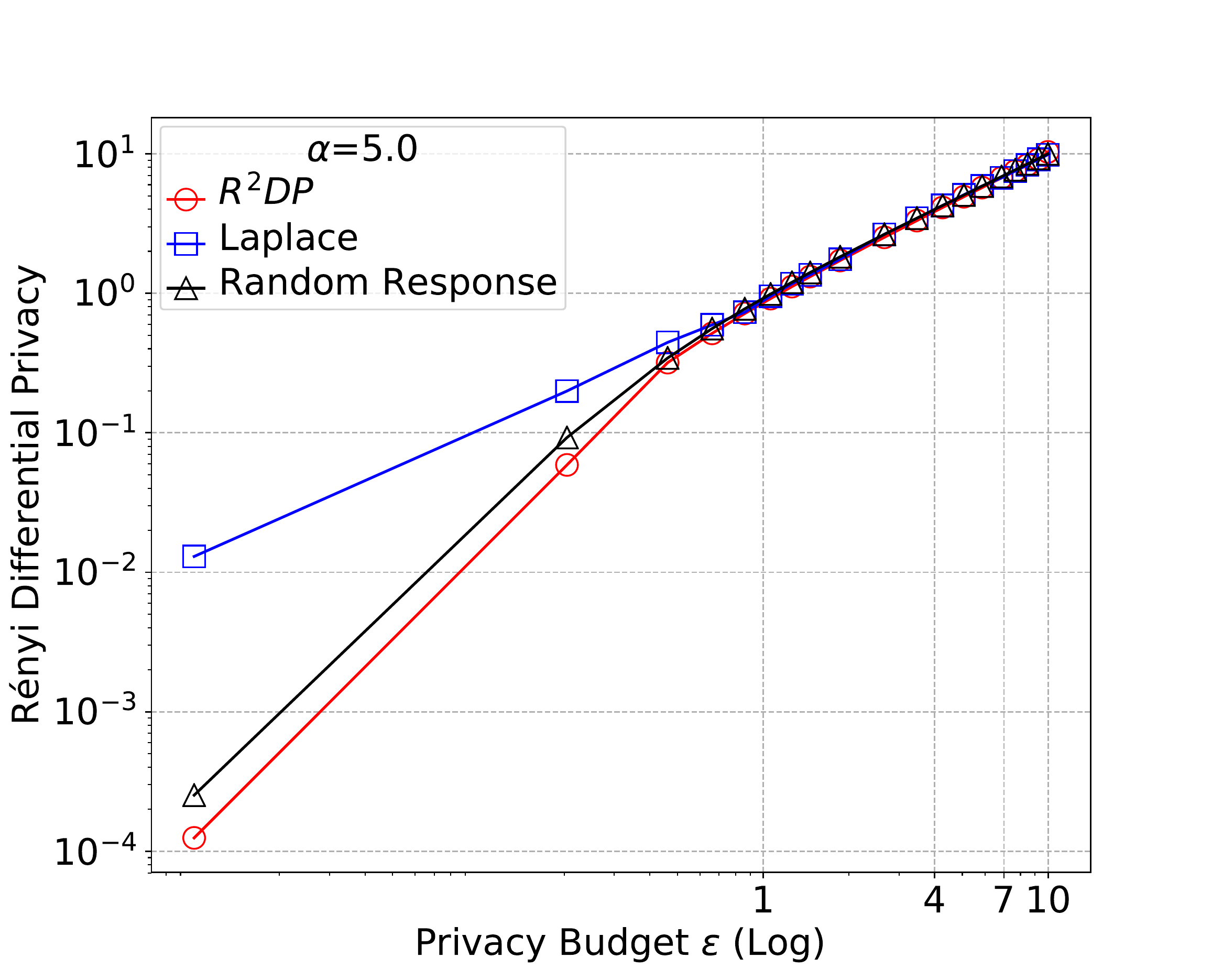} 	}\hspace{-0.25in}
	\subfigure[$\alpha=1$]{
		\includegraphics[angle=0, width=0.25\linewidth]{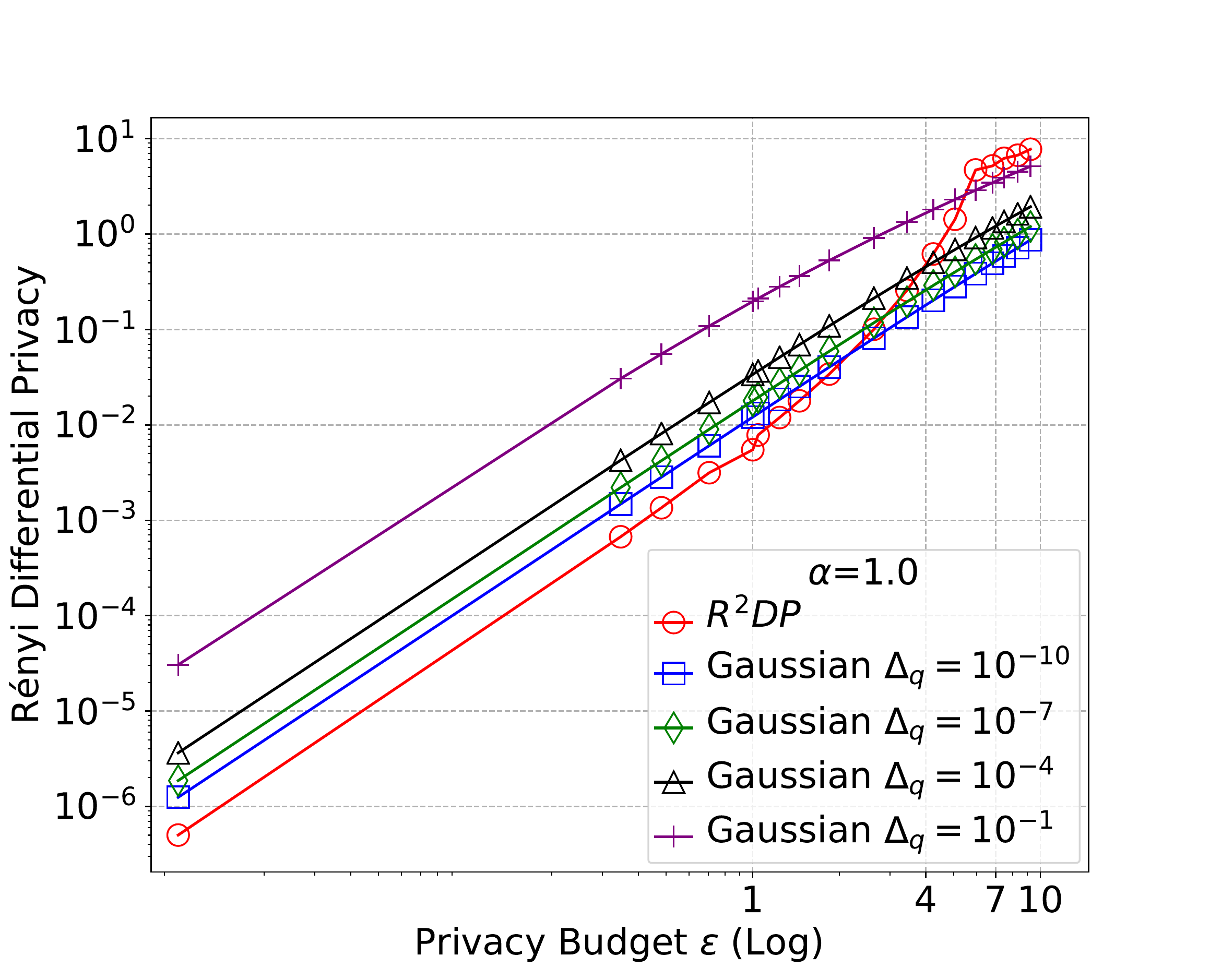}
		 }\hspace{-0.25in}
	\subfigure[$\alpha=3$]{
		\includegraphics[angle=0, width=0.25\linewidth]{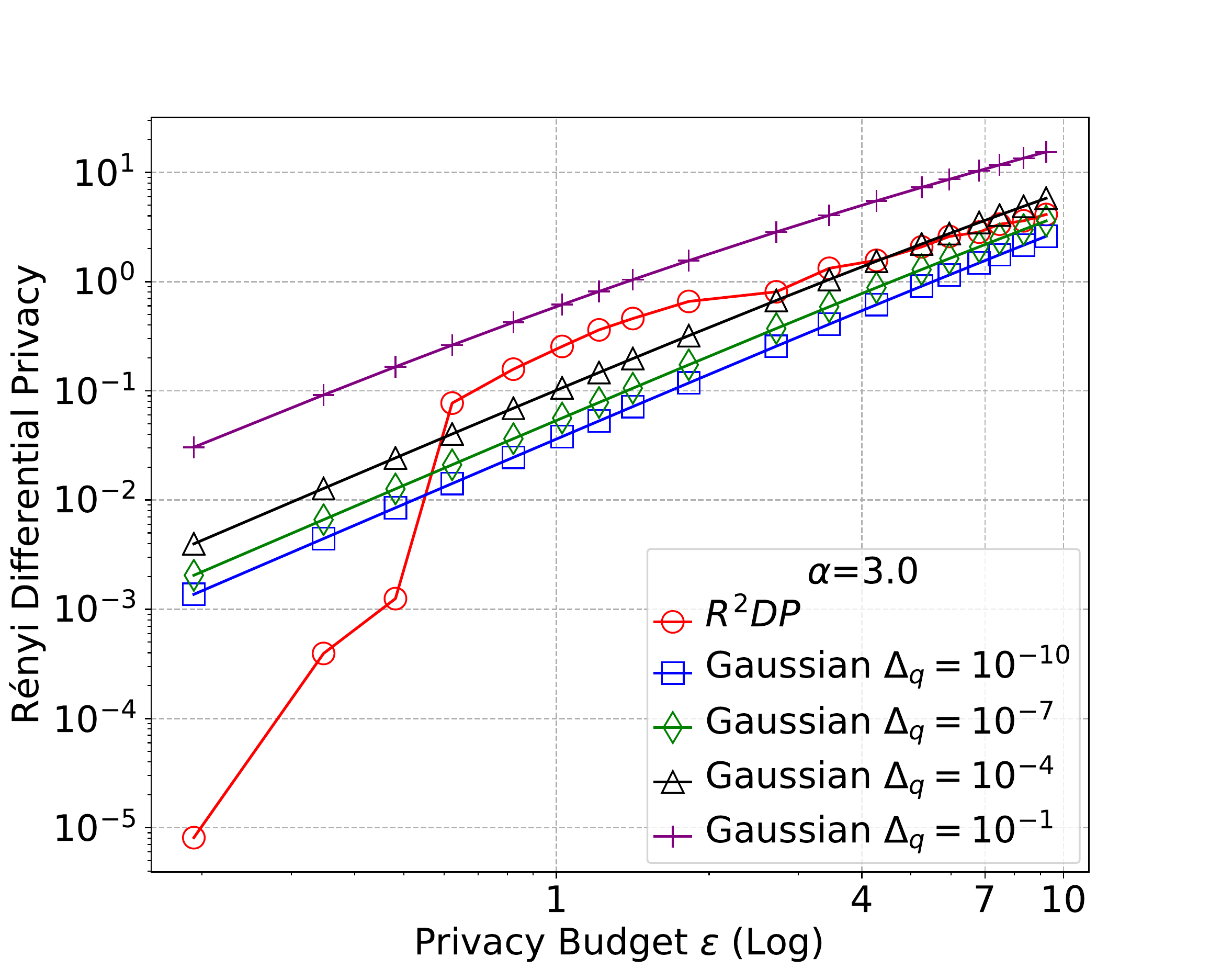} 		}\hspace{-0.25in}
\subfigure[$\alpha=4$]{
		\includegraphics[angle=0, width=0.25\linewidth]{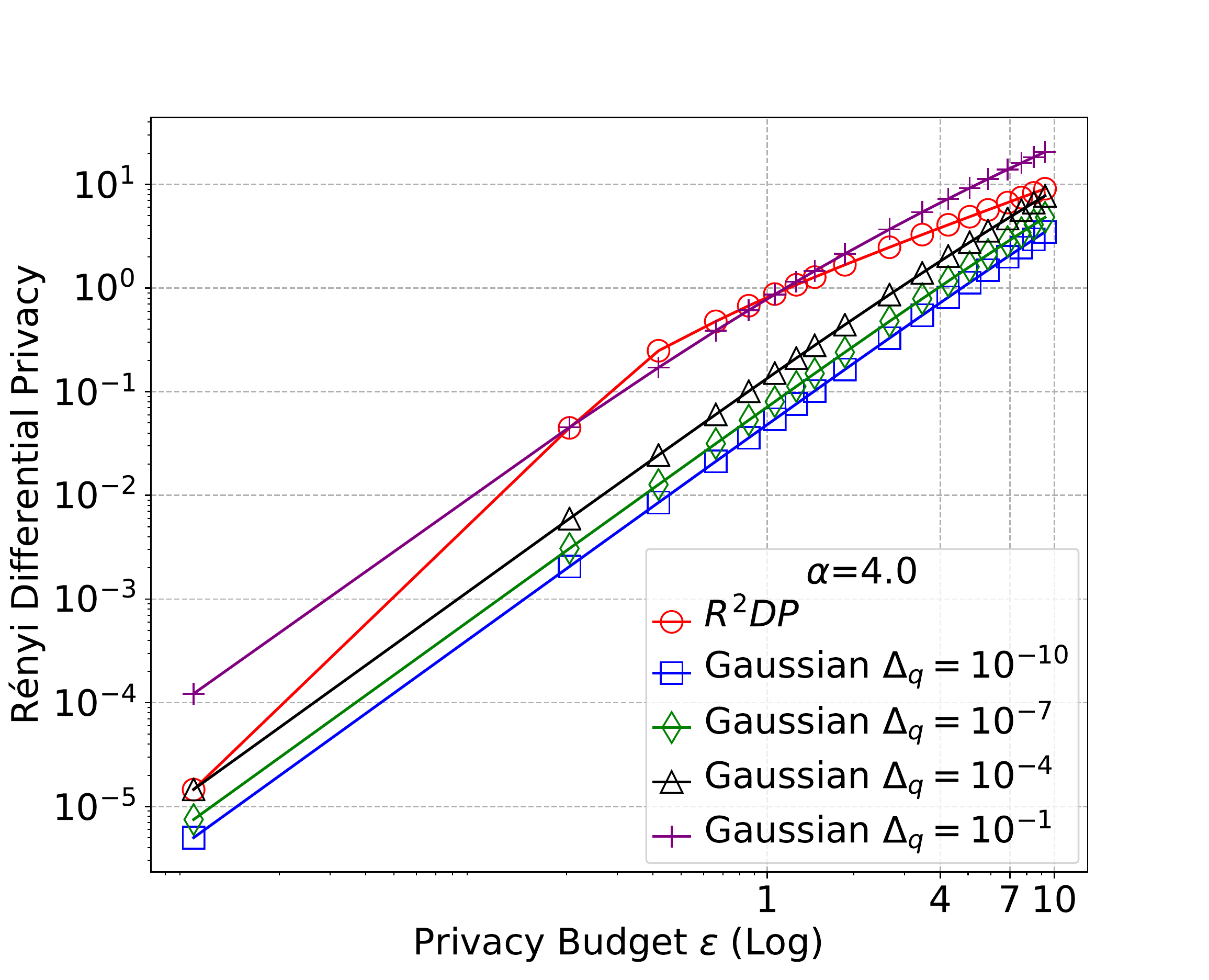}
		 }\hspace{-0.25in}
	\subfigure[$\alpha=5$]{
		\includegraphics[angle=0, width=0.25\linewidth]{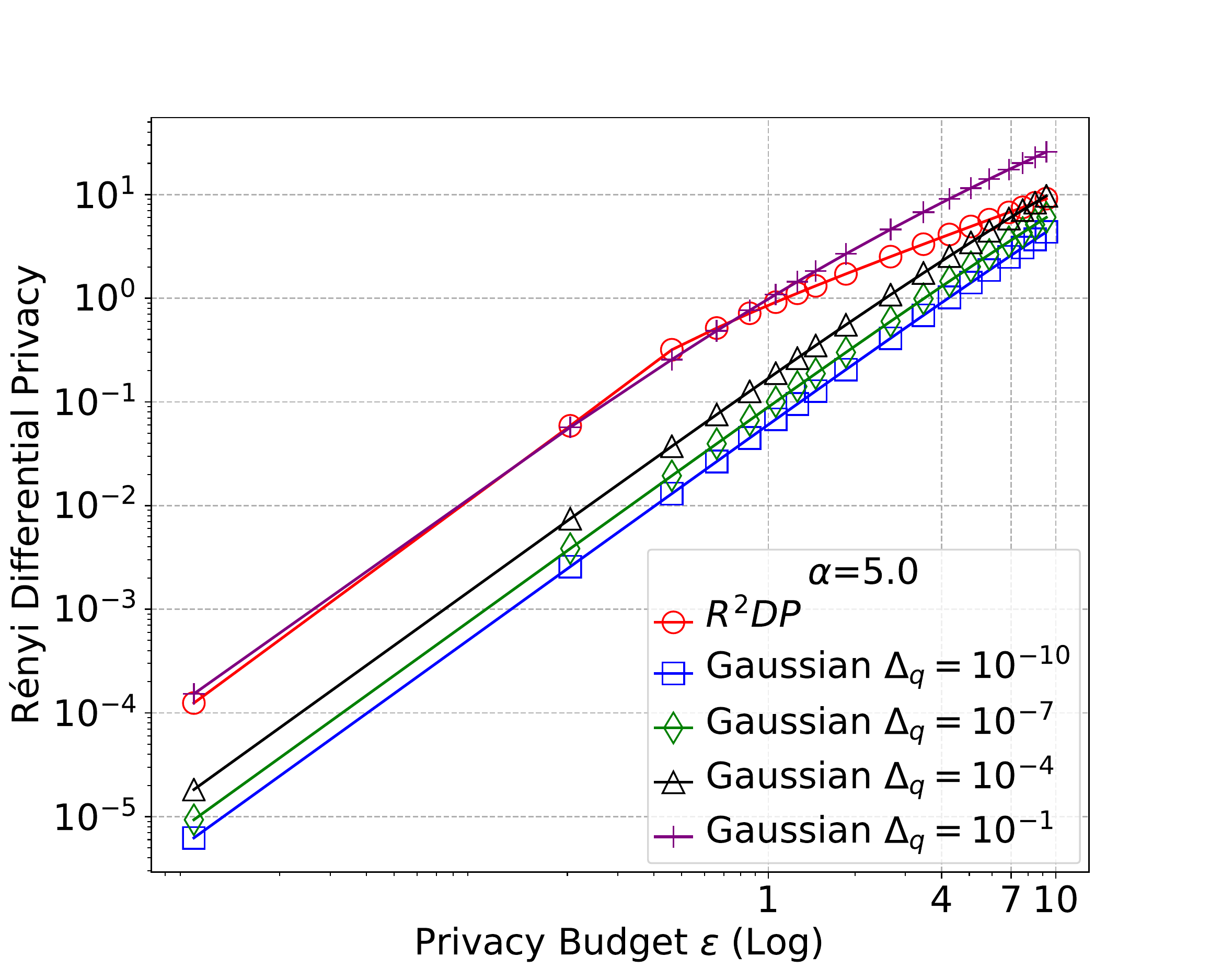}}
	\caption[Optional caption for list of figures]
	{R\'enyi Differential Privacy: (a-d) R$^2$DP  compared
to Laplace and Random Response mechanisms, and (e-h) R$^2$DP compared
to Gaussian mechanism.}
	\label{fig:RenyiDP}
\end{figure*}

\subsubsection{$\ell_1$ and $\ell_2$ Metrics}
\label{secl1l2}
We compare R$^2$DP with the baseline Laplace and Staircase mechanisms~\cite{Gupte1}, by varying the privacy budget $\epsilon$ and for four different sensitivities $\Delta q\in \{0.1, 0.5, 1, 1.5\}$. Our results validate the findings of Geng et al.~\cite{7353177}, i.e., in the low privacy regime ($\epsilon\rightarrow \infty$), the Staircase mechanism is optimal while in the high privacy regime ($\epsilon\rightarrow 0$), the Laplace mechanism is optimal. 

More importantly, our evaluations show that, for medium regime of
privacy budgets (which could be more desirable in practice), the class
of optimal noise can be totally different. In fact, as shown in Geng
et al.~\cite{7353177}, the lower-bound of $\epsilon$ at which the
Staircase distribution performs better than the Laplace distribution
is somewhere around $\epsilon=3$ for both $\ell_1$ and $\ell_2$
metrics. As illustrated in Figure~\ref{fig:newld}, in contrast to
$\ell_1$ metric (for which the results of laplace and staircase are
relatively tight), R$^2$DP can find a class of noises with
significantly improved $\ell_2$ metric for $\epsilon<3$ (a logarithmic
X axis is used to illustrate the performance in this region). The PDF
of this class of noises is mostly two-fold distributions with Laplace
distribution as the first fold, and Gamma distribution as the second
fold. This finding is in line with the optimal class of noise proposed
by Koufogiannis et al.~\cite{koufogiannis2015optimality}, i.e.,
$f(v)=\cfrac{\epsilon^n \Gamma
  (\frac{n}{2}+1)}{\pi^{\frac{n}{2}}\Gamma (n+1)} e^{-\epsilon ||v||_2
}$. Furthermore, our results suggest different classes of optimal
noises (than those found in the literature) for different parameters,
sensitivity, $\epsilon$ and $p$ (index of $\ell$ norm). In particular,
a larger $p$ tends to provide larger search spaces for R$^2$DP
optimization, which results in further improved results for
$\epsilon<3$ (Figure~\ref{fig:newld} (a,b,e,f) vs. (c,d,g,h)).

\subsubsection{Relative Entropy Metric}
As Wang et al~\cite{7039713} has already shown that the output entropy
of $\epsilon$-DP randomization mechanisms is lower bounded by
$1-\ln(\epsilon/2)$ (for count queries) and the optimal result is
achieved with Laplace mechanism, we focus our entropy metric
evaluation on relative entropy metrics, i.e., KL and R\'enyi
divergences. To define the prior distribution for this group of
experiments, we have created a histogram with $50$ bins of our data
and calculated the probability mass function (pmf) of the
bins. \footnote{2 millions records fall into 50 bins (e.g., equal
  range for each bin). Then, any counting and moving average query
  (with different sensitivities) can be performed within each of the
  50 bins to generate the distribution. Finally, the distance between
  the original and noisy distributions can be measured using the
  relative entropy metrics.} As illustrated in
Figure~\ref{fig:KLentropy_2}, we can draw similar observations for the KL entropy metric. In particular, we observe that R$^2$DP performs better for smaller sensitivity due to the larger
search space of PDFs used in optimization. Similarly the R\'enyi entropy depicted in Figure~\ref{fig:Renentropy_11}
shows a similar trend with different $\alpha$ (the index
of the divergence).

\subhead{Summary} The R$^2$DP mechanism can generate better results than most of the well-known distributions for utility metrics without known optimal distributions (e.g., usefulness), and our results asymptotically
approach to the optimal for utility metrics with known optimal distributions (e.g., $\ell_1$ and $\ell_2$). In particular, even though R$^2$DP is not specifically designed to optimize $\ell_1$ and $\ell_2$ metrics, we observe very similar performance between the
R$^2$DP results and the optimal Staircase results, e.g., the multiplicative gain compared to the Laplace results. We note that using a larger number of independent RVs drawn from different PDFs as the search space generator may further improve the results.

\subsection{Tightness of R$^2$DP under R\'enyi DP}
R\'enyi \DP~\cite{mironov2017renyi} is a recently proposed as a
relaxed notion of DP which effectively quantifies the bad outcomes in
($\epsilon,\delta$)-DP mechanisms and consequently evaluates how such
mechanisms behave under sequential compositions (see
Appendix~\ref{renyiDPsec} for details on R\'enyi DP). We now evaluate
how the privacy loss of R$^2$DP behaves under R\'enyi DP.

Specifically, this group of experiments are conducted to provide
insights about the privacy loss of R$^2$DP and other well-known
mechanisms. In particular, Figure~\ref{fig:RenyiDP} (a-d) depicts the
R\'enyi differential privacy of the R$^2$DP and two basic mechanisms
for counting queries: random response and Laplace mechanisms. These
results are based on the privacy guarantees depicted in
Table~\ref{tablerenyi}. Our results demonstrate that fine tuning
R$^2$DP can generate strictly more private results compared to the
other two $\epsilon$-DP mechanisms when the definition of the privacy
notion is relaxed. Furthermore, the level of such tightness depends on
the R\'enyi differential privacy index where a smaller value of
$\alpha$ pertains to a relatively tighter R$^2$DP mechanism. On the
other hand, all three mechanisms behave more similarly as $\alpha$
increases. Ultimately, at $\alpha \rightarrow \infty$, where R\'enyi
\DP becomes equivalent to the classic notion of $\epsilon$-DP, all
three mechanisms' privacy guarantees converge to $\epsilon$.

\begin{figure*}[!tb]
\centering
	\subfigure[$p=1, \Delta q=0.5$]{
		\includegraphics[angle=0, width=0.25\linewidth]{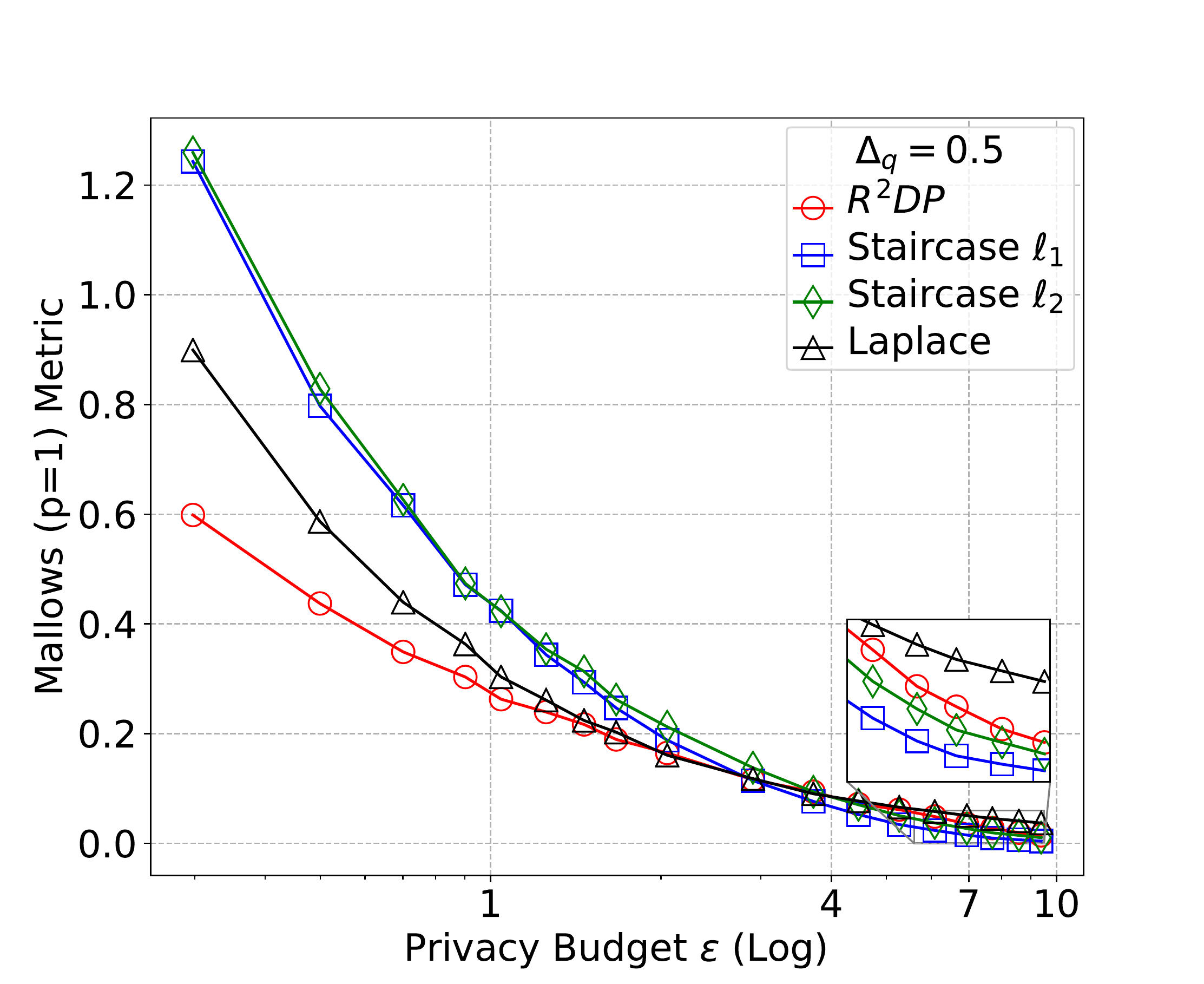}
		 }\hspace{-0.25in}
	\subfigure[$p=1, \Delta q=1$]{
		\includegraphics[angle=0, width=0.25\linewidth]{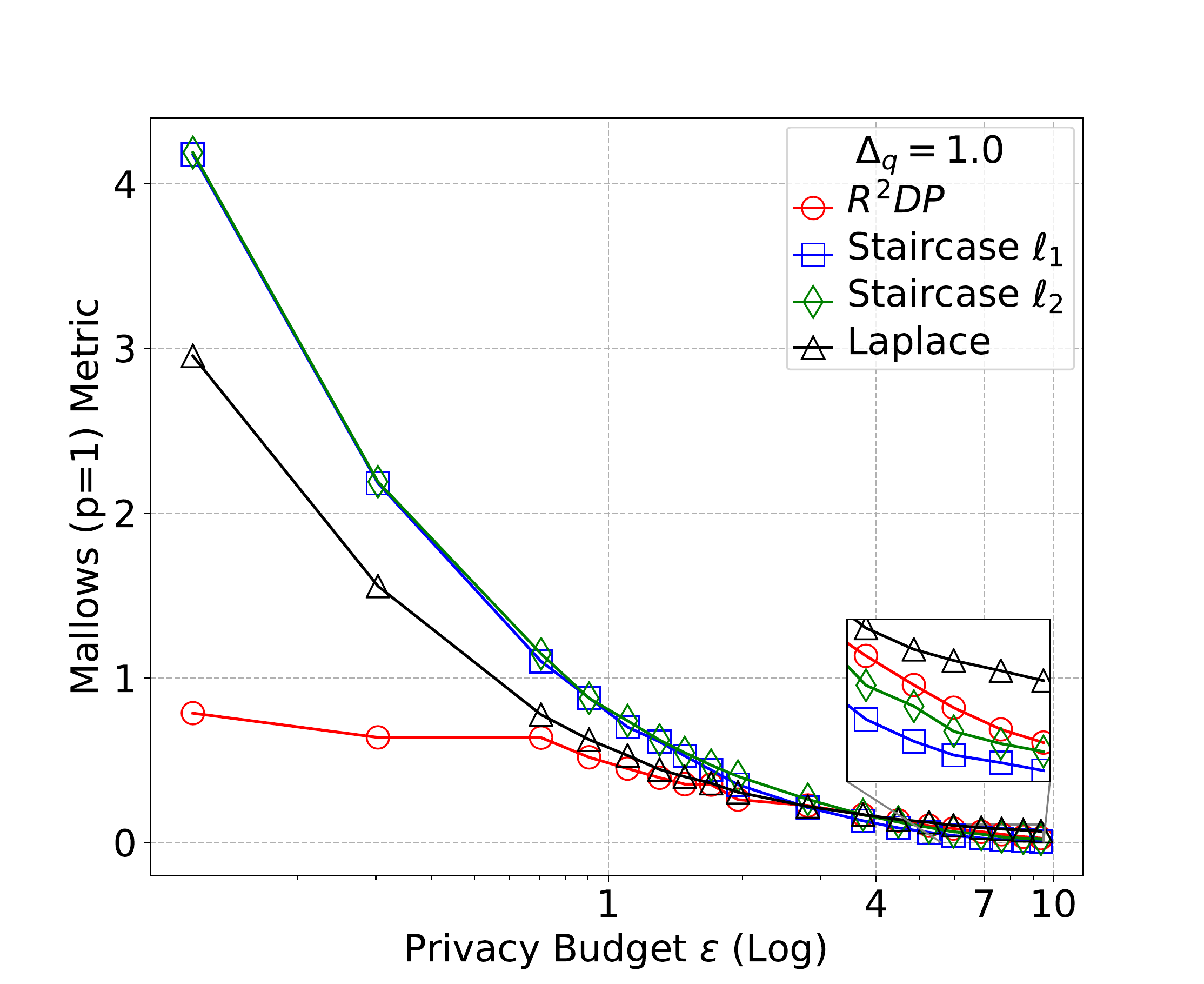} 		}\hspace{-0.25in}
\subfigure[$p=2, \Delta q=0.5$]{
		\includegraphics[angle=0, width=0.25\linewidth]{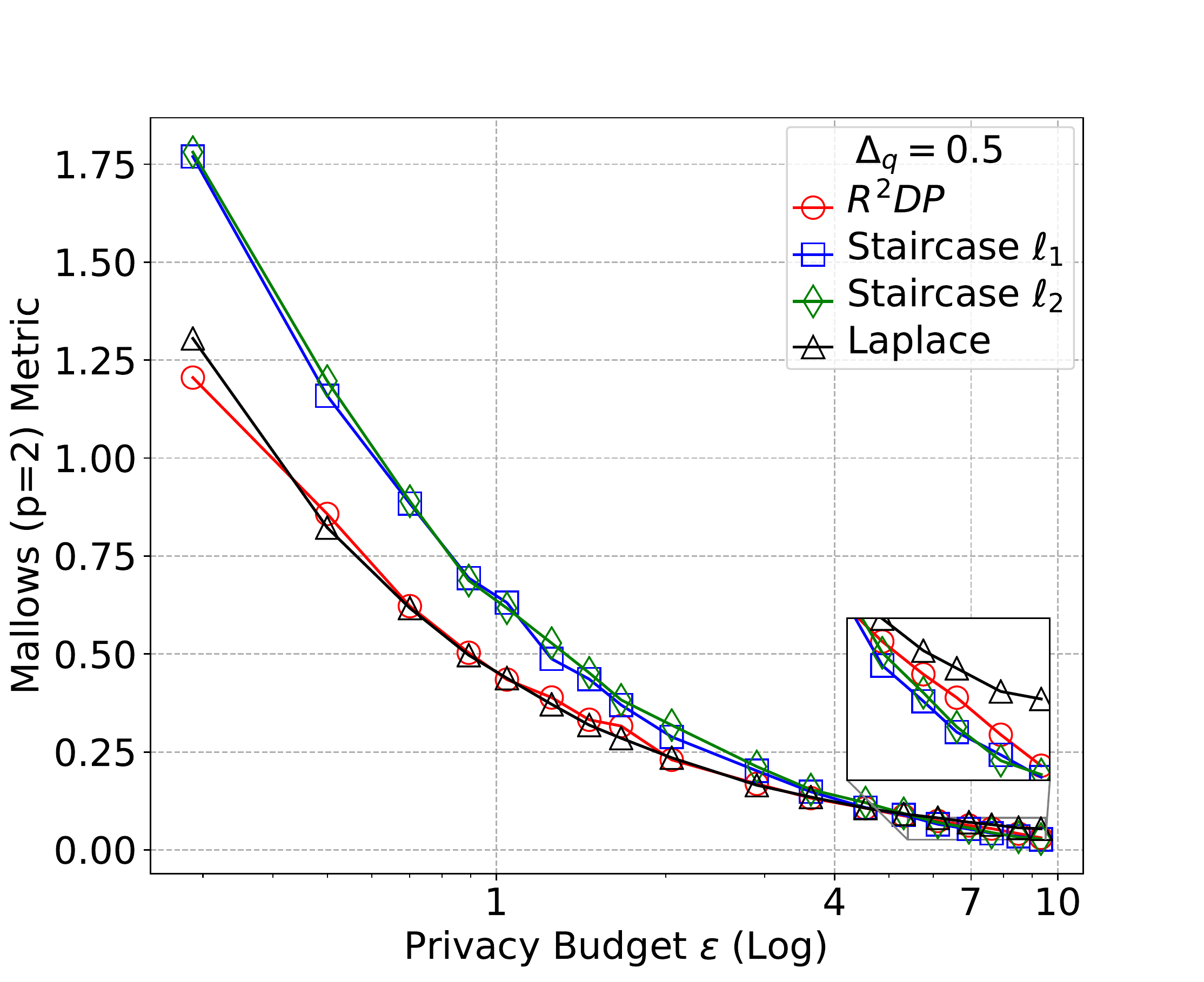}
		 }\hspace{-0.25in}
	\subfigure[$p=2, \Delta q=1$]{
		\includegraphics[angle=0, width=0.25\linewidth]{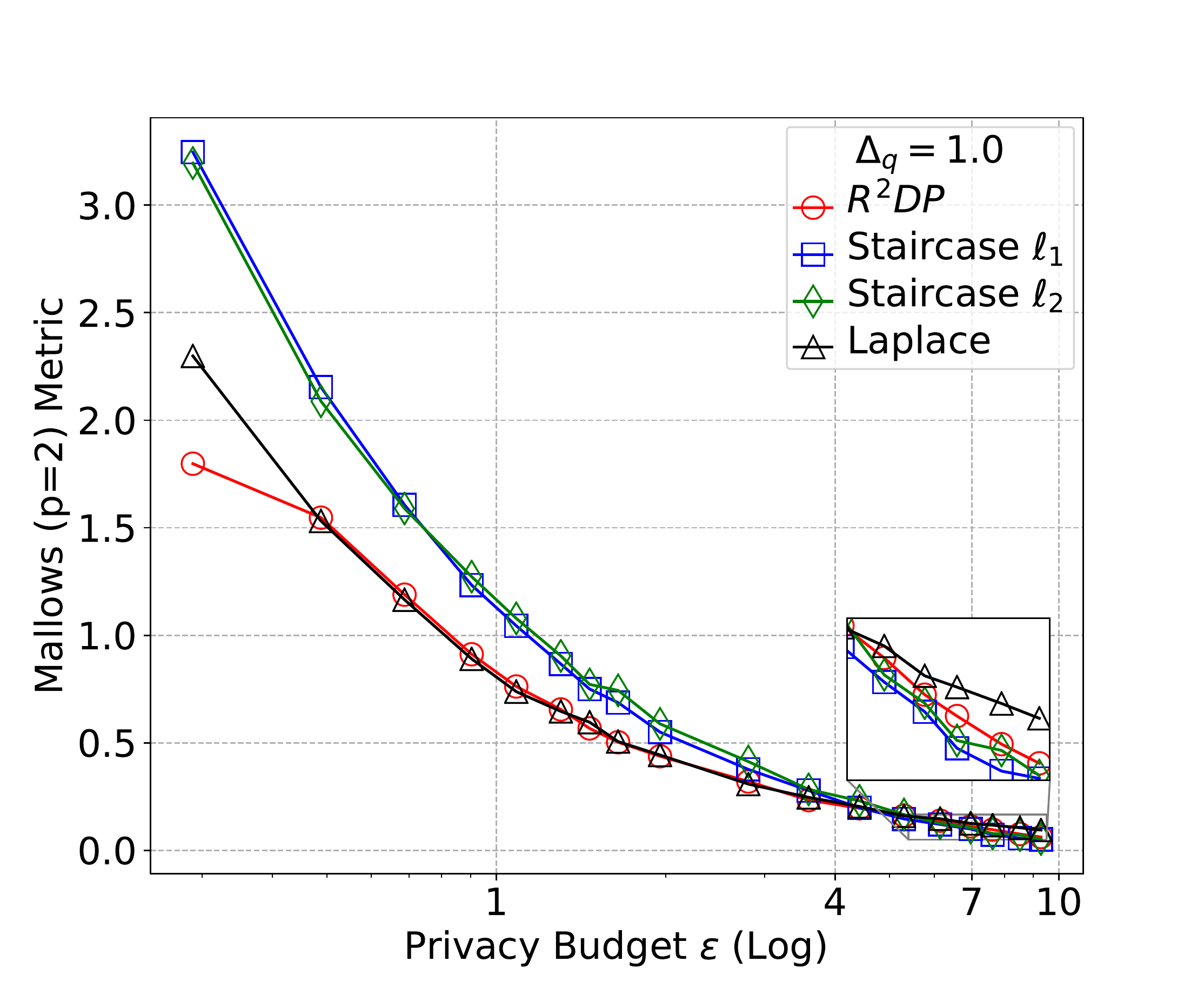}}
	\caption[Optional caption for list of figures]
	{Mallows metric: R$^2$DP compared
to Laplace and Staircase mechanisms for degree distribution (Facebook dataset).}
	\label{fig:mall}
\end{figure*}

\begin{figure*}[ht]
	\centering
	\subfigure[Precision vs. $\epsilon$ (UCI Adult)]
	{\includegraphics[angle=0, width=0.25\linewidth]{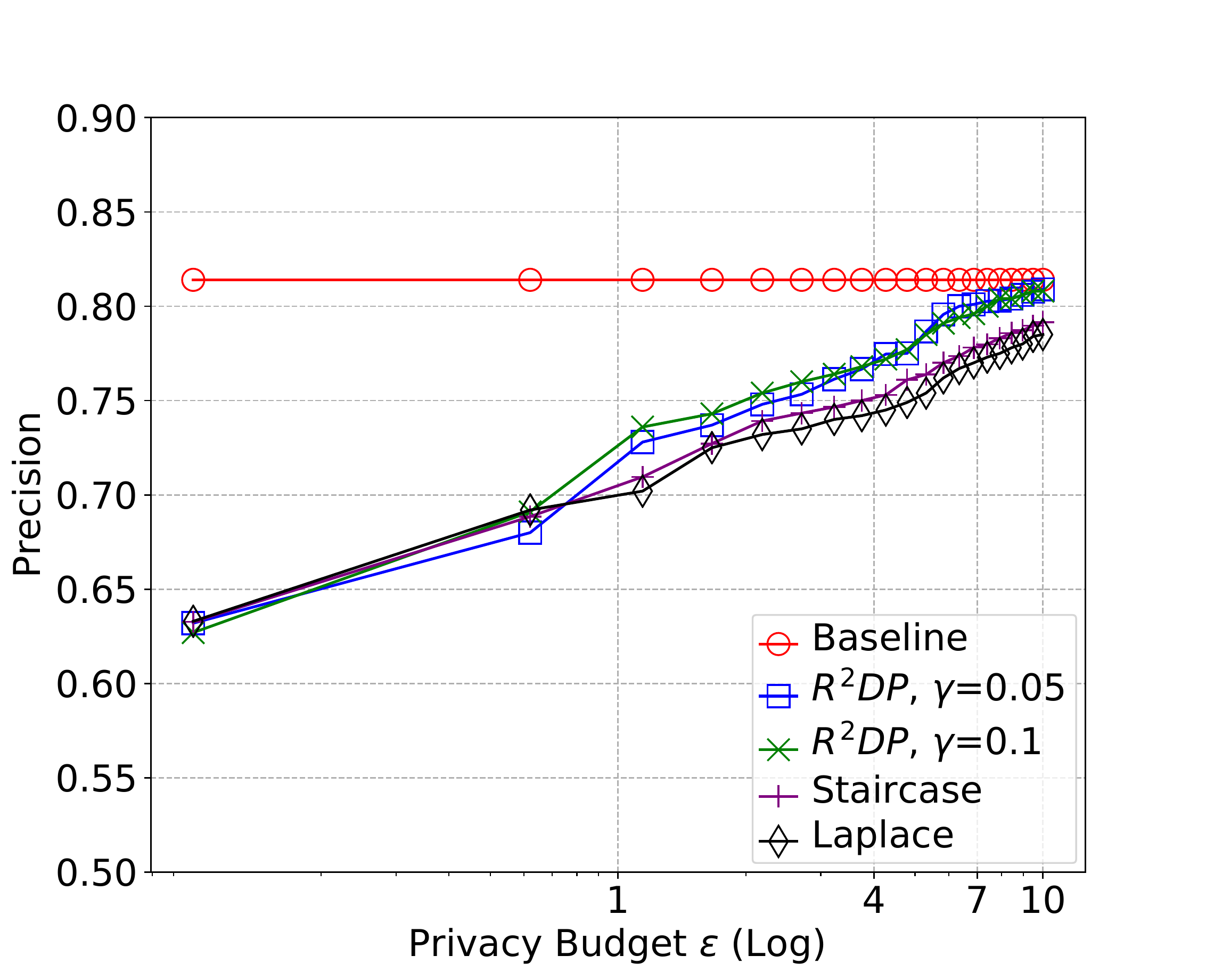}
		\label{fig:pre_eps} }\hspace{-0.25in}
	\subfigure[Recall vs. $\epsilon$ (UCI Adult)]{
		\includegraphics[angle=0, width=0.25\linewidth]{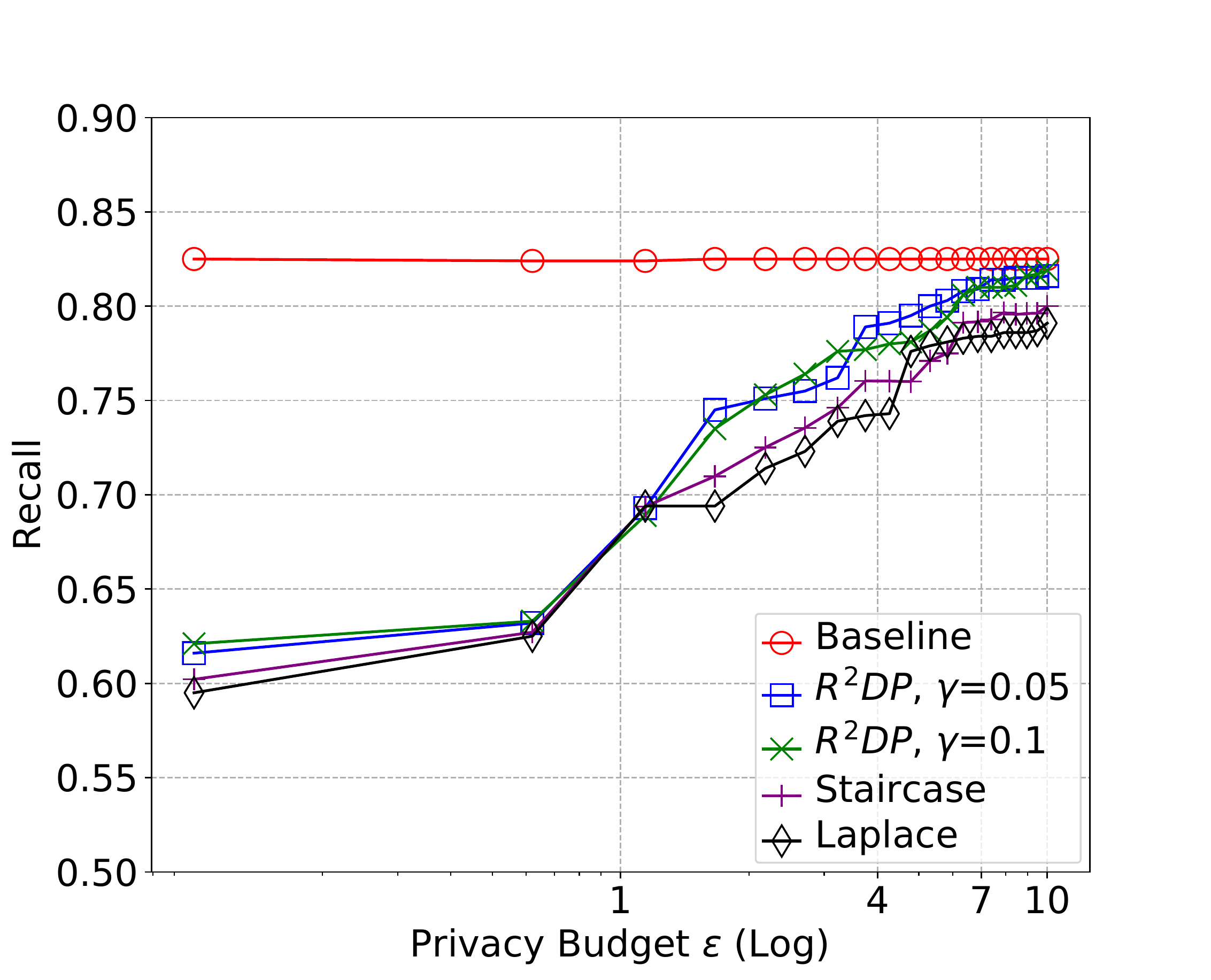} \hspace{-0.25in} \label{fig:pre_eps1}}
	\subfigure[Precision vs. $\epsilon$ (KDDCup99)]{
		\includegraphics[angle=0, width=0.25\linewidth]{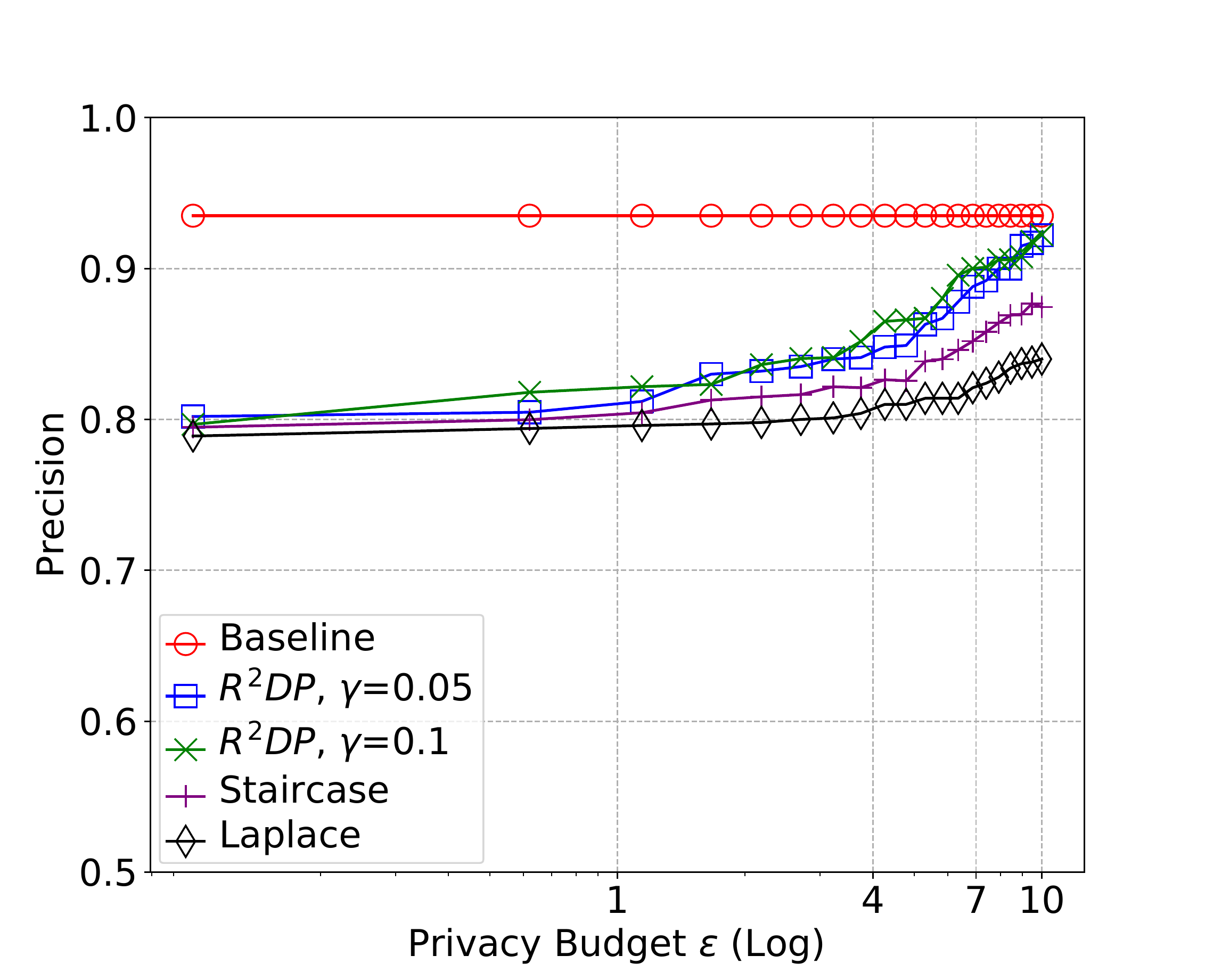}
		\label{kdd1} }\hspace{-0.25in}
	\subfigure[Recall vs. $\epsilon$ (KDDCup99)]{
		\includegraphics[angle=0, width=0.25\linewidth]{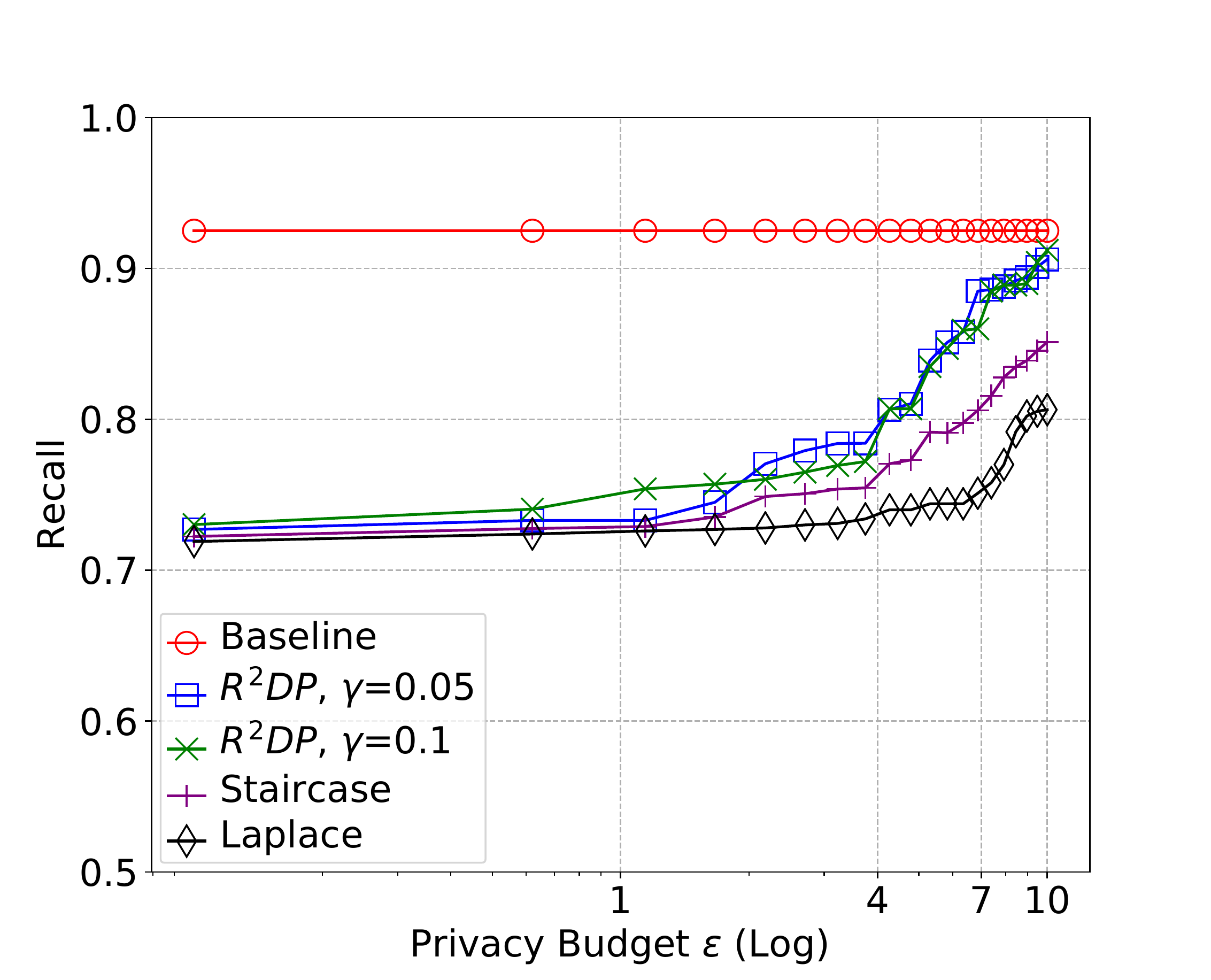}
		\label{kdd2} }
		\caption[Optional caption for list of figures]
	{Accuracy evaluation for classification (UCI Adult dataset and KDDCup99 dataset)}
	\label{fig:acc_size}
\end{figure*}

\begin{table}[!h]
\caption{Summary of R\'enyi DP parameters for four mechanisms based on Theorem~\ref{thmss}}
\centering
\begin{adjustbox}{width=0.5\textwidth,center}

\begin{tabular}{|c|c|c|}
\hline
\bf{Mechanism} & \bf{Differential Privacy} &  \bf{R\'enyi Differential Privacy for $\alpha$}
\\
\hline
\multirow{2}{*}{Laplace} & \multirow{2}{*}{$\frac{1}{b}$}& $\alpha >1:\frac{1}{\alpha-1} \log \left [ \frac{\alpha\cdot exp(\frac{\alpha-1}{b})+(\alpha-1) \cdot exp(\frac{-\alpha}{b})}{2\alpha-1} \right]$\\
& & $\alpha =1:\frac{1}{b}+exp(\frac{-1}{b})-1$\\
 \hline 
\multirow{2}{*}{Random Response} &  \multirow{2}{*}{$|\log \frac{p}{1-p}|$}& $\alpha >1: \frac{1}{\alpha-1} \log \left [ p^\alpha (1-p)^{1-\alpha} + p^{1-\alpha} (1-p)^{\alpha} \right]$\\
& &$\alpha =1:(2p-1)\log \frac{p}{1-p}$\\
\hline
\multirow{2}{*}{R$^2$DP}& \multirow{2}{*}{$M'_{\frac{1}{b}}(0)/M'_{\frac{1}{b}}(-1)$}&$\alpha >1:\frac{1}{\alpha-1} \log \left [ \frac{\alpha M_{\frac{1}{b}}(\alpha-1)+(\alpha-1) M_{\frac{1}{b}}(-\alpha)}{2\alpha-1} \right]$\\
& & $\alpha =1:M_{\frac{1}{b}}'(0)+M_{\frac{1}{b}}(-1)-1$\\
\hline
Gaussian & $\infty$&  $\cfrac{\alpha}{2\sigma^2}$ \\
\hline
\end{tabular}
\end{adjustbox}
	\label{tablerenyi}
\end{table}

In the next set of experiments, we compare the
R$^2$DP mechanism and Gaussian mechanism in terms of privacy guarantee
to understand how exactly the bad outcomes probability ($\delta$)
affects the privacy robustness of a privatized
mechanism. Figure~\ref{fig:RenyiDP} (e-h) gives such a
comparison. Specifically, since R\'enyi \DP at each $\alpha$ can be
seen as higher-order moments as a way of bounding the tails of the
privacy loss variable~\cite{mironov2017renyi}, we observe that each
value of $\alpha$ reveals a snapshot of such a privacy loss. As a
tangible observation, we conclude that the class of optimal
$\epsilon$-\DP mechanisms benefits from a very smaller privacy loss at
smaller moments (which are more decisive in overall protection) and
larger privacy loss at bigger moments.

\subsection{Social Network Analysis}

We conduct experiments to compare the performance of R$^2$DP, Laplace
and two staircase mechanisms based on PINQ queries in social network
analysis. Figure~\ref{fig:mall} compares the degree distribution for a
real Facebook dataset using Mallows metric (the prior, i.e.,$n=47,538$ nodes, and
$p=1$ or $2$ for computing the distribution distance using Mallows
metric). Again, our results confirm that  R$^2$DP  can
effectively generate PDFs to maximize this utility metric
suitable for social networking analysis. Note that, since the
definition of this metric is similar to $\ell_p$ metric (Mallows is
more empirical, depending on the number of nodes in the dataset), the
results for this metric display a similar pattern to those for
$\ell_p$ metric depicted in Figure~\ref{fig:newld}.

\subsection{Machine Learning}

We obtain our baseline results by applying the Naive Bayes classifier on the Adult dataset (45K training records and 5K testing records), the precision and recall results are derived as 0.814 and 0.825, respectively. Then, we evaluate the precision and recall of R$^2$DP and Laplace-based naive classification \cite{VaidyaSBH13} by varying the privacy budget for each PINQ query $\epsilon\in[0.1, 10]$ (sensitivity=1) where two different error bounds $\gamma=0.05, 0.1$ are specified for R$^2$DP. We have the following observations:

\begin{itemize}
    \item As shown in Figure \ref{fig:pre_eps} and \ref{fig:pre_eps1}, the R$^2$DP-based classification is more accurate than the Laplace and staircase mechanisms with the same total privacy budget for all the PINQ queries $\epsilon$. As the privacy budget $\epsilon$ increases, following our statistical query experiments, R$^2$DP offers a far better precision/recall compared to the Laplace-based classification (close to the results without privacy consideration) since it approaches to the optimal PDF.
    
    \item Among the precision/recall results derived with two different $\gamma$ in R$^2$DP-based classification, for each $\epsilon$, one out of the two specified error bounds (e.g., $\gamma=5\%$) may reach the highest accuracy (not necessarily the result with the smaller $\gamma$). 
     
     \item As shown in Figure \ref{kdd1} and \ref{kdd2}, we can draw similar observations from the KDDCup99 dataset. 
     
\end{itemize}

The above experimental results have validated the effectiveness of integrating R$^2$DP to improve the output utility for classification while ensuring $\epsilon$-differential privacy. In summary, all the experiments conducted in both statistical queries and real-world applications have validated the practicality of the R$^2$DP framework.

\section{Related Work}
\label{rel}
\label{relat}

Differential privacy~\cite{10.1007/11681878_14} is a model for preserving privacy while releasing the results of various useful functions, such as contingency tables, histograms and means~\cite{Dwork08differentialprivacy:}. Many existing works focus on improving the utility based on different mechanisms. 

\subhead{Noise Perturbation} Based on the general utility maximization framework from Ghosh et al.~\cite{Ghosh1}, Gupte and Sundararajan~\cite{Gupte1} further study the optimal noise probability distributions for single count queries. Later, Geng el al.~\cite{6875258,DBLP:journals/jstsp/GengKOV15} demonstrate the optimal noise distribution has a Staircase-shaped PDF for Laplace mechanism. Furthermore, Balle and Wang~\cite{DBLP:conf/icml/BalleW18} develop an optimal Gaussian mechanism in high privacy regime to minimize the noise and increase the utility for queries. Geng et al.~\cite{DBLP:journals/corr/abs-1809-10224} further show the optimal noise distribution is a uniform distribution over Gaussian mechanism. Moreover, Hardt et al.~\cite{Hardt:2010:GDP:1806689.1806786} study the privacy-utility trade-off for answering a set of linear queries over a histogram, where the error is defined as the worst expectation of the $\ell_2$-norm (identical to variance) of the noise among all possible outputs. Subsequently, Brenner et al.~\cite{5670945}
show that, for general query functions, no universally optimal DP mechanisms exist.

\subhead{Sampling and Aggregation}
Sampling and aggregation frameworks mostly split the database into
chunks, and aggregate the result using a DP
algorithm after querying each chunk~\cite{NissimRS07}. To
expand the applicability of output perturbation, 
Nissim et
al.~\cite{NissimRS07} propose a framework to formally analyze the
effect of instance-based noise. Observing the highly compressible
nature of many real-life data, researchers propose lossy compression
techniques to add noise calibrated to the compressed data.  
Acs et al.~\cite{DBLP:conf/icdm/AcsCC12} propose an optimization of Fourier
perturbation algorithm that clusters and exploits the redundancy between bins. Instead of directly adding noise to histogram counts, it first lossily compresses the data, then adds noise calibrated to the data. Li et al.~\cite{Li:2014:DWA:2732269.2732271} propose an algorithm to partitions a data domain into uniform regions and
adapts the strategy to fit the specific set of range queries to achieve a lower error rate. Zhang et al.~\cite{ZhangCXMX14} improve the clustering mechanism by sorting histogram bins based on the noisy counts.

\subhead{Data Composition}
Barak et al.~\cite{Barak:2007:PAC:1265530.1265569} propose
transforming the data into the Fouier domain, which could avoid the
violation of consistency for low-order marginals in database
tables. As efficiency is the main bottleneck for this approach when
the number of attributes is large, Hay et
al.~\cite{Hay:2010:BAD:1920841.1920970} ensure that the error rate
does not grow with the size of a database. The proposed hierarchical
histogram method also achieves a lower error for a fixed
domain. Different from one-dimensional datasets solution proposed by
Hay et al.~\cite{Hay:2010:BAD:1920841.1920970}, Xiao et
al.~\cite{XiaoWG10} propose \textit{Privelet} that improves accuracy
on datasets with arbitrary dimensions, which could reduce error to
25\% compared to 70\% as baseline error rate.  Cormode et
al.~\cite{6228069} apply \textit{quadtrees} and \textit{kd-trees} as
new techniques for parameter setting to improve the accuracy on
spatial data. Ding et al.~\cite{1fac2e61af8547b39fac891a12435f45}
introduce a general noise-control framework on data cubes.  Li et
al.~\cite{LiHRMM10} unify the two range queries over histograms into
one framework. Other techniques, such as principal component analysis
(PCA), linear discriminant analysis (LDA)~\cite{JiangJWMCO13},
and random projection~\cite{ChanyaswadLM19,XuRZQR17} are also used to
lower the data dimension for reducing the errors.
Cormode et al.~\cite{6228069} apply quadtrees (\textit{data-independent}) and kd-trees (\textit{data-dependent}) to add noise to a histogram output.

\subhead{Adaptive Queries}
In this technique, the improvement of utilities takes advantage of a
known set of queries, for example, Dwork et al.~\cite{DworkRV10}
propose \textit{Boosting for Queries} algorithm to obtain a better
accuracy of learning algorithms.  Hardt et
al.~\cite{HardtR10,HardtLM12} present multiplicative weights mechanism
to improve the efficiency of interactive queries. Instead of  polynomial
running time~\cite{DworkNRRV09}, this work achieves a
nearly linear running time with a relaxed utility requirement. Yuan et
al.~\cite{YuanZWXYH12,YuanZWXYH15} propose low-rank mechanism (LRM) to
further improve the adaptive queries. Other techniques such as
correlated noise~\cite{NikolovTZ13} and sparse vector technique
(SVT)~\cite{LyuSL17} are also used in adaptive queries.

\subhead{Applications}
Many researchers also work on improving the utility for different types of data, such as, the Fourier Perturbation Algorithm (FPA$_k$)~\cite{RastogiN10} in time-series data (e.g., location traces, web history, and personal health), \textit{kd-trees} on spatial data~\cite{6228069}, and matrix-valued query~\cite{ChanyaswadDPM18}.

\subhead{Summary}
Our R$^2$DP framework provides a complementary approach to  those existing works by providing the opportunity of searching for the maximal utility along an extra dimension. This framework also enables data recipients to specify their utility requirements and the computed parameter could be incorporated into  existing solutions to further improve utility.

\section{Conclusion}
\label{sec:conclusion}
This paper has proposed the R$^2$DP framework as a universal
solution for optimizing a variety of utility metrics requested in different applications. It can automatically identify a distribution that yields near-optimal utility, and hence is more practical for emerging applications. Specifically, we have shown that a differentially private mechanism could be defined based on a
random variable which is itself distributed according to some
parameterized distributions. We have also shown that such a mechanism
could explicitly take into account both the privacy requirements and
the utility requirements specified by the data owner and data
recipient, respectively. We have formally analyzed the privacy guarantee of R$^2$DP based on the well-known Laplace mechanism and formally proved the improvement of
utility over the baseline Laplace mechanism. Furthermore,
we discuss the potential of applying R$^2$DP to advanced algorithms. Finally, our experimental results based on
six different utility metrics for statistical queries, machine learning
and social network, as well as one privacy metric, have demonstrated
that R$^2$DP could significantly improve the utility of differentially
private solutions for a wide range of applications.
\section{Acknowledgements}

We thank the anonymous reviewers for their valuable comments and suggestions. This work is partially supported by the Natural Sciences and Engineering Research Council of Canada and Ericsson Canada under the Industrial Research Chair (IRC) in SDN/NFV Security. It is also partially supported by the National Science Foundation under Grant No. CNS-1745894.


\appendix

\section*{Appendix}

\section{Demonstration of Theorem~\ref{thm: RPLap mech}}
\label{demonst}
A Laplace distribution is of a $(\propto x\cdot e^{x\cdot t})$ order, where $x$ is the inverse of the scale parameter. Second, since $x\cdot e^{x\cdot t}=\diff{e^{x\cdot t}}{t}$, the cumulative distribution function (CDF) resulted from randomizing $x$ can be expressed in terms of the expectation $\mathbb E(e^{x\cdot t})$. We note that from now on, we will simply refer to R$^2$DP with Laplace distribution as the first fold PDF as \textit{the R$^2$DP mechanism}. 
\begin{figure}[!h]
\centering
\includegraphics[width=0.9\linewidth]{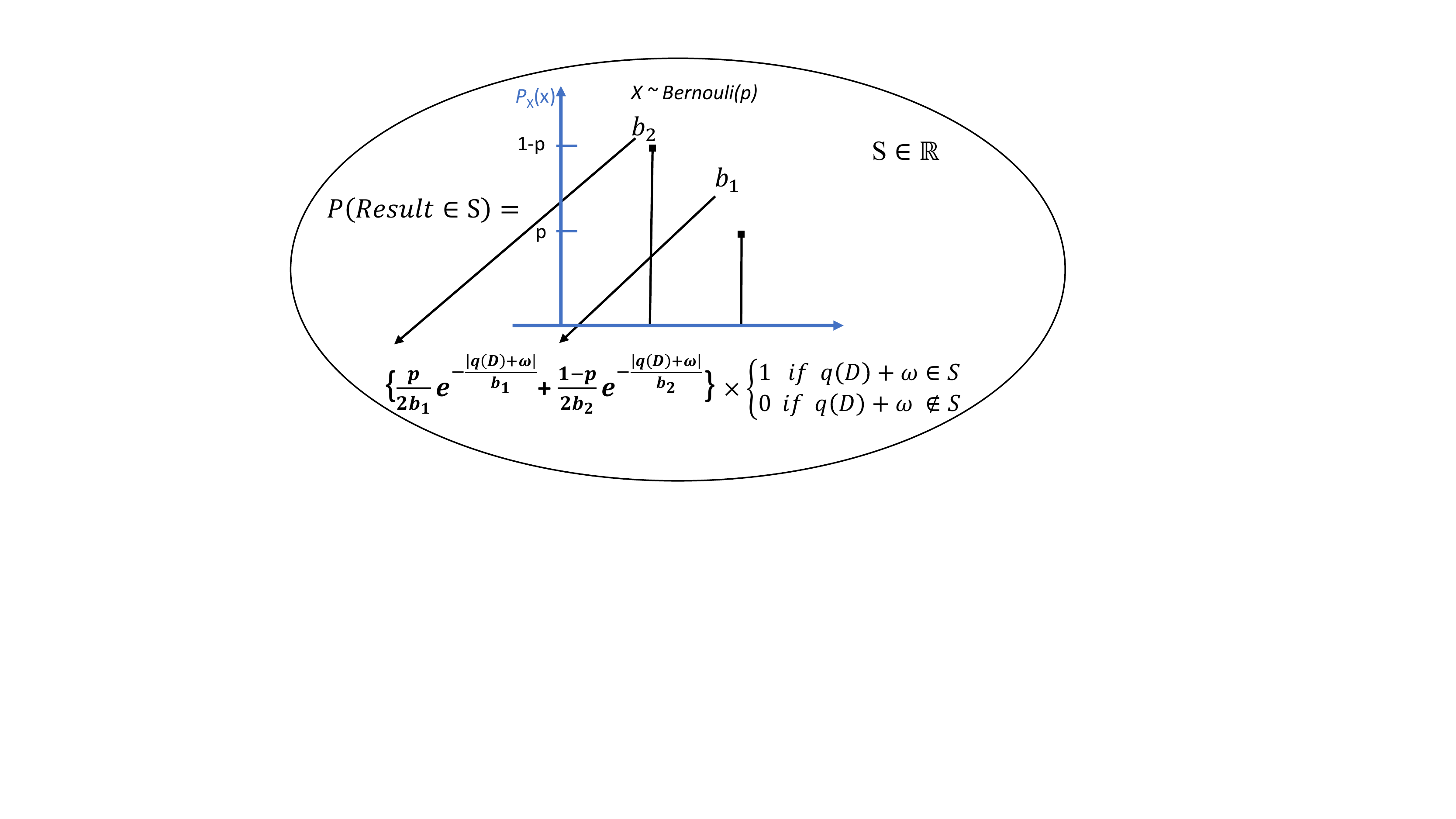}
\caption{The term in the parenthesis is the derivative of $\mathbb E(e^{\frac{1}{b}\cdot -|w|})$ w.r.t. $-|w|$, and hence the
above probability can be expressed in terms of the expectation}
\label{Fig:moment}
\end{figure}
\begin{exmp}
\label{exmp1}
Following Example~\ref{exmprpdp}, for a Bernoulli distributed scale parameter $b$, Figure~\ref{Fig:moment} illustrates the above finding (see Appendix~\ref{proofexamp} for proof). It can be verified that the term inside the braces is the derivative
of $\mathbb {E}(e^{\frac{1}{b}\cdot -|w|})$ w.r.t. $-|w|$, and hence the
above probability can be expressed in terms of the expectation.
\end{exmp}

\section{Case Study PDFs}
\label{cases}

\subsubsection{Discrete Probability Distributions}

First, we consider two different mixture Laplace distributions that
can be applied for constructing R$^2$DP with discrete
probability distribution $f_b$.

\vspace{0.05in}

(1) \textbf{Degenerate distribution.} A degenerate distribution is a
probability distribution in a (discrete or continuous) space with
support only in a space of lower
dimension~\cite{bremermann1965distributions}. If the degenerate
distribution is uni-variate (involving only a single random variable),
it will be a deterministic distribution and takes only a single
value. Therefore, the degenerate distribution is identical to the
baseline Laplace mechanism as it also assigns the mechanism one single
scale parameter $b_0$. Specifically, the probability mass function of
the uni-variate degenerate distribution is:
   \[ f_{\delta,k_0}(x)= \begin{cases} 
      1 & x= k_0  \\
      0 & x\neq k_0 
   \end{cases}
\]
 The MGF for the degenerate distribution $\delta_{k_0}$ is given by
 $M_k(t)=e^{t\cdot k_0}$~\cite{bulmer1979principles}. Using
 Equation~\ref{simple DPeq}, Theorem~\ref{degenerateDP} gives the
 same DP guarantee as the baseline Laplace mechanism.
 
 \begin{thm}
 \label{degenerateDP}
 The R$^2$DP mechanism $M_q(d,\epsilon)$, $\epsilon \sim f_{\delta,\frac{1}{b_0}}(\epsilon)$, is $\frac{\Delta q}{b_0}$-differentially private. 
 \end{thm}
 
Obviously, this distribution does not improve the bound in
Theorem~\ref{thm: RPLap mechut} but shows the
soundness of our findings.

\vspace{0.05in}

(2) \textbf{Bernoulli distribution.}
The probability mass function of this distribution, over possible outcomes $k$, is
\[ f_{B}(k;p)=\begin{cases}p&{\text{if }}k=1,\\1-p&{\text{if }}k=0.\end{cases}\]
Note that the binary outcomes $k=0$ and $k=1$ can be mapped to any two outcomes $X_0$ and $X_1$, respectively. Therefore, we consider the following Bernoulli outcomes
\[ f_{B,X_0, X_1}(X;p)=\begin{cases}p&{\text{if }}X=X_1,\\1-p&{\text{if }}X=X_0.\end{cases}\]
The MGF for Bernoulli distribution $f_{B,X_0, X_1}(X;p)$ is  $M_X(t)=p\cdot e^{t\cdot X_0}+(1-p)\cdot e^{t\cdot X_1}$~\cite{bulmer1979principles}. We now derive the precise \DP guarantee of an R$^2$DP mechanism with its scale parameter randomized according to a Bernoulli distribution. 
\begin{thm}
\label{bernoulidp}
 The R$^2$DP mechanism $M_q(d,\epsilon)$, $\epsilon\sim f_{B,\frac{1}{b_0},\frac{1}{b_1}}(\epsilon;p)$, satisfies $\ln [p\cdot e^{\frac{\Delta q}{b_0}}+(1-p)\cdot e^{\frac{\Delta q}{b_1}}]$ differential privacy. 
\end{thm}
This bound is exactly the mean of $e^{\epsilon(b)}$ given in Theorem~\ref{thm: RPLap mechut}.

\subsubsection{Continuous Probability Distributions}

We now investigate three compound Laplace distributions.

\vspace{0.05in}
(1) \textbf{Gamma distribution.} The gamma distribution is a
two-parameter family of continuous probability distributions with a
shape parameter $k>0$ and a scale parameter $\theta$. Besides the
generality, the gamma distribution is the maximum entropy probability
distribution (both w.r.t. a uniform base measure and w.r.t. a $1/x$ base measure) for a random variable $X$ for which
$\mathbb E (X) = k \theta = \alpha/\beta$ is fixed and greater than
zero, and $\mathbb E[\ln(X)] = \psi(k) + \ln(\theta) =
\psi(\alpha)-\ln(\beta)$ is fixed ($\psi$ is the digamma
function). Therefore, it may provide a relatively higher
privacy-utility trade-off in comparison to the other
candidates~\cite{zwillinger2002crc,jambunathan1954some}. A random
variable $X$ that is gamma-distributed with shape $\alpha$ and rate
$\beta$ is denoted by $X\sim \Gamma(k,\theta)$ and the corresponding
PDF is
\begin{equation*}
    f_{\Gamma}(X;k,\theta)
    {\displaystyle={\frac {x^{k -1}e^{-\frac{x}{\theta}}}{\Gamma (k)\cdot \theta^k}}\quad {\text{ for }}X>0{\text{ and }}k ,\theta >0,} 
\end{equation*}
where $\Gamma ( \alpha )$ is the gamma function. We now investigate the \DP guarantee provided by assuming that the reciprocal of the scale parameter $b$ in Laplace mechanism is distributed according to the gamma distribution (see Appendix~\ref{proofgama} for the proof).
\begin{thm}
\label{gamadist}
 The R$^2$DP mechanism $M_q(d,\epsilon)$, $\epsilon \sim f_{\Gamma}(\epsilon;k,\theta)$, satisfies $\big((k+1)\cdot \ln (1+\Delta q \cdot \theta)\big)$ differential privacy. 
\end{thm}
We now apply the necessary condition given in Equation~\ref{necc} (see Appendix~\ref{lemgama} for the proof).
\begin{lem}
\label{necsgama}
   R$^2$DP using Gamma distribution can satisfy the necessary condition in Equation~\ref{necc}.
\end{lem}
Therefore, Gamma distribution may improve over the baseline, and this can be
computed by optimizing the privacy-utility trade-off using the Lagrange multiplier
function in Equation~\ref{lagrange1}. Also, our numerical results show that, this distribution is more effective for large $\epsilon$ (weaker privacy guarantees). 

\vspace{0.05in}
(2) \textbf{Uniform distribution.}
In probability theory and statistics, the continuous uniform distribution or rectangular distribution is a family of symmetric probability distributions such that for each member of the family, all intervals of the same length on the support of the distribution are equally probable. The support is defined by the two parameters, $a$ and $b$, which are the minimum and maximum values. The distribution is often abbreviated as $U(a,b)$, which is the maximum entropy probability distribution for a random variable $X$ under no constraint; other than that, it is contained in the distribution's support~\cite{zwillinger2002crc,jambunathan1954some}. The MGF for $U(a,b)$ is \[ M_X(t)=\begin{cases}\frac{e^{tb}-e^{ta}}{t(b-a)}&{\text{for }}t\neq 0,\\1&{\text{for }}{\text{for }}t=0.\end{cases}\] Using Theorem~\ref{simple DP}, we now drive the precise \DP guarantee of an R$^2$DP mechanism for uniform distribution $U(a,b)$. 

\begin{thm}
\label{uniformdist}
 The R$^2$DP mechanism $M_q(d,\epsilon)$, $\epsilon \sim f_{U(a,b)}(\epsilon)$, is $\ln \big[\frac{\alpha^2-\beta^2}{2((1+\beta)e^{-\beta}-(1+\alpha)e^{-\alpha})} \big]$-differentially private, where $\alpha=a\cdot \Delta q$ and $\beta=b\cdot \Delta q$. 
\end{thm}
We now apply the necessary condition given in Equation~\ref{necc}. One can easily verify that the inequality holds for an infinite number of settings, e.g., $a=0.5$, $b=9$ and $\Delta q=1.2$.
\begin{lem}
\label{necsuniform}
   R$^2$DP using uniform distribution can satisfy the necessary condition in Equation~\ref{necc}.
\end{lem}
Therefore, R$^2$DP using uniform distribution may improve over the baseline, and this can be computed by optimizing the privacy-utility trade-off using the Lagrange multiplier function in Equation~\ref{lagrange1}. Also, our numerical results show that, this distribution can also be effective for both small and large $\epsilon$.

\vspace{0.05in}

(3) \textbf{Truncated Gaussian distribution.} The last distribution we
consider is the Truncated Gaussian distribution. This
distribution is derived from that of a normally distributed random
variable by bounding the random variable from either below or above
(or both). Therefore, we can benefit from the numerous useful
properties of Gaussian distribution, by truncating the
negative region of the Gaussian distribution. Suppose $X\sim \mathcal
N(\mu ,\sigma ^{2})$ has a Gaussian distribution and lies within the
interval $X\in (a,b),\;-\infty \leq a<b\leq \infty$. Then, $X$
conditional on $a<X<b$ has a truncated Gaussian distribution with the
following probability density function

\vspace{-0.15in}

\begin{equation*}
    f_{\mathcal N^T}(X;\mu ,\sigma,a,b)
    {\displaystyle={\frac {\phi(\frac{X-\mu}{\sigma})}{\sigma\cdot \big(\Phi(\frac{b-\mu}{\sigma})-\Phi(\frac{a-\mu}{\sigma})\big)}}\quad {\text{for }}a\leq x\leq b}
\end{equation*}
and by $f_{\mathcal N^T}=0$ otherwise. Here, $\phi(x)=\frac{1}{\sqrt{2\pi}\cdot} e^{-\frac{x^2}{2}}$ and $\Phi(x)=1-Q(x)$ are PDF and CDF of the standard Gaussian distribution, respectively.
Next, using Theorem~\ref{simple DP}, we give the \DP guarantee provided by the mechanism assuming that the reciprocal of $b$ is distributed according to the truncated Gaussian distribution.
\begin{thm}
\label{truncateddist}
 The R$^2$DP mechanism $\mathcal M_q(d,\epsilon)$, $\epsilon \sim f_{\mathcal N^T}(\epsilon;\mu ,\sigma,a,b)$, satisfies $\epsilon_{N^T}$- differential privacy, where 
 \begin{align}
    \epsilon_{N^T}=\ln\left[\cfrac{\mu+\cfrac{\sigma\cdot(\phi(\alpha)-\phi(\beta))}{(\Phi(\beta)-\Phi(\alpha))}}{\diff{M_{N^T}(t)}{t}|_t=-\Delta q}\right]
 \end{align}
 in which $\phi(\cdot)$ is the probability density function of the standard normal distribution, $\phi(\cdot)$ is its cumulative distribution function and $\alpha=\frac{a-\mu}{\sigma}$ and $\beta=\frac{b-\mu}{\sigma}$.
\end{thm}
\begin{lem}[see Appendix~\ref{label:trunclem} for the proof]
\label{necstrunc}
  R$^2$DP using truncated Gaussian distribution can satisfy the necessary condition in Equation~\ref{necc}.
\end{lem}

Therefore, truncated Gaussian distribution may improve over the baseline, and this can be computed by optimizing the privacy-utility trade-off using the Lagrange multiplier
function in Equation~\ref{lagrange1}. In particular, our numerical results show that, this distribution can also be effective for smaller $\epsilon$ (stronger privacy guarantees). 

\section{Proofs}
\label{proofsec}
\begin{proof}[Example~\ref{exmp1}]
\label{proofexamp}
Following Example~\ref{exmprpdp}, for a Bernoulli distributed scale parameter $b$, we have
\begin{eqnarray*}
\label{eqn10}
&\hspace{-2cm}\Prob(\mathcal M_q(d,b)\in S)\nonumber\\
&\hspace{-0.5cm}= \displaystyle\int_{\mathbb R} \frac{p}{2b_1} \cdot \mathds{1}_S\{q(d)+w\}  e^{\frac{-|w|}{b_1}} +\frac{1-p}{2b_2} \cdot \mathds {1}_S\{q(d)+w\}  e^{\frac{-|w|}{b_2}}dw  \nonumber  
\\
&\hspace{-0.5cm}=  \displaystyle\int_{\mathbb R} \big(\frac{p}{2b_1} \cdot e^{\frac{-|w|}{b_1}} +\frac{1-p}{2b_2} \cdot  e^{\frac{-|w|}{b_2}} \big) \mathds{1}_S\{q(d)+w\}dw   \nonumber \\
\end{eqnarray*}
where $\mathds{1}_{\{\cdot\}}$ denotes the indicator function. It can be verified that the term in the braces is the derivative
of $\mathbb E(e^{\frac{1}{b}\cdot -|w|})$ w.r.t. $-|w|$, and hence the
above probability can be expressed in terms of the expectation.
\end{proof}
\begin{proof}[Theorem~\ref{thm: RPLap mech}]
\label{thm3.1}
For an R$^2$DP Laplace mechanism and $\forall S\subset \mathbb R$ measurable and dataset $d$ in $\D$, we have 
\begin{eqnarray}
\label{eqn11}
&\Prob(\mathcal M_q(d,b)\in S)\nonumber\\
&=\displaystyle \int_{\mathbb R_{\geq 0}} f(b) \frac{1}{2b} \displaystyle \int_{\mathbb R} \mathds{1}_S\{q(d)+w\}  e^{\frac{-|w|}{b}} \ dw \ db \nonumber \\
&=\displaystyle \int_{\mathbb R_{\geq 0}} g(u)  \frac{u}{2} \int_{\mathbb R} \mathds{1}_S\{q(d)+w\}  e^{-|w|\cdot u} \ dw \ du  \nonumber \\
&= \displaystyle \int_{\mathbb R} \mathds{1}_S\{q(d)+w\} \int_{\mathbb R_{\geq 0}} g(u)  \frac{u}{2} e^{-|w|\cdot u} \ du \ dw \nonumber \\
&= \displaystyle \int_{\mathbb R} \mathds{1}_S\{q(d)+w\} \frac{1}{2} \frac{d M_u(t)}{dt}|_{t=-|w|} \ dw \nonumber \\
&= \frac{1}{2} \displaystyle \int_{S} \frac{d M_u(t)}{dt}|_{t=-|x-q(d)|} dx 
\end{eqnarray}
\begin{eqnarray}
\label{lR$^2$DPlap1}
&\hspace{-0.4cm}=\frac{1}{2} \cdot \Big[-M_u(-|x-q(d)|)|_{S_{\geq q(d)}}+M_u(-|x-q(d)|)|_{S_{< q(d)}}\Big] 
\end{eqnarray}
where $u=b^{-1}$, is reciprocal of random variable $b$ and $g(u)=\frac{1}{u^2}\cdot f(\frac{1}{u})$. Note that $M_u(t)$ is the MGF of random variable $u$ which is identical with $M_{\frac{1}{b}}(t)$.
\end{proof}
\begin{proof}[Theorem~\ref{simple DP}]
To prove this theorem, we first need to give two lemmas on the properties of R$^2$DP Laplace mechanism and MGFs. 
\begin{lem}\label{thm: RPLap mechDP}
 The R$^2$DP mechanism $\mathcal M_q(d,b)$, is 
\begin{equation}
\label{eqeps0}
   \ln \left[\max \limits_{\forall x\in\mathbb R} \left\{ \cfrac{\diff{M_{\frac{1}{b}}(t)}{t}|_{t=-|x-q(d)|}}{\diff{M_{\frac{1}{b}}(t)}{t}|_{t=-|x-q(d')|}} \right\}\right]\text{-differentially private.} 
\end{equation}
\end{lem}
\begin{proof}
\label{proofthm: RPLap mechDP}
According to Equation~\ref{eqn11}, 
\begin{eqnarray*}
\label{eq13}
&\Prob(\mathcal M_q(d,b) \in S)= \frac{1}{2} \displaystyle \int_{S} \diff{M_{\frac{1}{b}}(t)}{t}|_{t=-|x-q(d)|} dx \\
&= \frac{1}{2} \displaystyle \int_{S} \cfrac{\diff{M_{\frac{1}{b}}(t)}{t}|_{t=-|x-q(d)|}}{\diff{M_{\frac{1}{b}}(t)}{t}|_{t=-|x-q(d')|}} \cdot \diff{M_{\frac{1}{b}}(t)}{t}|_{t=-|x-q(d')|} dx
\nonumber
\end{eqnarray*}
Denote by 
\begin{eqnarray*}
  & e^\epsilon =\sup{\left\{ \cfrac{\diff{M_{\frac{1}{b}}(t)}{t}|_{t=-|x-q(d)|}}{\diff{M_{\frac{1}{b}}(t)}{t}|_{t=-|x-q(d')|}} ,\forall x\in S \right\}}, \\
   &\Rightarrow \Prob(\mathcal M_q(d,b) \in S) \leq e^\epsilon  \cdot \Prob(\mathcal M_q(d',b) \in S)
\end{eqnarray*}
and the choice of $S=\mathbb R$ concludes the proof.
\end{proof}
Next, we show the log-convexity property of the first derivative of moment generating functions.
\begin{lem}
First derivative of a moment generating function defined by $\diff{M(t)}{t}=\mathbb E(z\cdot e^{zt})$ is log-convex. 
\end{lem}
\begin{proof}
 For real- or complex-valued random variables $X$ and $Y$, H\"{o}lder's inequality~\cite{trove.nla.gov.au/people/850680} reads; $\mathbb E(|XY|)\leq (\mathbb E(|X|)^p)^{1/p} \cdot (\mathbb E(|Y|)^q)^{1/q}$
for any $1<p,q<\infty$ with $1/p+1/q=1$. Next, for all $\theta \in (0,1)$ and $0\leq x_1,x_2 <\infty$, define $X=z ^\theta \cdot e^{\theta x_1 z}$, $Y=z ^{1-\theta} \cdot e^{(1-\theta) x_2 z}$ and $p=1/\theta$, $q=1/(1-\theta)$. Therefore, we have
\begin{equation*}
  \mathbb E(z\cdot e^{(\theta x_1 +(1-\theta) x_2 z})\leq \mathbb E(z\cdot e^{x_1 z})^{\theta} \cdot \mathbb E(z\cdot e^{x_2 z})^{1-\theta}  
\end{equation*}
which shows the definition of log-convexity holds for $M'(t)$.
\end{proof}
Back to the original proof, following the DP guarantee in Lemma~\ref{thm: RPLap mechDP}, and using triangle inequality, we have 

\begin{eqnarray*}
&e^{\epsilon}= \max \limits_{\forall x\in\mathbb R} \left\{ \frac{\mathbb E(\epsilon\cdot e^{(-|x-q(d)|\cdot \epsilon)})}{\mathbb E(\epsilon\cdot e^{(-|x-q(d')|\cdot \epsilon)})} \right\}\leq \max \limits_{\forall t\in\mathbb R_{\leq 0}} \left\{\frac{\mathbb E(\epsilon \cdot e^{( t\cdot \epsilon)})}{\mathbb E(\epsilon\cdot e^{((t-\Delta q)\cdot \epsilon)})}\right\}\\
\end{eqnarray*}
Next, we show that $f(t)=\frac{\mathbb E(\epsilon \cdot e^{( t\cdot \epsilon)})}{\mathbb E(\epsilon\cdot e^{((t-\Delta q)\cdot \epsilon)})}$ is non-decreasing w.r.t. $t$. For this purpose, we must show that
\begin{equation*}
   f'(t)=\frac{M''(t)\cdot M'(t-\Delta q) -M'(t)\cdot M''(t-\Delta q)}{M'^2(t-\Delta q)} 
\end{equation*} 
is non-negative. However, this is equivalent to show that $\frac{M''(t)}{M'(t)} \geq \frac{M''(t-\Delta q)}{M'(t-\Delta q)}$ or more generally $\frac{M''(t)}{M'(t)}$ is not-decreasing. However, following the log-convexity of first $M'(t)$, the logarithmic derivative of $M'(t)$ denoted by $\frac{M''(t)}{M'(t)}$ is non-decreasing. Thus, for all $t<0$, $f(t)\leq f(0)$, and evaluating $e^{\epsilon(t)}$ at $ t=0$, concludes our proof.
\end{proof}

\begin{proof}[Theorem~\ref{thm: RPLap mechut}]
\label{necescond}
Following Theorem~\ref{thmuselap}, an $\epsilon$-DP Laplace mechanism is $(\gamma,e^{\frac{-\gamma}{b(\epsilon)}})$-useful for all $\gamma \geq 0$, where $b(\epsilon)=\frac{\Delta q}{\epsilon}$. Therefore, for the usefulness of the baseline Laplace mechanism at $\epsilon=\ln [ \mathbb E_{\frac{1}{b}}(  e^{\epsilon(b)} ) ]$, we have 
\begin{eqnarray*}
 e^{\frac{-\gamma\cdot\ln [ \mathbb E_{\frac{1}{b}} (  e^{\epsilon(b)} ) ]}{\Delta q}}=\big(\mathbb E_{\frac{1}{b}} (  e^{\epsilon(b)} ) \big)^{\frac{-\gamma}{\Delta q}}= \big(\mathbb E_{\frac{1}{b}} (  e^\frac{\Delta q}{b} ) \big)^{\frac{-\gamma}{\Delta q}} \leq  \mathbb E_{\frac{1}{b}} \big(e^{\frac{-\gamma}{b}} \big)
\end{eqnarray*}
where the last inequality relation is verified by Jensen inequality~\cite{jensen1906fonctions} as $g(x)=x^{\frac{-\gamma}{b}}$ is a convex function. Recall the following Jensen inequality: Let $(\Omega ,\mathfrak {F},\operatorname {P})$ be a probability space, $X$ an integrable real-valued random variable and $g$ a convex function. Then
\begin{equation*}
    g(\mathbb E(X))\leq \mathbb E(g(X))
\end{equation*}
Therefore, 
\begin{eqnarray*}
 1-e^{\frac{-\gamma\cdot\ln [ \mathbb E_{\frac{1}{b}} (  e^{\epsilon(b)} ) ]}{\Delta q}}\geq  1- \mathbb E_{\frac{1}{b}} \big(e^{\frac{-\gamma}{b}} \big) =U(\ln [ \mathbb E_{\frac{1}{b}} (  e^{\epsilon(b)} ) ],\Delta q,\gamma)
\end{eqnarray*}
This completes the proof.
\end{proof}

\begin{proof}[Theorem~\ref{degenerateDP}]
For $\frac{1}{b}\sim f_{\delta,\frac{1}{b_0}}(\frac{1}{b})$, the MGF is given by $M_{\frac{1}{b}}(t)=e^{\frac{t}{b_0}}$. Following Theorem~\ref{thm: RPLap mechDP}, one can write
\begin{eqnarray*}
&e^{\epsilon}= \max\limits_{\forall x \in \mathbb R}\left\{\frac{\frac{1}{b_0} \cdot e^{\frac{-|x-q(d)|}{ b_0}}}{\frac{1}{b_0} \cdot e^{\frac{-|x-q(d')|}{b_0}}}\right\}=\max\limits_{\forall x \in \mathbb R} \left\{e^{\frac{|x-q(d')|-|x-q(d)|}{ b_0}} \right\}\\
& \leq \max\limits_{\forall x \in \mathbb R} \left\{e^{\frac{|q(d)-q(d')|}{ b_0}} \right\}=e^{\frac{\Delta q}{ b_0}}
\end{eqnarray*}
where the last inequality is from triangle inequality.

\end{proof}
\begin{proof}[Theorem~\ref{bernoulidp}]
 The R$^2$DP Laplace mechanism $\mathcal M_q(d,b)$, $\frac{1}{b} \sim f_{B,\frac{1}{b_0},\frac{1}{b_1}}(\frac{1}{b};p)$ returns with probability $p$, a Laplace mechanism with scale parameter $b_1$, and with probability $1-p$ another Laplace mechanism with scale parameter $b_2$. To this end, we are looking for 
\begin{eqnarray*}
&e^{\epsilon}= \max\limits_{\forall x \in \mathbb R}\left\{\frac{\frac{p}{b_0}\cdot e^{\frac{-|x-q(d)|}{b_0}}+\frac{1-p}{b_1}\cdot e^{\frac{-|x-q(d)|}{b_1}}}{\frac{p}{b_0}\cdot e^{\frac{-|x-q(d')|}{b_0}}+\frac{1-p}{b_1}\cdot e^{\frac{-|x-q(d')|}{b_1}}}\right\}
\end{eqnarray*}
Therefore, using triangle inequality, we have 
\begin{eqnarray*}
&e^{\epsilon_1}=\max\limits_{\forall S \in \mathbb R}\left\{\frac{p\cdot e^{\frac{-|x-q(d)|}{b_0}}+(1-p)\cdot e^{\frac{-|x-q(d)|}{b_1}}}{p\cdot e^{\frac{-|x-q(d')|}{b_0}}+(1-p)\cdot e^{\frac{-|x-q(d')|}{b_1}}} \right\}\\
& \leq \max\limits_{\forall x \geq q(d)}\left\{\frac{p\cdot e^{\frac{\Delta q-|x-q(d')|}{b_0}}+(1-p)\cdot e^{\frac{\Delta q+-|x-q(d')|}{b_1}}}{p\cdot e^{\frac{-|x-q(d')|}{b_0}}+(1-p)\cdot e^{\frac{-|x-q(d')|}{b_1}}} \right \}
\end{eqnarray*}
Let us make the substitutions $X=e^{\frac{-|x-q(d')|}{b_0}}$, $a=e^{\frac{\Delta q}{b_0}}$ and k=$\frac{b_0}{b_1} >1$. Hence, we have
 \begin{eqnarray*}
& e^{\epsilon}\leq \max\limits_{\forall X \in (0,1)}\left\{\frac{p\cdot a \cdot X+(1-p)\cdot (a\cdot X)^k}{p\cdot X+(1-p)\cdot X^k} \right\}
\end{eqnarray*}
To obtain $e^{\epsilon}$, we need to find all the critical points of $e^{\epsilon_1}(X)=\frac{p\cdot a \cdot X+(1-p)\cdot (a\cdot X)^k}{p\cdot X+(1-p)\cdot X^k}$. However, the critical points of a fractional function are the roots of the numerator of its derivative. Hence, suppose
\begin{equation*}
    \diff{e^{\epsilon}(X)}{X}=\frac{N(X)}{D(X)}
\end{equation*}
then
 \begin{eqnarray*}
& \Rightarrow N(X)= \big(p\cdot a +(1-p)\cdot k \cdot a \cdot (a\cdot X)^{k-1} \big)\\
& \cdot \big( p\cdot X+(1-p)\cdot X^k\big) - \big( p+(1-p)\cdot k \cdot X^{k-1}\big) \\&
\cdot \big(p\cdot a \cdot X+(1-p)\cdot (a\cdot X)^k \big)\\&
=p \cdot (1-p) \cdot (k-1) \cdot (a^{k-1} -1) \cdot X^k
\end{eqnarray*}
 However, all the terms in the last expression are strictly positive. Therefore, the only critical points are $X=0$ and $X=1$ and as the function is strictly increasing,
 \begin{eqnarray*}
& e^{\epsilon}\leq e^{\epsilon}(1)=p\cdot a +(1-p)\cdot (a)^k\\
&=p\cdot e^{\frac{\Delta q}{b_0}}+(1-p)\cdot e^{\frac{\Delta q}{b_1}}
\end{eqnarray*}
which is the bound in the Theorem. 
\end{proof}
\begin{proof}[Theorem~\ref{gamadist}]
\label{proofgama}
For a Gamma distribution with shape parameters $k$ and scale parameters $\theta$, the MGF at point $t$ is given as $(1-\theta\cdot t)^{-k}$. Since $\frac{1}{b} \sim f_{\Gamma}(\frac{1}{b};k,\theta)$, following Theorem~\ref{thm: RPLap mechDP}, one can write
\begin{eqnarray*}
&e^{\epsilon}= \max\limits_{\forall x \in \mathbb R}\left\{\frac{k \cdot \theta  \cdot (1+\theta\cdot |x-q(d)|)^{-k-1}}{k \cdot \theta  \cdot (1+\theta\cdot |x-q(d')|)^{-k-1}}\right\}\\
&\Rightarrow \epsilon= \max\limits_{\forall x \in \mathbb R} \  \left\{ (k+1) \cdot \ln \left[ \frac{(1+\theta\cdot |x-q(d')|)}{(1+\theta\cdot |x-q(d)|)} \right] \right \} \\
\end{eqnarray*}
to find the maximum of the $\ln$ term, denote by $X=1+\theta\cdot |x-q(d)|)$. Moreover, since  $|x-q(d')|\leq |x-q(d)|+ \Delta q$, we have 
\begin{eqnarray*}
\Rightarrow \epsilon \leq \max\limits_{\forall X \geq 1} \left\{\frac{X+\Delta q \cdot \theta}{X}\right\}
\end{eqnarray*}
However, since \begin{eqnarray*}
\forall X \geq 1, \ \frac{X+\Delta q \cdot \theta}{X}
\end{eqnarray*} is strictly decreasing, we have 
\begin{eqnarray*}
&\Rightarrow \epsilon=(k+1) \cdot \ln \big[ 1+\theta\cdot \Delta q \big]
\end{eqnarray*}
This completes the proof.
\end{proof}
\begin{proof}[Lemma~\ref{necsgama}]
\label{lemgama}
    We need to show that there exist $k$ and $\theta$ such that $(k+1)\cdot \ln (1+\Delta q \cdot \theta) < -k \cdot \ln (1-\Delta q \cdot \theta)$ , $\theta< \frac{1}{\Delta q}$. Given $\theta = \frac{1}{2 \Delta q}$, we need to show that $\exists k, k \cdot \ln (2)> (k+1) \cdot \ln (1.5)$, 
    which always holds for all $k>1.4094$.
\end{proof}

\begin{proof}[Lemma~\ref{necstrunc}]
\label{label:trunclem}
Using exhaustive search, suppose $\mu = 0.5223$,$ \sigma=1.5454$, $a= 0.5223$ and for $\epsilon=1.1703$ and $\Delta q=0.6$, we will get $\ln(M_{\mathcal{N}^T}(\Delta q))=1.2417$.
\end{proof}

\section{Lagrange Multiplier Function}
\label{sec:lag}

The Lagrange Multiplier Function (all possible linear combinations of the Gamma, uniform and truncated Gaussian distributions) is:

\begin{eqnarray}
\label{objboth}
 &\hspace{-2cm}\mathcal{L} (a_1,a_2,a_3,k, \theta,a_u,b_u,\mu,\sigma,a_{\mathcal N^T}, b_{\mathcal N^T}, \Lambda) \\    &\hspace{-2cm}= M_{\Gamma(k,\theta)}(-a_1\gamma) \cdot M_{U(a_u,b_u)}(-a_2\gamma) \nonumber\\
 &\cdot M_{\mathcal{N}^T(\mu,\sigma,a_{\mathcal N^T}, b_{\mathcal N^T})}(-a_3\gamma) +\Lambda \cdot (\ln \Bigg[\cfrac{\mathsf{N}}{\mathsf{D}}\Bigg]-\epsilon)   \nonumber
\end{eqnarray}
where the numerator and the denominator $\mathsf{N, \ D}$ are
\begin{eqnarray*}
&\hspace{-7.5cm}\mathsf{N}=\\&(a_1 \cdot k \cdot \theta)+ (a_2 \cdot \frac{a+b}{2})+(a_3 \cdot (\mu+(\cfrac{\sigma\cdot\phi(\alpha)-\phi(\beta))}{(\Phi(\beta)-\Phi(\alpha))}))
\end{eqnarray*}
\begin{eqnarray*}
&\hspace{-0.7cm}\mathsf{D}=a_1 \cdot M'_{\Gamma(k,\theta)}(-a_1\cdot \Delta q) \cdot M_{U(a_u,b_u)}(-a_2 \cdot\Delta q) \\ &\cdot M_{\mathcal{N}^T(\mu,\sigma,a_{\mathcal N^T}, b_{\mathcal N^T})}(-a_3\cdot \Delta q)\\& +a_2 \cdot M_{\Gamma(k,\theta)}(-a_1\cdot \Delta q) \cdot M'_{U(a_u,b_u)}(-a_2 \cdot\Delta q) \\ &\cdot M_{\mathcal{N}^T(\mu,\sigma,a_{\mathcal N^T}, b_{\mathcal N^T})}(-a_3\cdot \Delta q)\\& +a_3\cdot M_{\Gamma(k,\theta)}(-a_1\cdot \Delta q) \cdot M_{U(a_u,b_u)}(-a_2 \cdot\Delta q) \\ &\cdot M'_{\mathcal{N}^T(\mu,\sigma,a_{\mathcal N^T}, b_{\mathcal N^T})}(-a_3\cdot \Delta q)
\end{eqnarray*}

 \begin{figure*}[!t]
\includegraphics[width=0.9\linewidth]{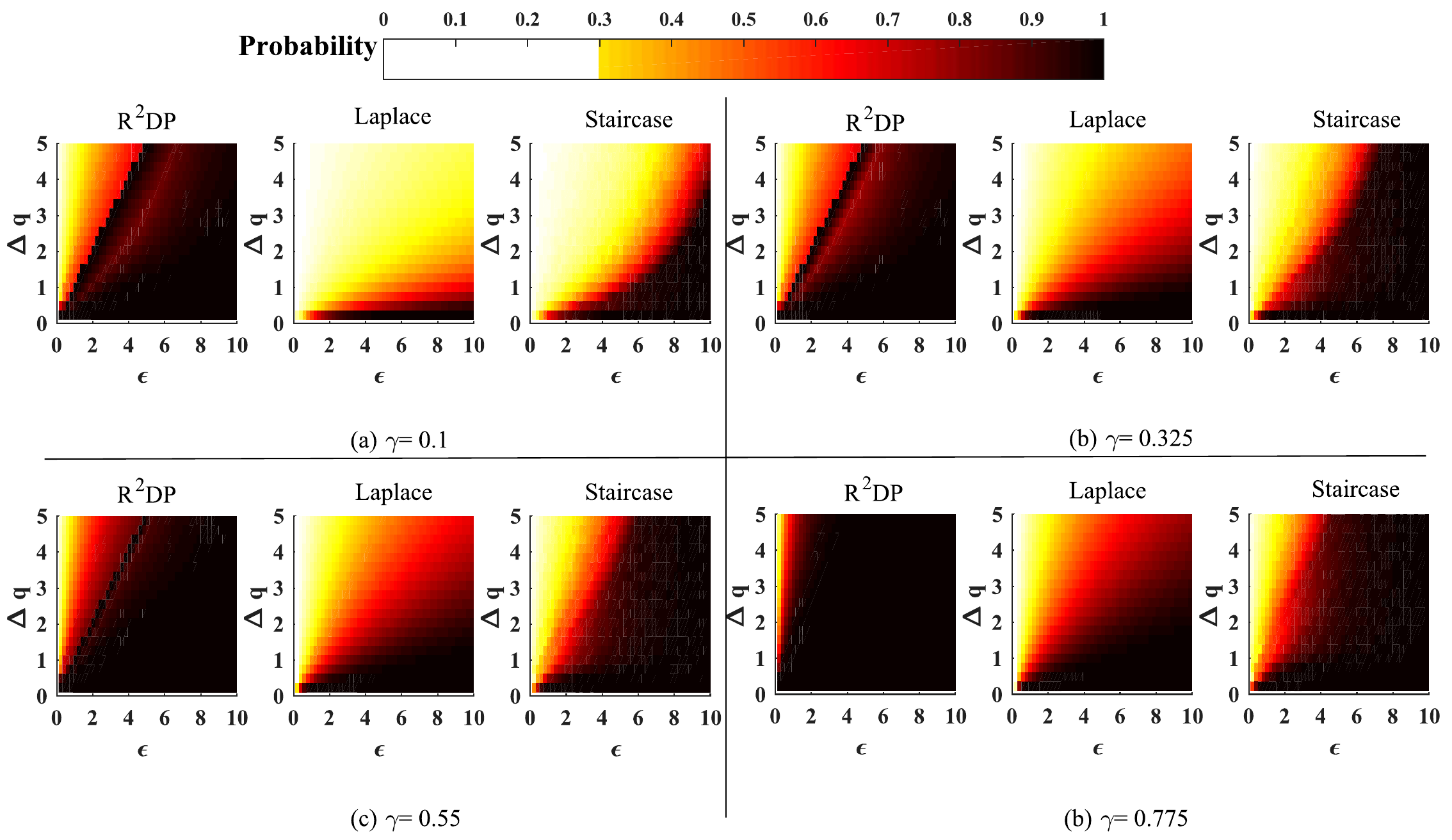}
\centering
\caption{The R$^2$DP mechanism significantly outperforms the competing Laplace and the staircase mechanisms in maximizing the usefulness metric (an example of a utility metric with no known optimal PDF).}
\label{fig:both}
\end{figure*}

\section{Numerical Analysis}
\label{numericsec}
 We also demonstrate the effectiveness of R$^2$DP through numerical results based on Algorithm~\ref{alg:analyst-actions1} (the ensemble R$^2$DP algorithm). In particular, Figure~\ref{fig:both} depicts the corresponding usefulness (the probability of the results to be within a pre-specified error bound) of the R$^2$DP, the Laplace and the Staircase mechanisms. Figure~\ref{fig:both} clearly demonstrates the fact that the R$^2$DP mechanism can significantly improve both already considered to be competing mechanisms. In particular, we observe the power of the R$^2$DP mechanism in generating very high utility results, e.g., results with more than $0.8$ probability fallen inside only $\gamma=0.1$ error-bound, owing to automatically searching a large search space of PDFs.

\section{R\texorpdfstring{$^2$}{2}DP and Other DP Mechanisms}
\label{sec:Gauss}
In this section we briefly discuss the application of the R$^{2}$DP framework in two other well-known baseline DP mechanisms. 
\subsection{R\texorpdfstring{$^2$}{2}DP Exponential Mechanism}
\label{sec:exponential}
The exponential mechanism was designed for situations in which
we wish to choose the “best” response but adding noise directly to the
computed quantity can completely destroy its value, such as setting a
price in an auction, where the goal is to maximize revenue, and adding a
small amount of positive noise to the optimal price (in order to protect
the privacy of a bid) could dramatically reduce the resulting revenue~\cite{10.1561/0400000042}. The exponential mechanism is the natural building block for answering queries with arbitrary utilities (and arbitrary non-numeric
range), while preserving differential privacy. Given some arbitrary
range $\mathcal R$, the exponential mechanism is defined with respect to some
utility function $u : \mathbb{N}^{|\mathcal X|} \times \mathcal R \rightarrow \mathbb R$, which maps database/output pairs
to utility scores. Intuitively, for a fixed database $x$, the user prefers that
the mechanism outputs some element of $\mathcal R$ with the maximum possible
utility score. Note that when we talk about the sensitivity of the utility
score $u : \mathbb{N}^{|\mathcal X|} \times \mathcal R \rightarrow \mathbb R$, we care only about the sensitivity of $u$ with
respect to its database argument; it can be arbitrarily sensitive in its range argument:
\begin{equation*}
   \Delta u \equiv \max\limits_{r\in \mathcal R} \max\limits_{x,y: \norm{x-y} \leq 1} |u(x,r)-u(y,r)|.
\end{equation*}
The intuition behind the exponential mechanism is to output each possible $r\in \mathcal R$ with probability proportional to exp($\epsilon u(x, r)/\Delta u$) and so
the privacy loss is approximately:

\begin{equation} \ln \Big(\cfrac{exp(\epsilon u(x, r)/ \Delta u)}{exp(\epsilon u(y, r)/\Delta u)}\Big)= \epsilon [u(x,r)-u(y,r)/ \Delta u]  \leq \epsilon 
\end{equation}

 The exponential mechanism is a canonical $\epsilon$-DP mechanism, meaning that it describes a class of mechanisms that includes all possible differentially private mechanisms. However, the exponential mechanism can define a complex distribution over a large arbitrary domain, and so it may not be possible to implement the exponential mechanism efficiently when the range of $u$ is super-polynomially large in the natural parameters of the problem~\cite{10.1561/0400000042}. This is the main restrictive aspect of the exponential mechanism against leveraging different accuracy metrics. However, the exponential mechanism can benefit from the additional randomization of privacy budget, to handle the complexity (excessive sharpness) of the defined probability distribution. In particular, as we mentioned earlier, compound (or mixture) distributions arise naturally where a statistical
population contains two or more sub-population which is the case for the exponential mechanism. Thus, we motivate the application of the R$^2$DP framework in designing exponential mechanisms with rather smooth but accurate distributions around each element in the range of $u$. However, further discussion on R$^2$DP exponential mechanism requires formal analysis, e.g., deriving the DP guarantee of such a mechanism.  

\subsection{R\texorpdfstring{$^2$}{2}DP and Differential Privacy Relaxations}
\label{sec:relaxed}

R$^2$DP can also be studied under various relaxations of differential privacy, e.g., $(\epsilon,\delta)$-\DP or R\'enyi Differential Privacy~\cite{mironov2017renyi} which is a privacy notion based on the R\'enyi divergence~\cite{van2014renyi}. These relaxations allow suppressing the long tails of the mechanism's distribution where pure $\epsilon$-differential privacy guarantees may not hold. Instead, they offer asymptotically smaller cumulative loss under
composition and allow greater flexibility in the selection of privacy preserving mechanisms~\cite{mironov2017renyi}. In the following, we briefly discuss the application of R$^2$DP in two of such relaxed notions of the \DP.

\subsubsection{R\texorpdfstring{$^2$}{2}DP Gaussian Mechanism}

A relaxation of $\epsilon$-differential privacy allows an additional bound $\delta$ in its defining
inequality:

\begin{defn}[($\epsilon, \delta$)-differential privacy~\cite{Dwork06_DPgaussian}]\label{def: differential}
A randomized mechanism $M: \D \times \Omega \to \R$ is 
$(\epsilon, \delta)$-differentially private if for all adjacent $d,d' \in \D$, we have
\begin{align}	\label{eq: stand}
\Prob(M(d) \in S) \leq e^{\epsilon} \Prob(M(d') \in S) + \delta, \;\; \forall S \subset \R. 
\end{align}
\end{defn}

This definition quantifies the allowed deviation ($\delta$) for the output distribution of a $\epsilon$-differentially private mechanism, when a single individual is added or removed from a dataset. A differentially private mechanism proposed in~\cite{Dwork06_DPgaussian} modifies an answer to a numerical query by adding the independent and identically distributed zero-mean Gaussian noise.

Given the definition of the $\mathcal Q$-function 
$
\mathcal Q(x) := \frac{1}{\sqrt{2 \pi}} \int_x^{\infty} e^{-\frac{u^2}{2}} du
$, we have the following theorem~\cite{Dwork06_DPgaussian, LeNyDP2012_journalVersion}.

\begin{thm}	\label{thm: Gaussian mech}
Let $q: \D \to \mathbb R$ be a query and $\epsilon>0$.
Then the Laplace mechanism $ \mathcal M_q: \D \times \Omega \to \mathbb R$ 
defined by $\mathcal M_q(d) = q(d) + w$, with $w \sim \mathcal N\left(0,\sigma^2 \right)$, 
where $\sigma \geq \frac{\Delta q}{2 \epsilon}(K + \sqrt{K^2+2\epsilon})$ and $K = \mathcal Q^{-1}(\delta)$,
satisfies $(\epsilon,\delta)$-DP.
\end{thm}

We define $\kappa_{\delta,\epsilon} = \frac{1}{2 \epsilon} (K+\sqrt{K^2+2\epsilon})$, then the standard deviation $\sigma$ in Theorem \ref{thm: Gaussian mech}
can be written as $\sigma(\delta,\epsilon) = \kappa_{\delta,\epsilon} \Delta q$.
It can be shown that $\kappa_{\delta,\epsilon}$ behaves roughly as $O(\ln(1/\delta))^{1/2}/\epsilon$.
For example, to ensure $(\epsilon,\delta)$-differential privacy with $\epsilon = \ln(2)$ and $\delta = 0.05$, 
the standard deviation of the injected Gaussian noise should be about $2.65$ times the $\ell_1$-sensitivity of $q$.

\begin{thm}
\label{gausut}
The Gaussian Mechanism in Theorem~\ref{thm: Gaussian mech} 
is $(\gamma,2\cdot \mathcal Q \big(\frac{ \gamma}{\sigma(\delta,\epsilon)}\big))$-useful.
\end{thm} 
Similar to our R$^2$DP Laplace mechanism, we can formulate an optimization problem for the R$^2$DP model using Gaussian mechanism. Therefore, using Theorems~\ref{thm: Gaussian mech} and \ref{gausut}, we have the following. 
\begin{coro}
\label{co2}
Denote by $u$, the set of parameters for a probability distribution $f_{\sigma}$. Then, the optimal usefulness of an R$^2$DP Gaussian mechanism utilizing $f_{\sigma}$, at each quadruplet $(\epsilon, \delta, \Delta q, \gamma)$ is
\begin{eqnarray}
\small
    \label{gen:gaus}
       &\hspace{-0.5cm}  U_f(\epsilon, \delta, \Delta q, \gamma)=\max\limits_{u\in \mathbb{R}^{|u|}} \big (1- 2\cdot \mathbb{E}_{\sigma}(Q \big(\frac{\gamma}{\sigma(\delta,\epsilon)}\big) \big)) \nonumber\\
& \text{subject to     } \\
& \max \limits_{\forall S \in \R} \left\{\frac{\Prob(\mathcal M_q(d,\sigma)\in S)}{\Prob(\mathcal M_q(d',\sigma)\in S)}\right\}=\epsilon,  \nonumber\\
& \mathbb{E}_{\sigma}(Q(\epsilon \sigma -\frac{1}{2\sigma}))=\delta  \nonumber
    \end{eqnarray}
\end{coro}

\subsubsection{R\texorpdfstring{$^2$}{2}DP and R\'enyi Differential Privacy}
\label{renyiDPsec}

Despite its notable advantages in numerous applications,
the definition of ($\epsilon,\delta$)-differential privacy has the following two limitations. 

First, $(\epsilon, \delta)$-differential privacy was applied to the analysis of the Gaussian mechanism~\cite{Dwork06_DPgaussian}. In contrast
to the Laplace mechanism (whose privacy guarantee is characterized tightly and accurately by $\epsilon$-differential privacy),
a single Gaussian mechanism satisfies a curve of $(\epsilon(\delta), \delta)$-differential privacy definitions~\cite{Dwork06_DPgaussian}. Picking any one point on this
curve may leave out important information about the mechanism's actual behavior \cite{mironov2017renyi}.

Second, $(\epsilon, \delta)$-differential privacy 
also has limitations on the composition of differential privacy \cite{McSherry09}. 
By relaxing the guarantee to $(\epsilon, \delta)$-differential privacy, advanced composition allows tighter analyses for
compositions of (pure) differentially private mechanisms. Iterating
this process, however, quickly leads to a combinatorial explosion of parameters, as each application of an advanced
composition theorem leads to a wide selection of possibilities
for $(\epsilon(\delta), \delta)$-differentially private guarantees. 

To address these shortcomings, R\'enyi differential privacy was proposed as a natural relaxation of differential privacy in~\cite{mironov2017renyi}. 

\begin{defn}[($\alpha, \epsilon$)-RDP]
A randomized mechanism $M: \D \times \Omega \to \R$ is said to have $\epsilon$-R\'enyi differential privacy of order $\alpha$, or ($\alpha$, $\epsilon$)-RDP for short, if for if for all adjacent $d,d' \in \D$, we have 
$D_{\alpha} (M(d)||M(d')) \leq \epsilon$, where $D_{\alpha}(\cdot)$ is the (parameterized) R\'enyi divergence~\cite{van2014renyi}.
\end{defn} 

Compared to $(\epsilon, \delta)$-differential
privacy, R\'enyi differential privacy is a strictly stronger privacy definition. It offers an operationally convenient and quantitatively
accurate way of tracking cumulative privacy loss
throughout execution of a standalone differentially private mechanism and across many such mechanisms~\cite{mironov2017renyi}. Next, we give the R\'enyi differential privacy guarantee of our R$^2$DP mechanism and show that the privacy loss of R$^2$DP under R\'enyi DP can significantly (asymptotically for small $\alpha$) outperform Laplace, Gaussian and Random Response mechanisms. 

\begin{thm}
\label{thmss}
If real-valued query $q$ has sensitivity
$1$, then the R$^2$DP mechanism $\mathcal{M}_q$, leveraging MGF $M$, satisfies 
\[   
     \begin{cases}
       (\alpha, \frac{1}{\alpha-1} \log \left [ \frac{\alpha M(\alpha-1)+(\alpha-1) M(-\alpha)}{2\alpha-1} \right])\text{-RDP}. &\quad if \  \alpha >1\\
      (1, M'(0)+M(-1)-1)\text{-RDP}. &if \  \alpha =1 \\ 
     \end{cases}
\]
\end{thm}
\begin{proof}
The above RDP guarantee follows Corollary 2 in~\cite{mironov2017renyi} on the RDP guarantee of the classic Laplace mechanism. In particular, the above equations are derived using the following substitutions 
$exp(t/b) \rightarrow M(t)$ and $1/b \rightarrow M'(0)$ due to the second-fold randomization of $b$. 
\end{proof}

\section{Other Applications of R\texorpdfstring{$^2$}{2}DP}
\label{sec:RPDP as a Patch to Existing Work1}

R$^2$DP represents a very general concept which could potentially be applied in a broader range
of contexts. In general, applying R$^2$DP to design more application-aware mechanisms may further improve the utility of many existing solutions \cite{NissimRS07}.
We now briefly discuss some of the potential applications as follows.

\vspace{0.05in}

\noindent\textbf{R$^2$DP and Query-Workload Answering} \cite{LiMHMR15}. Given a workload (aka. a batch of queries), the matrix mechanism
generates a different set of queries, called \textit{strategy
queries}, which are answered using a standard Laplace or Gaussian
mechanism. The noisy answers to the workload queries can then be
derived from the noisy answers to the strategy
queries~\cite{li2010optimizing}. This two-stage process can result in a correlated noise distribution that preserves differential privacy
and also increases utility.

Given a triplet ($\epsilon$, query, metric), R$^2$DP can be applied
to replace the Laplace or Gaussian mechanism for answering the
strategy queries of the matrix mechanism. As a result, R$^2$DP will
provide additional improvement in utility (in terms of the TotalError
as defined in \cite{li2010optimizing}) over the improvement already
provided by the matrix mechanism.  More specifically, we compare the
total errors of Laplace and R$^2$DP mechanisms in
Table~\ref{tablematrix} for specific workloads of interest (similar to
those considered in~\cite{li2010optimizing}). These two workloads were
analyzed in \cite{li2010optimizing} using two $n$-sized query
strategies, each of which can be envisioned as a recursive
partitioning of the domain based on the Haar
wavelet~\cite{xiao2010differential}. We denote by $f(\epsilon,\Delta
q)$ the improvement in the TotalError for applying an R$^2$DP noise
instead of a Laplace noise in the matrix mechanism. For instance,
leveraging the results of R$^2$DP (w.r.t. $\ell_1$ or $\ell_2$) shown
in Section~\ref{secl1l2}, for a workload of size $n=6$, at
$\epsilon=2.3$, the improvement for range queries ($\Delta q= 36$) and
predicate queries ($\Delta q= 64$) are $\sim$20\% and $\sim$10\%,
respectively.

\begin{table}[ht]
\caption{Total error of matrix mechanisms comparison (with R$^2$DP vs. Laplace) -- two workloads and two query strategies}
\centering
\begin{adjustbox}{width=0.48\textwidth,center}

\begin{tabular}{|c|c|c|c|}
\hline
\multicolumn{2}{|c|}{TotalError}& \multicolumn{2}{|c|}{Matrix Strategies}
\\
\hline
Mechanisms & Workload Queries &Binary Hierarchy of Sums &  Matrix of the Haar Wavelet
\\
\hline
\multirow{2}{*}{Laplace}& Range Queries &  $\Theta (n^2\log^3(n)/ \epsilon^2) $&  $\Theta (n^2 \log^3(n)/ \epsilon^2 )$
\\
 \cline{2-4} 
& Predicate Queries  &  $\Theta (n 2^n \log^2(n)/ \epsilon^2)$&  $\Theta (n 2^n\log^2(n)/ \epsilon^2)$ \\
\hline

\multirow{2}{*}{R$^2$DP}& Range Queries &  $\Theta (f(\epsilon,n^2) n^2\log^3(n)/ \epsilon^2) $&  $\Theta (f(\epsilon,n^2) n^2 \epsilon^2 \log^3(n))$
\\
 \cline{2-4} 
& Predicate Queries &  $\Theta ( f(\epsilon, 2^n) n 2^n \log^2(n)/ \epsilon^2)$&  $\Theta (f(\epsilon, 2^n) n 2^n\log^2(n)/ \epsilon^2)$ \\
\hline
\end{tabular}
\end{adjustbox}
	\label{tablematrix}
\end{table}

\noindent\textbf{R$^2$DP and Composition}. R$^2$DP may be applied for reducing the privacy leakage due to sequential or parallel querying over a dataset, of which the objective will be to maximize the number of compositions under a specified $\epsilon$-\DP constraint.

\vspace{0.05in}

\noindent\textbf{R$^2$DP and Local Differential Privacy}. 
In this context, R$^2$DP can be regarded as a new randomized response model. In
particular, the randomized response scheme presented in~\cite{Wang2016UsingRR} can be produced using R$^2$DP for the Bernoulli distribution when $b_0 \rightarrow 0$ and $b_1 \rightarrow
\infty$. Therefore, designing more efficient local differential privacy schemes using R$^2$DP is an interesting future direction.

\vspace{0.05in}

\noindent\textbf{R$^2$DP for Continual Observation Applications}. 
Providing differential privacy guarantees on data streams represents another important future
direction for R$^2$DP. As an example, the multi-input multi-output (MIMO) systems process streams of signals originated from many sensors capturing privacy-sensitive events about individuals, and statistics of interest need to be continuously published in real time~\cite{Dwork:2010:DPU:1806689.1806787, LeNyDP2012_journalVersion}, e.g., privacy-preserving traffic monitoring over multi-lane roads~\cite{Brown:2013:HPR:2525314.2525323}. In this context, R$^2$DP can leverage the constraint related to the number of inputs and the number of outputs (e.g., the sensitivity of the output of MIMO filter $G$ with $m$ inputs and $p$ outputs is proportional to the $\mathcal{H}_2$ norm of $G$ which itself is an increasing function of $m$ and $p$
~\cite{7944526}) into its model to build more efficient differentially private mechanisms for the MIMO scenarios.

\end{document}